\documentclass[a4paper,11pt]{article}

\pdfoutput=1
\usepackage[paper=letterpaper,margin=1.25in]{geometry}

\usepackage[hyperfootnotes=false, linktocpage=true, colorlinks, citecolor=blue, linkcolor=blue, urlcolor=blue, filecolor=blue]{hyperref}
\usepackage{graphicx}
\usepackage{subfig} 
\usepackage{amsmath}
\usepackage{amssymb}
\usepackage{caption}
\usepackage{cite}
\usepackage{color}
\usepackage{algpseudocode} 
\usepackage{algorithm} 

\usepackage{amssymb,amsmath,epsfig}
\usepackage{bbm}
\usepackage{color}
\usepackage{cite}
\usepackage{epic}

\usepackage{multirow} 

\numberwithin{equation}{section}
\newlength{\xtrawidth}
\newlength{\xtraheight}
\def\C{\mathbb{C}}
\def\Z{\mathbb{Z}}
\def\R{\mathbb{R}}

\def\P{\mathbb{P}}

\newcommand{\beq}{\begin{equation}}
\newcommand{\eeq}{\end{equation}}

\newcommand{\shortline}{\newline\vskip -7mm{\hbox to 2cm{\hrulefill}}\vskip 3mm}

\newtheorem{theorem}{Theorem}[section]
\newtheorem{lemma}[theorem]{Lemma}

\newenvironment{proof}[1][Proof]{\begin{trivlist}
\item[\hskip \labelsep {\bfseries #1}]}{\end{trivlist}}

\newenvironment{remark}[1][Remark]{\begin{trivlist}
\item[\hskip \labelsep {\bfseries #1}]}{\end{trivlist}}

\newcommand{\qed}{\nobreak \ifvmode \relax \else
      \ifdim\lastskip<1.5em \hskip-\lastskip
      \hskip1.5em plus0em minus0.5em \fi \nobreak
      \vrule height0.75em width0.5em depth0.25em\fi}

\newcommand{\be}{\begin{equation}} 
\newcommand{\ee}{\end{equation}} 
\newcommand{\bi}{\begin{itemize}} 
\newcommand{\ei}{\end{itemize}} 

\newcommand{\mc}[1]{\mathcal{#1}} 
\newcommand{\mb}[1]{\mathbb{#1}} 

\begin{document}
\begin{centering}
\vspace*{1.2cm}
{\Large \bf  Discrete Symmetries of Calabi-Yau Hypersurfaces \\[2mm]
                  in Toric Four-Folds}

\vspace{1cm}

{\bf Andreas P. Braun}\footnote{andreas.braun@physics.ox.ac.uk},
{\bf Andre Lukas}\footnote{lukas@physics.ox.ac.uk},
{\bf Chuang Sun}\footnote{chuang.sun@physics.ox.ac.uk},

{\small
\vspace*{.5cm}
Rudolf Peierls Centre for Theoretical Physics, University of Oxford\\
  1 Keble Road, Oxford OX1 3NP, UK\\[0.3cm]
}

\begin{abstract}\noindent
We analyze freely-acting discrete symmetries of Calabi-Yau three-folds defined as hypersurfaces in ambient toric four-folds. An algorithm which allows the systematic classification of such symmetries which are linearly realised on the toric ambient space is devised. This algorithm is applied to all Calabi-Yau manifolds with $h^{1,1}(X)\leq 3$ obtained by triangulation from the Kreuzer-Skarke list, a list of some $350$ manifolds. All previously known freely-acting symmetries on these manifolds are correctly reproduced and we find five manifolds with freely-acting symmetries. These include a single new example, a manifold with a $\mathbb{Z}_2\times\mathbb{Z}_2$ symmetry where only one of the $\mathbb{Z}_2$ factors was previously known. In addition, a new freely-acting $\mathbb{Z}_2$ symmetry is constructed for a manifold with $h^{1,1}(X)=6$. While our results show that there are more freely-acting symmetries within the Kreuzer-Skarke set than previously known, it appears that such symmetries are relatively rare. 
\end{abstract}
\end{centering}


\newpage

\tableofcontents
\setcounter{footnote}{0}

\section{Introduction}
Discrete symmetries of Calabi-Yau (CY) manifolds, both freely and non-freely acting, are important for a number of reasons. As an immediate application, freely-acting symmetries can be divided out, thereby leading to new CY manifolds with smaller Hodge numbers. Such quotient CY manifolds have been used to fill out the previous fairly sparse tip of the Hodge number plot~\cite{Gabella:2008id} and, due to their relatively small moduli spaces, they are useful for string compactifications. Further, most standard constructions lead to CY manifolds with a trivial fundamental group. Quotients of CY manifolds by freely-acting symmetries, on the other hand, have a non-trivial fundamental group. In the context of string compactifications on CY manifolds, the presence of a non-trivial fundamental group is required whenever gauge field Wilson lines need to be included. For CY compactifications of the heterotic string in particular the inclusion of Wilson lines appears to be the only viable way~\cite{Anderson:2014hia} to arrive at phenomenologically promising models. In fact,  the limited knowledge about such freely-acting discrete symmetries is one of the current ``bottlenecks" in the attempt to systematically construct heterotic CY vacua~\cite{Anderson:2011ns,Anderson:2012yf}. \\
Freely (and non-freely) acting discrete symmetries of CY manifolds can be important for string compactifications in yet another way. Provided they are not divided out, such symmetries can translate into discrete symmetries of the resulting lower-dimensional theory. Such discrete symmetries can carry important information about the structure of the effective theory, for example forbid certain unwanted operators, and they have been extensively used in particle phenomenology. The possible origin of such ``phenomenological" discrete symmetries within string theory is an important question which provides a further motivation for the present paper.\\

Let us summarise the systematic knowledge on freely-acting symmetries of CY three-folds to date. The oldest and conceptually simplest set of CY three-folds consists of complete intersection CY three-folds in products of projective spaces (CICYs), some $7890$ manifolds which have been classified in Ref.~\cite{Candelas:1987kf}. Considerable progress has been made over the past years in finding freely-acting symmetries of CICYs~\cite{Candelas:2008wb,Candelas:2010ve}, culminating in Ref.~\cite{Braun:2010vc} which provides a classification of freely-acting symmetries with a linear realisation on the projective ambient space for the entire CICY dataset. It was found that only $195$ of the $7890$ CICY manifolds have freely-acting symmetries of this kind, although many of these $195$ manifolds allow for multiple symmetries. These results have been used to construct new CICY quotients with small Hodge numbers~\cite{Candelas:2008wb,Constantin:2016xlj} and to systematically search for physically viable heterotic models on CICYs~\cite{Anderson:2011ns,Anderson:2012yf}.  Especially the latter work, which has led to the largest set of heterotic standard models to date, would have been impossible to complete without these results on discrete symmetries for CICYs.\\

Much less is known about discrete symmetries for the largest known set of CY three-folds, hypersurfaces in four-dimensional toric varieties, which were classified by Kreuzer and Skarke in Ref.~\cite{Kreuzer:2000xy}. Somewhat imprecisely but for ease of terminology we will refer to these manifolds as toric CY (TCY) manifolds. The Kreuzer-Skarke list consists of about half a billion reflexive polytopes leading to an even larger number of associated TCY manifolds obtained by triangulation. In Ref.~\cite{Altman:2014bfa}, this process of triangulation has been carried out for all cases with Picard number $h^{1,1}(X)\leq 6$ and we will be relying on this dataset.\\
Batyrev and Kreuzer~\cite{Batyrev:2005jc} have classified all TCY manifolds $X$ with freely-acting symmetries $\Gamma$ such that the quotient $X/\Gamma$ can again be described as a TCY manifold. Amazingly, among the half a billion reflexive polytopes, there are only $16$ freely-acting symmetries of this kind, associated to $16$ of these reflexive polytopes.\\
In general, the quotient manifold associated to a freely-acting symmetry of a TCY manifold does not have to be a TCY manifold itself. For this reason, there is no expectation that these $16$ cases provide a complete list of freely-acting symmetries for TCY manifolds. In fact, this is already clear from the overlap between the two sets of CICY and TCY three-folds. For example, it is well-known that the quintic in $\mathbb{P}^4$, which appears in both data sets, has a freely-acting $\mathbb{Z}_5\times\mathbb{Z}_5$ symmetry. The quintic is also among the $16$ cases identified by Batyrev and Kreuzer, who find a freely-acting $\mathbb{Z}_5$ symmetry. In other words, only a single $\mathbb{Z}_5$ symmetry is ``toric" and appears in the classification of Batyrev and Kreuzer, while the full $\mathbb{Z}_5\times\mathbb{Z}_5$ symmetry is not toric and is, hence, not obtained by their method. Further examples of this kind are provided by the bi-cubic in $\mathbb{P}^2\times\mathbb{P}^2$ (with a toric $\mathbb{Z}_3$ symmetry among the $16$ cases and a non-toric $\mathbb{Z}_3\times\mathbb{Z}_3$ symmetry group) and the tetra-quadric in $\mathbb{P}^1\times \mathbb{P}^1\times\mathbb{P}^1\times\mathbb{P}^1$ (with a toric $\mathbb{Z}_2$ symmetry among the $16$ cases and various larger, non-toric symmetry groups). These examples certainly show that the $16$ spaces identified by Batyrev and Kreuzer can have larger symmetry groups than the toric ones identified in their paper. It is also expected that there are more TCY manifolds with freely-acting (but non-toric) symmetries and the present paper provides an example which shows this is indeed the case. 

The main purpose of this paper is to take first steps towards a classification of discrete symmetries for TCY three-folds. While in itself a fairly technical undertaking, this is highly relevant for the abovementioned tasks, including the construction of CY quotient manifolds with small Hodge numbers, finding phenomenologically promising CY vacua of the heterotic string and understanding the string origin of discrete symmetries in particle physics. Given the size and complexity of the Kreuzer-Skarke dataset, a complete classification is also a formidable task well beyond the present scope of available algorithms and computing power. The required triangulations of the reflexive polytopes in the Kreuzer-Skarke list have only been carried out for $h^{1,1}(X)\leq 6$~\cite{Altman:2014bfa} and this sets an upper limit on what we can currently hope to achieve. In this paper we will, in fact, be more modest and, for our systematic search, restrict to TCY manifolds with $h^{1,1}(X)\leq 3$. This amounts to a class of some $350$ spaces. A further restriction concerns the type of symmetries we take into consideration. For a TCY three-fold $X$ embedded into a toric ambient four-fold $A$, we first determine the group $G_A$ of symmetries of $A$ acting linearly on the homogeneous coordinates of $A$. Then we classify  symmetries groups $\Gamma$ of $X$ which can be embedded into $G_A$, that is, symmetries which descend from linear actions on the ambient space.\\
Broadly, the algorithm for finding such symmetries $\Gamma$ of TCY three-folds $X$, defined as the zero loci of families of polynomials $p$ in a toric four-fold $A$, proceeds as follows.
\begin{itemize}
\item Find the ambient space symmetry group $G_A$ from the toric data for $A$.
\item For a given finite group $\Gamma$, study all group monomorphisms $\Gamma\rightarrow G_A$.
\item For each such group monomorphism, find the family of invariant polynomials $p$.
\item Check if the hypersurfaces defined by such polynomials $p$ are generically smooth.
\item For freely-acting symmetries, check if the symmetry action is fixed point free on $X$.
\end{itemize}
For the case of freely-acting symmetries, the Euler number, $\eta(X)$, of the CY manifolds needs to be divisible by the group order $|\Gamma|$. A given TCY manifold $X$ (with $\eta(X)\neq 0$), therefore, has a finite list of candidate symmetry groups $\Gamma$. There are further indices~\cite{Candelas:1987du},  depending on the embedding $X\subset A$, which need to be divisible by the group order $|\Gamma|$ and which can be used to constrain the possible freely-acting groups $\Gamma$ further (or, in cases where $\eta(X)=0$, provide some non-trivial constraints). The net result is a finite list of group orders and, hence, a finite list of possible freely-acting symmetry groups for a given TCY three-fold, each of which can be analysed using the algorithm outlined above.\\  
For non freely-acting groups there is no simple a priori constraint on the group order. However, we can still apply the above algorithm (dropping the last step of checking fixed points) for the finite number of groups at any given group order, starting with group order two and successively increasing up to a desired maximum. In this way, we can find all symmetries (embedded into $G_A$), freely-acting or not, up to a given maximal group order. In practice, since symmetry groups of CY manifolds tend to be small, this can amount to a complete classification provided the group order can be pushed sufficiently high.\\

By applying this procedure, we find all freely-acting symmetries, linearly realised on the ambient space, of the $\sim 350$ TCY three-folds with $h^{1,1}(X)\leq 3$. We recover all symmetries among these manifolds known from the Batyrev-Kreuzer classification, as well as those known from the overlap with CICY manifolds. We find a single new, non-toric symmetry group, $\mathbb{Z}_2\times\mathbb{Z}_2$, on one of the $16$ Batyrev-Kreuzer manifolds, where one of the $\mathbb{Z}_2$ sub-groups is the toric symmetry identified by Batyrev and Kreuzer. In total, there are five TCY three-folds with $h^{1,1}(X)\leq 3$ and with freely-acting symmetries. Our results are summarised in Section~\ref{scan}.\\
While these systematic results show that there are new freely-acting symmetries, beyond those identified by Batyrev-Kreuzer and those known from CICY manifolds, the search has not led to new {\it manifolds} with freely-acting symmetries. To show that such manifolds exist, we also use our algorithm to construct a freely-acting $\mathbb{Z}_2$ symmetry for a TCY manifold with $h^{1,1}(X)=6$ which is neither among the 16 manifolds of Batyrev-Kreuzer nor a CICY manifold.\\

The outline of the paper is as follows. In the next section, we describe the mathematical background for our algorithm at a relatively informal level, with technical details and proofs relegated to the Appendices. Section~\ref{sec3} sets out the classification algorithm which is applied to two specific examples in Section~\ref{eg}. Our results from the systematic search of all TCY manifolds with $h^{1,1}(X)\leq 3$ are given in Section~\ref{scan} and we conclude in Section~\ref{concl}.

\section{Construction of Discrete Symmetries} \label{sec2}

In this section we review some background material on toric geometry and describe our algorithm for classifying freely acting symmetries of TCY manifolds. Our presentation will be informal and we will focus on the main ideas and results relevant for the classification algorithm. Technical details and proofs can be found in the Appendices.

\subsection{Toric Calabi-Yau hypersurfaces}
\label{conv}

The central objects of interest are Calabi-Yau three-folds defined as hypersurfaces in toric four-folds, TCY manifolds for short. This section provides a short review of TCY manifolds, mainly to set the scene and introduce the required notation. For a more thorough discussion the reader may, for example, consult Refs.~\cite{Batyrev:1994hm,Kreuzer:2006ax,Altman:2014bfa}. We begin with a broad outline of the construction.

Following Ref.~\cite{Batyrev:1994hm}, TCY manifolds (along with their mirrors) can be constructed elegantly from a pair of reflexive polytopes $\Delta$ and $\Delta^\circ$.
Underlying the construction are two four-dimensional lattices $M\cong\mathbb{Z}^4$ and $N\cong\mathbb{Z}^4$ with pairing $\langle \cdot\,, \cdot \rangle$, such that $\Delta \in M_{\mathbb{R}}$ and $\Delta^\circ \in N_{\mathbb{R}}$ satisfy
\begin{equation}
 \langle \Delta, \Delta^\circ \rangle \geq -1 \, . \label{polydual}
\end{equation}
A fan $\Sigma$ can be associated to the polytope $\Delta^\circ$ in the following way. First note that reflexivity of $\Delta^\circ$ means that the origin of $N$ is the unique interior lattice point of $\Delta^\circ$. All other lattice points ${\bf v}_i$ of $\Delta^\circ$ are primitive generators of the rays of the fan $\Sigma$. The cones of $\Sigma$ are given by a triangulation of $\Delta^\circ$, that is by forming subsets of the ${\bf v}_i$, each of which contains the generators of a cone. We will only consider triangulations which are star, regular and fine\footnote{When constructing a triangulation, one may ignore all lattice points on $\Delta^\circ$ in the interior of three-dimensional faces as such an omission has no effect on a Calabi-Yau hypersurface.}. 

The fan $\Sigma$ gives rise to a toric four-fold $A=\P_\Sigma$ by a construction which we will review below. Within this toric four-fold, TCY three-folds are defined as zero loci of polynomials $p$, whose constituent monomials are encoded by $\Delta$, the Newton polytope of the hypersurface equation, as explained in more detail below. The so-defined TCY hypersurfaces are denoted by $X$. Although the ambient space $\P_\Sigma$ is not necessarily smooth, for CY three-folds every fine triangulation leads to a smooth hypersurface $X$~\cite{Batyrev:1994hm}. 

It can be computationally expensive procedure to find all triangulations of a polytope $\Delta^\circ$ which give rise to a smooth Calabi-Yau hypersurface. However, not all of the triangulation data is crucial for the geometry of $X$. In particular, we may ignore everything which happens inside of faces of codimension one of $\Delta^\circ$. Furthermore, it is sufficient to specify a triangulation of $\Delta^\circ$ in terms of its cones of maximal dimension, which are denoted by ${\rm tr}(\Delta^\circ)$. The task of constructing all triangulations for sufficiently simple polytopes, that is, those with few lattice points, has been completed in Ref.~\cite{Altman:2014bfa} and this work represents the natural input data for our algorithm.\\

Let us now be more specific about the construction. The data in Ref.~\cite{Altman:2014bfa} provides a list ${\bf v}_i=(v_{ir})_{r=1,\ldots ,4}\in \Delta^\circ\cap N$, where $i=1,\ldots ,n$, of the lattice points of the polytope $\Delta^\circ$, along with a triangulation ${\rm tr}(\Delta^\circ)$, that is a list of cones which constitute the fan $\Sigma\subset N_{\mathbb{R}}$.

There are several ways to construct the toric ambient space $A=\P_\Sigma$ from the data encoded in the fan $\Sigma$. Here we will focus on the ``global construction" in terms of homogeneous coordinates which is best suited for our discussion of symmetries. In this approach, $\P_\Sigma$ is found as a quotient of an appropriate ``upstairs'' space by a group ${\cal G}$. The starting point for this construction is the space $\mathbb{C}^n$ with coordinates $x_i$ which are in one-to-one correspondence with the ray generators ${\bf v}_i$ of $\Sigma$. From this space we need to remove the exceptional set $Z(\Sigma)$ which is obtained as follows. Consider index sets $I=\{i_1,\ldots ,i_p\}\subset\{1,\ldots ,n\}$ for which the corresponding generators $\{{\bf v}_i\, |\, i \in I\}$ do not share a common cone in $\Sigma$. For any such $I$, this property is preserved if we enlarge $I$ by adding further indices, so that there is a minimal choice of such index sets generating all of them. This choice is related to the generators of the Stanley-Reisner-ideal and we denote the set of such minimal index sets by ${\cal I}$. With the corresponding zero loci
\begin{equation}
 Z(I)=\{{\bf x}\in\mathbb{C}^n\,|\,x_i=0\;\;\forall\, i\in I\}\subset\mathbb{C}^n\; ,\qquad Z(\Sigma)=\bigcup_{I\in{\cal I}}Z(I)\; , \label{Zdef}
\end{equation} 
we can construct the ``upstairs" space of the quotient we want to describe by removing the total zero locus $Z(\Sigma)$ from $\mathbb{C}^n$, that is, by
\begin{equation}
 B=\mathbb{C}^n-Z(\Sigma)\; . \label{Bdef}
\end{equation} 

To obtain the toric variety $A$, this upstairs space needs to be divided by the toric group ${\cal G}\subset({\mathbb{C}^*})^n$ whose elements ${\bf t}=(t_1,\ldots ,t_n)\in  (\mathbb{C}^*)^n$ act on $\mathbb{C}^n$ multiplicatively as $x_i\rightarrow t_i x_i$. The group ${\cal G}$ consists of all ${\bf t}\in (\mathbb{C}^*)^n$ which satisfy
$$
 \prod_{i=1}^nt_i^{v_{ir}}=1\;\;\mbox{ for all }\; r=1,\ldots ,4\; , 
\label{Gdef}
$$
and, in general, has a continuous part isomorphic to $(\mathbb{C}^*)^{n-4}$ and a finite part. Note that ${\cal G}$ only depends on the lattice vectors ${\bf v}_i$ but not on how they are grouped into cones. Consequently, ${\cal G}$ is the same for any triangulation of a given polytope. The continuous part of ${\cal G}$ can be explicitly obtained from the linear relations 
\begin{equation}
 \sum_{i=1}^nq_{ri}\,{\bf v}_i=0 \label{Qdef}
\end{equation}
between the generators ${\bf v}_i$. The coefficients $q_{ri}$ in these relations form a $(n-4)\times n$ matrix $Q$, the so-called charge matrix, whose columns we denote by ${\bf q}_i$. The continuous part of the toric group can then be written as
\begin{equation}
{\cal G}_{\rm cont}=\{({\bf s}^{{\bf q}_1},\ldots ,{\bf s}^{{\bf q}_n})\,|\,{\bf s}\in 
({\mathbb{C}^*})^{n-4}\}\cong (\mathbb{C}^*)^{n-4}\; ,
\end{equation}
where ${\bf s}=(s_1,\ldots ,s_{n-4})$ are the group parameters and we have used the short-hand notation ${\bf s}^{{\bf q}_i}=\prod_{r=1}^{n-4}s_r^{q_{ri}}$.

The toric variety $A = \P_\Sigma$ associated to the fan $\Sigma$ can now be constructed as the quotient
\begin{equation}\label{eq:toricasqoutient}
 A=\frac{B}{{\cal G}}=\frac{\mathbb{C}^n-Z(\Sigma)}{{\cal G}}\; .
\end{equation}

With the ambient space in hand, the family of polynomials defining the TCY manifolds $X\subset A$ is finally given by
\begin{equation}\label{eq:cydefequn}
 p = \sum_{{\bf m} \in \Delta \cap M} c_{\bf m} \prod_{{\bf v}_i \in \Delta^\circ  \cap N \setminus \{0\}} x_i^{\langle {\bf m},{\bf v}_i\rangle + 1 } \, .
\end{equation}
for a choice of complex constants $c_{\bf m}$. The sum in this expression runs over all lattice point of the Newton polytope $\Delta$ which can be obtained from $\Delta^\circ$ using Eq.~\eqref{polydual}. The Hodge numbers $h^{1,1}(X)$ and $h^{2,1}(X)$ of $X$ can also be computed from the data encoded by the polytopes $\Delta$ and $\Delta^\circ$ using the well-known results from Ref.~\cite{Batyrev:1994hm}.

The following is an example of a global TCY construction, following the route outlined above. This example will also be useful in order to illustrate our procedure for classifying symmetries and we will repeatedly come back to it as we go along.\\[3mm]
\underline{\bf Example:} Let us consider a polytope $\Delta^\circ$ with five vertices ${\bf v}_1,\ldots, {\bf v}_5$ given by the columns of the matrix
\begin{equation}
({\bf v}_1,\ldots ,{\bf v}_5)=
\left(\begin{array}{rrrrr}
-1 & -1 &  1  & -1  & -1 \\
 0 &  0 & -1  &  0  &  4 \\
 0 &  0 &  0  &  2  & -2 \\
 0 &  1 &  0  &  0  & -1
\end{array}\right)\, .
\end{equation}
Furthermore, this polytope also contains the integral points ${\bf v}_6=(-1, 2, -1, 0)$ and ${\bf v}_7 = (-1, 0, 1, 0)$, which are interior to one-dimensional faces, the origin $(0,0,0,0)$, which is the unique interior point, as well as the integral point ${\bf v}_8 = (-1, 1, 0, 0)$, which is interior to a three-dimensional face and can, hence, be neglected. The triangulation is unique and reads explicitly
\begin{eqnarray}
 {\rm tr}(\Delta^\circ)&=&\left\{ \left[{\bf v}_5,{\bf v}_6,{\bf v}_1,{\bf v}_7 \right], \left[{\bf v}_5,{\bf v}_6,{\bf v}_1,{\bf v}_3\right], \left[{\bf v}_5,{\bf v}_6,{\bf v}_7,{\bf v}_4\right],
\left[{\bf v}_5,{\bf v}_6,{\bf v}_4,{\bf v}_3\right], \right.\nonumber\\
&&\;\,\left[{\bf v}_5,{\bf v}_1,{\bf v}_7,{\bf v}_3\right], \left[{\bf v}_5,{\bf v}_7,{\bf v}_4,{\bf v}_3\right],\left[{\bf v}_6,{\bf v}_1,{\bf v}_7,{\bf v}_2\right], \left[{\bf v}_6,{\bf v}_1,{\bf v}_2,{\bf v}_3\right], \nonumber\\
&&\;\,\left.\left[{\bf v}_6,{\bf v}_7,{\bf v}_4,{\bf v}_2\right],\left[{\bf v}_6,{\bf v}_4,{\bf v}_2,{\bf v}_3\right], \left[{\bf v}_1,{\bf v}_7,{\bf v}_2,{\bf v}_3\right], \left[{\bf v}_7,{\bf v}_4,{\bf v}_2,{\bf v}_3\right] \right\}
\end{eqnarray}
It follows that the zero set associated to this triangulation is described by the index set
\begin{equation}
 {\cal I} = \{\{1,4\},\{2,5\},\{3,6,7\}\} \label{Iex}
\end{equation}
and is explicitly given by
\begin{eqnarray}
Z(\Sigma) &=& Z(\{1,4\})\cup Z(\{2,5\})\cup  Z(\{3,6,7\})\nonumber\\
&=& \{x_1=x_4=0\}  \cup\{x_5=x_2=0\} \cup   \{x_6=x_7=x_3=0\}\; .
\end{eqnarray}
From the above vertices and Eq.~\eqref{Qdef} the charge matrix (with the columns ordered as ${\bf v}_1,{\bf v}_2,\ldots ,{\bf v}_7$) can be determined as
\begin{equation}\label{Qex}
Q= \left(
\begin{array}{rrrrrrr}
 1 &  0 & 0 & 1 &  0 & 0  & -2 \\ 
 0 &  1 & 0 & 0  & 1 & -2 & 0 \\ 
 0 &  0 & 2 & 0  & 0 & 1  & 1    
\end{array}  
\right) \, .
\end{equation}
The polar dual polytope $\Delta$ can be obtained from Eq.~\eqref{polydual} and it has the five vertices given by the columns of the matrix
\begin{equation}
 \left(\begin{array}{rrrrr}
-7 & 1 & 1 & 1 & 1 \\
-6 & 0 & 2 & 2 & 2 \\
-4 & 0 & 0 & 0 & 4 \\
-8 & 0 & 0 & 8 & 0
\end{array}\right) \, .
\end{equation}
The defining polynomial, $p$, of the TCY manifold can then be computed from Eq.~\eqref{eq:cydefequn} by summing over all the lattice points of $\Delta$. For example, the five monomials in $p$ which arise from the above vertices of $\Delta$ are explicitly given by
\begin{equation}
x_1^8 x_7^4,\hspace{.5cm} x_3^2 , \hspace{.5cm}x_5^8x_6^4 , \hspace{.5cm}x_2^8 x_6^4, \hspace{.5cm}x_4^8 x_7^4\; . \label{vertmon}
\end{equation}
However, the most general polynomial $p$ has many more terms from the other lattice points of $\Delta$. From the standard formulae, the Hodge numbers of the so-defined TCY manifold $X=\{p=0\}\subset A=\mathbb{P}_\Sigma$ are given by $h^{1,1}(X)=3$ and $h^{2,1}(X)=83$.
\shortline

\subsection{Constraints on the order of symmetry groups}\label{sect:possiblegroups}
Freely-acting symmetries $\Gamma$ of a CY manifold $X$ are severely constrained since certain topological indices need to be divisible by the group order $|\Gamma|$. The simplest example for such an index is the Euler number, $\eta(X)$, of the manifold. The Euler numbers of $X$ and its quotient $X/\Gamma$ by a freely-acting symmetry $\Gamma$ are related by $\eta(X/\Gamma)=\eta(X)/|\Gamma|$ and this shows that $\eta(X)$ must indeed be divisible by the group order $|\Gamma|$. For manifolds with $\eta(X)\neq 0$ this already implies a finite number of possible group orders, and, hence, a finite number of possible freely-acting groups. 

When the manifold $X$ is embedded into an ambient space $A$ with normal bundle $N$, as is the case for TCY manifolds, there are further indices which need to be divisible by $|\Gamma|$. From Ref.~\cite{Candelas:1987du}, one class of such indices is given by the indices of the bundles $N^k\otimes TX^l$, where $k,l\in\mathbb{Z}^{\geq 0}$. Explicitly, they are given by
\begin{equation}
 \chi_{k,l}:={\rm ind}(N^k\otimes TX^l)=\int_X{\rm td}(TX)\wedge{\rm ch}(N^k\otimes TX^l) \label{indkl}
\end{equation}
where ${\rm td}(TX)=1+\frac{1}{12}{\rm c}_2(TX)$ is the Todd class of the CY manifold $X$. These quantities need to be divisible by $|\Gamma|$ for all $(k,l)$ but, as has been shown in Ref.~\cite{Candelas:1987du}, this is the case if and only if $|\Gamma|$ is a divisor of the four indices $\chi_{k,l}$ with $(k,l)\in\{(0,1),(1,0),(2,0),(3,0)\}$. 

Another relevant set of indices~\cite{Candelas:1987du} are the modified signatures of the bundles $N^k\otimes TX^l$, defined by
\begin{equation}
 \sigma_{k,l}:=\sigma(N^k\otimes TX^l)=\int_XL(TX)\wedge\widetilde{\rm ch}(N^k\otimes TX^l)\; , \label{sigkl}
\end{equation} 
where $L(TX)=1-\frac{1}{3}{\rm c}_2(TX)$ is the L-genus of the CY manifold and $\widetilde{\rm ch}$ is the Chern character with the curvature $R$ replaced by $2R$. It has been argued in Ref.~\cite{Candelas:1987du} that the only new condition from these signatures arises from $\sigma_{1,1}$. 

In summary, in order to strengthen the constraint on $|\Gamma|$ from the Euler number (or, in the case $\eta(X)=0$, to obtain constraints in the first place), we should consider divisibility by $|\Gamma|$ of the bundle indices $\chi_{0,1}$, $\chi_{1,0}$, $\chi_{2,0}$ and $\chi_{3,0}$ in Eq.~\eqref{indkl}  and divisibility of the signature $\sigma_{1,1}$ in Eq.~\eqref{sigkl}. This leads to a finite list of possible group orders $|\Gamma|$ and, therefore, to a finite list of possible freely-acting groups $\Gamma$ which can be systematically analysed based on the algorithm described below. For our example this works as follows.\\[3mm]
\underline{\bf Example:} Let us work out the above indices for our standard example. For the Euler number it follows that $\eta(X)=2(h^{1,1}(X)-h^{2,1}(X))=-160$. Using the explicit formulae in Ref.~\cite{Candelas:1987du} we find for the other indices
\begin{equation}
 \chi_{0,1} = -80 \;,\quad \chi_{1,0}=104\;\quad\chi_{2,0}=720\;,\quad \chi_{3,0}=2360\;,\quad \sigma_{1,1} = -1280\; .
\end{equation} 
The greatest common divisor of these indices is $8$ and, hence, the possible orders of freely-acting groups for the manifold are $|\Gamma|\in\{2,4,8\}$.
\shortline
As mentioned before, we are not aware of similar constraints on the order of non-freely acting symmetries. In this case, the best we can do is to apply our algorithm to all finite groups up to a given maximal order.

\subsection{Symmetry group of the ambient space}
\label{symzeroset}

With a description of the ambient toric space in hand, we are now in a position to find the subgroup, $G_A$, of the ambient space automorphism group which is linearly realized on the homogeneous coordinates $x_i$. For simplicity, we will simply refer to $G_A$ as the toric symmetry group. More details on the full automorphism group of a toric variety and the toric symmetry group, including proofs for the statements in this section, can be found in Appendix \ref{theoryfree}.

As mentioned above, we will only consider linear transformations on the homogeneous coordinates $x_i$, that is, elements of ${\rm Gl}(n,\mathbb{C})$. Our first task is to find the sub-group $G_B\subset {\rm Gl}(n,\mathbb{C})$ of linear automorphisms of the ``upstairs'' space $B$, defined in Eq.~\eqref{Bdef}. Then, we will study how such actions descend to the quotient $A$ \eqref{eq:toricasqoutient}, which we will refer to as the ``downstairs'' space. 

We recall that the structure of $B$ is encoded in the set ${\cal I}$ which contains index sets $I\subset\{1,\ldots ,n\}$, one for each zero locus $Z(I)$ which is removed from $\mathbb{C}^n$ in order to obtain $B$. While the index sets $I\in{\cal I}$ can overlap, that is corresponding zero sets $Z(I)$ can have a non-trivial intersection, it is straightforward to define a refinement of ${\cal I}$, denoted ${\cal J}$, whose index sets form a partition of $\{1,\ldots ,n\}$. Specifically, we can define an equivalence relation $\sim$ on $\{1,\ldots ,n\}$ by
\begin{equation}
 i\sim j\quad:\;\Longleftrightarrow\quad i,j\mbox{ are contained in the same sets }I\in{\cal I}\; .
\end{equation} 
The refinement ${\cal J}$ then consists of the equivalence classes under this relation.\\[3mm]
\underline{\bf Example:} The index set ${\cal I}$ for our example~\eqref{Iex} is a clean partition of $\{1,\ldots ,7\}$ and for this reason we have
\begin{equation}
 {\cal J}={\cal I} = \{\{1,4\},\{2,5\},\{3,6,7\}\}\; . \label{IJex}
\end{equation} 
The equality of the two sets is a special feature of this particular example. In general, the index sets ${\cal I}$ and ${\cal J}$ can be different.
\shortline
There are a number of obvious sub-groups of $G_B$ which can be written down immediately. These include the discrete permutation groups
\begin{equation}
 S=\{\sigma\in S_n\,|\,\sigma(I)=I,\;\forall\, I\in{\cal I}\}\subset S_n\; ,\qquad R=\{\sigma\in S_n\,|\,\sigma(I)\in{\cal I},\;\forall I\in{\cal I}\}\subset S_n\; ,
\end{equation}
which, via the obvious embedding $S_n\hookrightarrow {\rm Gl}(n,\mathbb{C})$, can also be thought off as sub-groups of ${\rm Gl}(n,\mathbb{C})$. The first of these groups is the stabiliser group of the zero sets labelled by $I\in{\cal I}$, that is, it leaves each component, $Z(I)$, of the zero set unchanged. The second group, $R$, maps the zero sets into each other. Clearly, both $S$ and $R$ leave the total zero set $Z(\Sigma)$ invariant and are, hence, sub-groups of $G_B$.

It can be shown (see Appendix~\ref{theoryfree}) that
\begin{equation}
S=\bigotimes_{J\in{\cal J}}S(J)\;,\qquad R\cong P\ltimes S\; , \label{SPres}
\end{equation}
where $S(J)$ is the group of permutations on the set $J$, and $P=R/S$. The group $P$ can be worked out explicitly by finding all permutation in $S_n$ which map the index sets $J\in{\cal J}$ into each other and which preserve the ``natural" order of indices within each set $J$. In other words, $P$ should be thought of as the group which permutes zero sets.

Another obvious sub-group of $G_B$ is given by
\begin{equation}
 H_B=\{g\in G_B\,|\, g(Z(I))=Z(I),\;\forall I\in{\cal I}\}\subset G_B\; , 
\end{equation}
that is, the sub-group which leaves the components $Z(I)$ of the zero set invariant individually. It can be shown (see Appendix~\eqref{theoryfree}) that
\begin{equation}
 H_B=\bigotimes_{J\in{\cal J}}{\rm Gl}(J,\mathbb{C})\; , \label{HBres}
\end{equation}
where ${\rm Gl}(J,\mathbb{C})$ denotes the general linear group acting on the coordinates $x_j$, where $j\in J$. It turns out that the full upstairs symmetry group $G_B$ can be expressed in terms of the above groups and is, in fact, given by
\begin{equation}
  G_B=P\ltimes H_B\; . \label{GBres}
\end{equation}  
In summary, the upstairs symmetry group $G_B$ consists of two parts, the group $H_B$ which consists of blocks of general linear groups whose action leaves the components of the zero sets $Z(I)$ unchanged and a permutation part, $P$, which permutes components of the zero set.\\[3mm]
\underline{\bf Example:} With the set ${\cal J}$ for our example in Eq.~\eqref{IJex} and Eq.~\eqref{SPres} it follows that the permutation groups $S$ and $P$ are given by
\begin{eqnarray}
 S&=&S(\{1,4\})\times S(\{2,5\})\times  S(\{3,6,7\})\cong S_2\times S_2\times S_3\label{Sex}\\
 P&=&\{{\rm id},((2,1),(4,5),(3),(6),(7))\}\cong S_2\label{Pex}
 \end{eqnarray}
 where the permutation which generates $P$ (given above in cycle notation) swaps the two blocks $\{2,5\}$ and $\{1,4\}$. From Eq.~\eqref{HBres}, we find for the continues part of the upstairs symmetry group 
\begin{equation}
 H_B={\rm Gl}(\{1,4\},\mathbb{C})\times {\rm Gl}(\{2,5\},\mathbb{C})\times {\rm Gl}(\{3,6,7\},\mathbb{C})\; .
\end{equation} 
In accordance with this structure, let us order the coordinates as $((1,4),(2,5),(3,6,7))$. The elements $h\in H_B$ can be represented as $7\times 7$ matrices which, relative to this ordering of coordinates, have the form
\begin{equation}
 h=\left(\begin{array}{ccc}h_2&0&0\\0&\tilde{h}_2&0\\0&0&h_3\end{array}\right)\; . \label{hmatex}
\end{equation}
Here $h_2$, $\tilde{h}_2$ are general linear $2\times 2$ matrices, acting on the coordinates $(1,4)$ and $(2,5)$, respectively, while $h_3$ is a general linear $3\times 3$ matrix acting on the coordinates $(3,6,7)$. For the same ordering of the coordinates, the matrix version of the generator of $P$ in Eq.~\eqref{Pex} can be written as
\begin{equation}
 p=\left(\begin{array}{ccc}0&\mathbbm{1}_2&0\\\mathbbm{1}_2&0&0\\0&0&\mathbbm{1}_3\end{array}\right)\; . \label{pmatex}
\end{equation} 
The upstairs symmetry group $G_B=P\ltimes H_B$ can then be viewed as all matrices generated by the matrices in Eqs.~\eqref{hmatex} and \eqref{pmatex}.
\shortline
The symmetry group $G_A$ of the toric variety $A$ is the part of the upstairs symmetry groups $G_B$ which descends to the quotient $A=B/{\cal G}$. On general grounds, this part is given by the formula
\begin{equation}
 G_{A}=\frac{N_{G_{B}}({\cal G})}{{\cal G}}\; , \label{GAdef}
\end{equation} 
where $N_{G_{B}}({\cal G})$ denotes the normalizer of ${\cal G}$ within $G_B$. The task is to evaluate this formula for the upstairs symmetry group $G_B$ as given in Eq.~\eqref{GBres} and this is explicitly carried out in Appendix~\ref{theoryfree}. Here, we summarise the result which can be expressed in terms of two further index sets ${\cal K}$ and ${\cal L}$. The index set ${\cal K}$ is defined as containing the sets $K\subset\{1,\ldots ,n\}$ which label coordinates with the same charge ${\bf q}_i$. This means that the toric group ${\cal G}$ can also be written as 
\begin{equation}
  {\cal G}=\bigotimes_{K\in{\cal K}}{\cal G}_{{\bf q}_K}(K)\; . \label{Gres}
\end{equation}
where  ${\cal G}_{{\bf q}_K}(K)$ denotes the group which acts by multiplying the coordinates $x_k$ for $k\in K$ by ${\bf s}^{{\bf q}_K}$. The second index set, ${\cal L}$, is simply the refinement of the index sets ${\cal J}$ and ${\cal K}$, that is, it consists of all non-empty intersections $J\cap K$, where $J\in{\cal J}$ and $K\in{\cal K}$. In general, this will be a new index set, distinct from the ones introduced previously. However, for reflexive polytopes, which is our case of interest, it can been shown (see Appendix~\ref{theoryfree}) that ${\cal L}={\cal K}$.

The downstairs versions of the groups $R$ and $S$ are defined as
\begin{equation}
 R_A=R\cap N_{{\rm Gl}(n,\mathbb{C}^n)}({\cal G})\;,\qquad S_A=S\cap C_{{\rm Gl}(n,\mathbb{C}^n)}({\cal G})\; ,
\end{equation}
where $C_{{\rm Gl}(n,\mathbb{C}^n)}({\cal G})$ is the centralizer of ${\cal G}$ in ${\rm Gl}(n,\mathbb{C}^n)$. As before, we can write $R_A$ as the semi-direct product
\begin{equation}
R_A=P_A\ltimes S_A\;,\qquad P_A=R_A/S_A\; .
\end{equation}
In other words, the group $R_A$ consists of all the elements of $r\in R$ which normalize ${\cal G}$, that is, all $r\in R$ for which the equations
\begin{equation}
 {\bf s}^{{\bf q}_{r(i)}}=\tilde{\bf s}^{{\bf q}_i} \label{PAeqs}
\end{equation} 
have solutions $\tilde{\bf s}\in (\mathbb{C}^*)^{n-4}$ for all ${\bf s}\in (\mathbb{C}^*)^{n-4}$. Then $P_A$ contains permutations in $R_A$ which preserve the order of indices within each block $K\in{\cal K}$ but permute the blocks with one another. Hence, $P_A$ can also be viewed as a sub-group of $S({\cal K})$, the permutation group of the blocks $K\in{\cal K}$. With this notation the symmetry group of the toric space $A$ can be written as
\begin{equation}
 G_A=P_A\ltimes(H_A/{\cal G})\; ,\quad H_A=\bigotimes_{K\in{\cal K}}{\rm Gl}(K,\mathbb{C})\; . \label{GAres} 
\end{equation} 
As usual, the above semi-direct product is defined by the multiplication rule
\begin{equation}
 (p,h) (\tilde{p},\tilde{h}):=(p\tilde{p},\tilde{p}^{-1}h\tilde{p}\tilde{h}) \label{sdprod}
\end{equation}
for $p,\tilde{p}\in P_A$ and $h,\tilde{h}\in H_A/{\cal G}$. Evidently, this symmetry group has the same structure as its upstairs counterpart. There is a continuous part, $H_A/{\cal G}$, with blocks of general linear groups acting on the homogeneous coordinates in a way that is consistent with the structure of $A$, and a discrete part, $P_A$, which permutes these blocks. For our example, following the above steps leads to the following.\\[3mm]
\underline{\bf Example:} The charge matrix~\eqref{Qex} for our example shows that the index set ${\cal K}$ is given by
\begin{equation}
 {\cal K}=\{\{1,4\},\{2,5\},\{3\},\{6\},\{7\}\}\; .
\end{equation} 
This is indeed merely a refinement of the index set ${\cal J}$ in Eq.~\eqref{IJex} so that the intersection of ${\cal J}$ and ${\cal K}$ does indeed not lead to a new index set, in line with our earlier claim. If we write the charge vectors~\eqref{Qex} as
\begin{equation}
 {\bf q}_{\{1,4\}}=\left(\begin{array}{c}1\\0\\0\end{array}\right)\,,\; {\bf q}_{\{2,5\}}=\left(\begin{array}{c}0\\1\\0\end{array}\right)\,,\;
 {\bf q}_{\{3\}}=\left(\begin{array}{c}0\\0\\2\end{array}\right)\,,\;{\bf q}_{\{6\}}=\left(\begin{array}{r}0\\-2\\1\end{array}\right)\,,\;
 {\bf q}_{\{7\}}=\left(\begin{array}{r}-2\\0\\1\end{array}\right)\; ,
\end{equation} 
then, from Eq.~\eqref{Gres}, the toric group can be written as
\begin{equation}
 {\cal G}={\cal  G}_{{\bf q}_{\{1,4\}}}(\{1,4\})\times {\cal  G}_{  {\bf q}_{\{2,5\}}}(\{2,5\})\times {\cal  G}_{  {\bf q}_{\{3\}}}(\{3\})\times {\cal  G}_{  {\bf q}_{\{6\}}}(\{6\})\times {\cal  G}_{  {\bf q}_{\{7\}}}(\{7\})\; . \label{Gex}
\end{equation}
In practice, adopting the coordinate ordering $((1,4),(2,5),(3),(6),(7))$ as before, the elements of ${\cal G}$ are given by the $7\times 7$ matrices
\begin{equation}
 g={\rm diag}({\bf s}^{{\bf q}_{\{1,4\}}},{\bf s}^{{\bf q}_{\{1,4\}}},{\bf s}^{{\bf q}_{\{2,5\}}},{\bf s}^{{\bf q}_{\{2,5\}}},{\bf s}^{{\bf q}_{\{3\}}},{\bf s}^{{\bf q}_{\{6\}}},{\bf s}^{{\bf q}_{\{7\}}})\; , \label{gmatex}
\end{equation} 
with ${\bf s}=(s_1,s_2,s_3)$ the three parameters of the toric group. From Eq.~\eqref{GAres}, we also find 
\begin{equation}
 H_A={\rm Gl}(\{1,4\})\times {\rm Gl}(\{2,5\})\times {\rm Gl}(\{3\})\times {\rm Gl}(\{6\})\times {\rm Gl}(\{7\})\; , \label{HAex}
\end{equation} 
with associated matrices $h\in H_A$ of the form
\begin{equation}
h=\left(\begin{array}{ccccc}h_2&0&0&0&0\\0&\tilde{h}_2&0&0&0\\0&0&h_1&0&0\\0&0&0&\tilde{h}_1&0\\0&0&0&0&h_1'\end{array}\right)\; , \label{hAmatex}
\end{equation}
where $h_2$, $\tilde{h}_2$ are general linear $2\times 2$ matrices and $h_1,\tilde{h}_1,h_1'\in\mathbb{C}^*$. The permutation part $P_A$ can be determined by solving the Eqs.~\eqref{PAeqs} for this case, using the charge matrix~\eqref{Qex}. The result is that $P_A$ is generated by the matrix 
\begin{equation}
 \tilde{p}=\left(\begin{array}{ccccc}0&\mathbbm{1}_2&0&0&0\\\mathbbm{1}_2&0&0&0&0\\0&0&1&0&0\\0&0&0&0&1\\0&0&0&1&0\end{array}\right)\; , \label{ptex}
\end{equation} 
which simultaneously permutes the blocks $(1,4)$ with $(2,5)$ and $(6)$ with $(7)$. Hence $P_A\cong S_2$.
\shortline

\subsection{Construction of representations} \label{constrep}
Having determined the symmetry group $G_A$ of the toric ambient space, our next step is to study group monomorphisms $R:\Gamma\rightarrow G_A$ from a finite group $\Gamma$. Given the structure of $G_A$ in Eq.~\eqref{GAres}, for  $\gamma\in\Gamma$, we can always write
\begin{equation}
 R(\gamma)=(\pi(\gamma),r(\gamma))\quad\mbox{where}\quad\pi:\Gamma\rightarrow P_A\; ,\quad r:\Gamma\rightarrow H_A/{\cal G}\; .
\end{equation} 
The multiplication rule~\eqref{sdprod} in $G_A$ means that the maps $\pi$ and $r$ must satisfy
\begin{equation}
\pi(\gamma\tilde{\gamma})=\pi(\gamma)\pi(\tilde{\gamma})\; ,\qquad r(\gamma\tilde{\gamma})=\pi(\tilde{\gamma})^{-1}r(\gamma)\pi(\tilde{\gamma}
)r(\tilde{\gamma})\; ,
\end{equation}
and, hence, $\pi:\Gamma\rightarrow P_A$ is a group homomorphism (indeed a permutation representation given that $P_A$ can be viewed as a sub-group of the permutation group $S({\cal K}$) and $r:\Gamma\rightarrow H_A/{\cal G}$ is a $\pi$-homomorphism. Therefore, we proceed by first generating the permutation representations $\pi$ and then, for each $\pi$, the corresponding $\pi$-representations $r$. 

Permutation representations can be obtained by a standard method available in the Mathematics literature, see for example Ref.~\cite{aschbacher2000finite}, particularly Section 2, for details. Broadly, this method works as follows. For a permutation representation $\pi:\Gamma\rightarrow S({\cal K})$, where $S({\cal K})$ is the permutation group of the blocks $K\in{\cal K}$, the set ${\cal K}$ splits into subsets ${\cal K}_i$ on each of which $\pi$ acts transitively.  For an $x\in{\cal K}_i$ we denote by ${\cal H}_i\subset \Gamma$ the stabiliser sub-group of $x$. Then the coset $\Gamma/{\cal H}_i$ and ${\cal K}_i$ can be identified via the map ${\cal H}_i\gamma\rightarrow \gamma x$ for $\gamma\in \Gamma$. The action of the permutation representation $\pi$ on ${\cal K}_i$ can now be described, via this identification, by a right-multiplication on the cosets in $\Gamma/{\cal H}_i$. Hence, in essence, all that is required to construct the permutation representations $\pi:\Gamma\rightarrow S({\cal K})$ is knowledge of all the sub-groups of $\Gamma$, something that can be worked out by standard group theory methods. From those representations $\pi$ we then have to select the ones whose image is contained in $P_A\subset S({\cal K})$. Note that, while we are asking for the full representation $R$ to be injective, this does not necessarily have to be the case for $\pi$.

The second step is to construct the maps $r:\Gamma\rightarrow H_A/{\cal G}$ for each $\pi$ and there are two technical complications we have to resolve in this context. First, we have to deal with the fact that $r$ is, in general, a $\pi$-homomorphisms rather than a regular group homomorphism and, secondly, that $r$ corresponds not to a linear but a (multi-)projective representation, since its target space is not simply $H_A$ but the quotient $H_A/{\cal G}$. 

To disentangle these two problems, let us first consider the case where $\pi$ is the trivial representation so that
\begin{equation}
 r:\Gamma\rightarrow \frac{H_A}{\cal G}=\frac{\bigoplus_{K \in {\cal K}} {\rm Gl}(K,\C)}{\cal G} \label{rmap}
\end{equation}
is actually a group homomorphism. There is a standard method to classify projective representations $\rho:\Gamma\rightarrow {\rm Gl}(n,\mathbb{C})/\mathbb{C}^*$ of a group $\Gamma$ using the group's Schur cover, $\hat{\Gamma}$. The Schur cover is a certain central extension of the original group and can be computed purely from group theory. Then, all projective representations $\rho$ of $\Gamma$ can be obtained from linear representations $\hat{\rho}:\hat{\Gamma}\rightarrow {\rm Gl}(n,\mathbb{C})$ of the Schur cover by simply projecting to the quotient (see, for example, Ref.~\cite{aschbacher2000finite} or the summary in Appendix~ \ref{sect:schurcoversstandard} for details). This method can be trivially generalised to the multi-projective case, that is, to homomorphisms of the form
\begin{equation}
 \rho:\Gamma\rightarrow  \bigoplus_{K \in {\cal K}} Gl(K,\C)/\C^* \, .
\end{equation} 
Specifically, we can study all linear representations $\hat{\rho}:\hat{\Gamma}\rightarrow  \bigoplus_{K \in {\cal K}} Gl(K,\C)$ of the Schur cover group and then project to the quotient in order to obtain all such multi-projective representations $\rho$. These are not quite yet the representations $r$ we are interested in, since the toric group ${\cal G}$ in Eq.~\eqref{rmap} is usually a genuine sub-group of $(C^*)^{|{\cal K}|}$. However, we can deal with this complication by simply selecting from the representations $\hat{\rho}$ those representations $\hat{r}$ which lead to well-defined homomorphisms upon taking the quotient by ${\cal G}$.\\[3mm]
\underline{\bf Example:} Let us illustrate the discussion so far with or standard example, using the group $\Gamma=\mathbb{Z}_2=\{1,-1\}$ for simplicity. Note that, from our previous index analysis, $|\Gamma|=2$ is indeed a possible order of a freely-acting group for this TCY manifold. 

The permutation part, $P_A$, of the ambient space symmetry group is isomorphic to $S_2$ so there are two possible choices for the representation $\pi$. In line with our simplifying assumption above we will first discuss the case where $\pi$ is trivial. (The other case will be exemplified below.) The Schur cover of $\mathbb{Z}_2$ is $\mathbb{Z}_2$ so we can simply study linear representations $\hat{r}:\mathbb{Z}_2\rightarrow H_A$, for $H_A$ as in Eq.~\eqref{HAex}. There are many possible choices and one illustrative example, for the coordinate ordering $((1,4),(2,5),(3),(6),(7))$, is
\begin{equation}
 \pi(\gamma)=\mathbbm{1}_7\;,\qquad \hat{r}(1)=\mathbbm{1}_7\;,\qquad \hat{r}(-1)={\rm diag}(1,-1,1,-1,-1,-1,1)\; . \label{symm3}
\end{equation} 
Clearly, $\hat{r}(-1)$ chosen in this way normalises the toric group ${\cal G}$ in Eq.~\eqref{Gex} and, hence, this representation descends to a well-defined homomorphism $r:\mathbb{Z}_2\rightarrow H_A/{\cal G}$.
\shortline

We are now ready to describe the algorithm for constructing all $\pi$-representations $r$ in the general case. We start with one of the representations $\pi:\Gamma\rightarrow P_A$ constructed as described above and recall that $P_A$ can be viewed as a subgroup of $S({\cal K)}$, the permutation group of the blocks $K\in{\cal K}$. In general, there are various orbits under the action of $\pi$ but in order to simplify the discussion let us assume that there is only a single orbit. The multi-orbit case is a straightforward generalization. Let us denote the blocks in this single orbit by $K_i\in{\cal K}$, where $i=1,\ldots ,b$, and their stabiliser sub-groups under the representation $\pi$ by $\Gamma_i\subset\Gamma$. We can think of the representation $R$ as being given in the following form
\begin{equation}
 R(\gamma)=\pi(\gamma){\rm diag}(r_1(\gamma),\ldots ,r_b{\gamma})\; ,
\end{equation}
where the $r_i(\gamma)$ act on the blocks $K_i$. The maps $r_i$ can now be restricted to $\Gamma_i$ so that they induce projective representations $\tilde{r}_i:\Gamma_1\rightarrow{\rm Gl}(K_i,\mathbb{C})/\mathbb{C}^*$. Focusing on $\tilde{r}_1$, we can construct all such representations from linear representations $\hat{r}_1:\hat{\Gamma}_1\rightarrow {\rm Gl}(K_1,\mathbb{C})$ of the Schur cover $\hat{\Gamma}_1$, as discussed above. The point is now that $\tilde{r}_1$ already determines the full $\pi$ homomorphism $r$, as is explained in detail in Appendix~\ref{app:toricpitwrep}. Broadly, this works as follows. First choose $\gamma_i\in\Gamma$ which map the blocks $K_i$ to the block $K_1$ under the representation $\pi$, that is, which satisfy $\pi(\gamma_i)(i)=1$. It can then be shown that, for every block $i$,  $\gamma\in\Gamma$ can be written uniquely in the form
\begin{equation}
 \gamma=\gamma_{\pi(\gamma)(i)}h\gamma_i^{-1}\; , \label{decompo}
\end{equation}
where $h\in\Gamma_1$. Then it follows that
\begin{equation}
 r_i(\gamma)=\tilde{r}_1(h)\; , \label{rres}
\end{equation}
so that all $r_i$ and, hence, $r$ is fixed in terms of $\tilde{r}_1$.\\[3mm]
\underline{\bf Example:} Coming back to our standard example, with group $\Gamma=\mathbb{Z}_2=\{1,-1\}$, we would now like to discuss a case with a non-trivial choice of $\pi$. The only non-trivial choice is, in fact, $\pi(1)=\mathbbm{1}_7$ and $\pi(-1)=\tilde{p}$, with the matrix $\tilde{p}$ given in Eq.~\eqref{ptex}. There are three orbits under the action of $\pi$ defined in this way, namely ${\cal K}_1=\{\{1,4\},\{2,5\}\}$, ${\cal K}_2=\{\{6\},\{7\}\}$ and ${\cal K}_3=\{\{3\}\}$.

Let us start with the orbit ${\cal K}_1$. The stabilizer group $\Gamma_1$ of the block $\{1,4\}$ is the trivial group and for the group elements $\gamma_i$ in Eq.~\eqref{decompo} we can choose $\gamma_{\{1,4\}}=1$ and $\gamma_{\{2,5\}}=-1$. This means that the correspondence between $\gamma\in\Gamma$ and $h\in \Gamma_1$ established by Eq.~\eqref{decompo} is given by $1\rightarrow 1$ and $-1\rightarrow 1$. As a result, from Eq.~\eqref{rres}, $r$ must be trivial on the orbit ${\cal K}_1$. A similar argument shows that $r$ must also be trivial on the orbit ${\cal K}_2$. 

Finally, for the orbit ${\cal K}_3$ the stabiliser group for the only block $\{3\}$ in this orbit is the full group $\Gamma_1=\Gamma$ and we can, for example, choose $-1\in\Gamma$ to act as $-1$ in this direction. This means the full homomorphism $R$ for the ordering $((1,4),(2,5),(3),(6),(7))$ is specified by
\begin{equation}
\pi(-1)=\tilde{p}\;,\qquad r(-1)={\rm diag}(1,1,1,1,-1,1,1)\; ,
\end{equation} 
with the matrix $\tilde{p}$ in Eq.~\eqref{ptex}. Note that this action is not freely acting.  
\shortline
To conclude our discussion, we summarise the algorithm to construct the monomorphisms $R:\Gamma\rightarrow G_A$.
\begin{itemize}
\item Find all representations $\pi:G\rightarrow P_A$, not necessarily faithful.
\item Focus on one of the representations $\pi$ and find all orbits of the blocks under its action.
\item For each orbit pick one block, $K_1$, and consider its stabiliser group $\Gamma_1$.
\item Study all projective representations $\tilde{r}_1:\Gamma_1\rightarrow {\rm Gl}(K,\mathbb{C})/\mathbb{C}^*$ by using the Schur cover $\hat{\Gamma}_1$ of $\Gamma_1$. 
\item For each such representation $\tilde{r}_1$, re-construct $r$ on each orbit and, hence, the full representation $r$.
\item Check if the combination of $\pi$ and $r$  does indeed define a monomorphism $R:\Gamma\rightarrow G_A$.
\end{itemize}

\subsection{Smoothness of hypersurfaces} \label{sect:smoothness}
A group monomorphism $R:\Gamma\rightarrow G_A$, constructed as described in the previous sub-section, provides an automorphism of the ambient space $A$ and it defines an action on the homogeneous coordinates $x_i$ of $A$. For $R$ to give rise to an automorphism of the TCY manifold $X\subset A$ the defining polynomial $p$ of $X$ needs to be invariant under the action of $R$. The general invariant polynomial $p^{R(\Gamma)}$ under the action of $R$ can be found from a generic defining polynomial $p$ of $X$ by applying the Reynolds operator to $p$, that is, by
\begin{equation}
p^{R(\Gamma)}(x) = \sum_{\gamma \in \Gamma} p\left(R(\gamma)x\right) \; . \label{reynolds}
\end{equation}
In general, the invariant polynomials $p^{R(\Gamma)}$ have fewer independent coefficients than a general $p$ and it is not guaranteed that the generic TCY space $X^{R(\Gamma)}$ defined by $p^{R(\Gamma)}$ is smooth. In the following, we devise a method to check smoothness for a TCY space $X$ defined as the zero locus of a polynomial $p$.

So far, we have used the global description of a toric variety in terms of homogeneous coordinates. Alternatively, a toric variety may be glued together from patches, one for each cone of maximal dimension (see Appendix~ \ref{sect:patches} for details). Our smoothness check will be carried in this latter description of a toric variety, proceeding patch by patch.

Concretely, to each four-dimensional cone $\sigma$ in the fan $\Sigma$ we can associate a patch $V(\sigma)$ with affine coordinates $x_i^\sigma$, which correspond to the rays ${\bf v}_i$ of $\sigma$. The defining polynomial, $p^\sigma$, on this patch can be obtained from $p$ by setting the coordinates $x_i=x_i^\sigma$ for ${\bf v}_i\in\sigma$ and setting all other coordinates to one. For the toric varieties we are considering, the patches $V(\sigma)$ are either of the form $\mathbb{C}^4$ or $\mathbb{C}^4/G$ for a finite Abelian group $G$. In the latter case, the coordinates $x_i^\sigma$ should be interpreted as ``upstairs" coordinates on $\mathbb{C}^4$. In the case $V(\sigma)=\mathbb{C}^4$ the patch is clearly non-singular. For $V(\sigma)=\mathbb{C}^4/G$ and for fans constructed from fine triangulations of reflexive four-dimensional polytopes, the action of $G$ is such that it has fixed points only at the origin, $x_i^\sigma=0$, of the patch. As a practical matter, a patch $V(\sigma)$ is smooth precisely if the matrix formed by ${\bf v}_i\in\sigma$ has determinant $\pm 1$. For all non-smooth patches identified in this way, we have to check whether or not the origin lies on the TCY hypersurface, that is, if $p^\sigma(0)= 0$. If it is for at least one patch the TCY hypersurface is singular and we can stop.

Otherwise, we can proceed and check if the defining equation $p$  gives rise to singularities. This is done patch by patch, considering the Jacobi ideals
\begin{equation}
 I^\sigma=\langle p^\sigma,\frac{\partial p^\sigma}{\partial x_1^\sigma},\ldots,\frac{\partial p^\sigma}{\partial x_4^\sigma}\rangle\; . \label{jacobi}
\end{equation} 
If all ideals $I^\sigma$ have dimension $-1$ the TCY hypersurface is smooth, otherwise it is singular.\\[3mm]
\underline{\bf Example:} We would like to check smoothness for our standard example with symmetry $\Gamma=\mathbb{Z}_2$ and the symmetry action 
\begin{equation}
(x_1,x_2,x_3,x_4,x_5,x_6,x_7) \mapsto (x_1,x_2,-x_3,-x_4,-x_5,-x_6,x_7) \;, \label{symmex1}
\end{equation}
as given in Eq.~\eqref{symm3}. Specifically, we need to verify that a generic element in the family $X^{R(\Gamma)}$ is smooth. Our general algorithm performs this check for a polynomial $p^{R(\Gamma)}$ with random coefficients, as generated from Eq.~\eqref{reynolds}. A simple choice for $X^{R(\Gamma)}$ is given by the simple polynomial
\begin{equation}\label{eq:ex3hs}
 p^{R(\Gamma)} = x_3^2 + x_1^8 x_7^4 + x_4^8 x_7^4 + x_2^8 x_6^4 + x_5^8 x_6^4 \; ,  
\end{equation}
which contains only the monomials~\eqref{vertmon} which correspond to the vertices of the polytope $\Delta$. 

We first note that there are four singular patches on the ambient space which correspond to the cones
\begin{equation}
\left[{\bf v}_1,{\bf v}_5,{\bf v}_6,{\bf v}_7\right], \left[{\bf v}_4,{\bf v}_5,{\bf v}_6,{\bf v}_7\right], \left[{\bf v}_1,{\bf v}_2,{\bf v}_6,{\bf v}_7\right], \left[{\bf v}_2,{\bf v}_4,{\bf v}_6,{\bf v}_7\right] \, .
\end{equation}
As none of the associated loci meets our hypersurface \eqref{eq:ex3hs}, these singularities do not induce singularities on $X^{R(\Gamma)}$. 

To see that there are no singularities induced by the hypersurface equation~\eqref{eq:ex3hs}, note that this equation and all its partial derivatives can only vanish in one of the patches when $x_3=0$. As the whole situation is furthermore symmetric under swapping any of $x_1 \leftrightarrow x_4$, $x_2 \leftrightarrow x_5$ or $x_6 \leftrightarrow x_7$, we only need to check a single four-dimensional cone, which we can take to be the one generated by $\left[{\bf v}_1,{\bf v}_2,{\bf v}_3,{\bf v}_6\right]$. In the corresponding patch, the hypersurface equation becomes
\begin{equation}\label{eq:ex3hsaffine}
p^{\sigma}= x_3^2 + x_1^8  + 1 + x_2^8 x_6^4 + x_6^4 \, . 
\end{equation}
It can be checked that the Jacobi ideal, $I^\sigma$, of this equation has dimension $-1$ so that $X^{R(\Gamma)}$ is indeed smooth.
\shortline
Our method for checking smoothness can, hence, be summarised as follows:
\begin{itemize}
 \item Determine the most general invariant family of $R(\Gamma)$-invariant hypersurfaces $X^{R(\Gamma)}$ by finding the corresponding polynomials $p^{R(\Gamma)}$ from Eq.~\eqref{reynolds}.
 \item For each four-dimensional cone, $\sigma\in\Sigma$, with rays ${\bf v}_i$ introduce affine coordinates $x_i^\sigma$ and the affine version, $p^\sigma$, of the polynomial $p^{R(\Gamma)}$ on the patch $V(\sigma)$. 
 \item For each four-dimensional cone $\sigma$ of $\Sigma$ carry out the following.
 \begin{itemize}
 \item Check if $V(\sigma)$ is smooth by checking if the determinant of the matrix formed by ${\bf v}_i\in\sigma$ is $\pm 1$. If $V(\sigma)$ is not smooth check if the singularity at $x_i^\sigma=0$ lies on the TCY hypersurface, that is, check if $p^\sigma(0)=0$. If so the TCY hypersurface is singular and we can stop.
 \item Check if the Jacobi ideal $I^\sigma$ in Eq.~\eqref{jacobi} has dimension $-1$. If it does not the TCY hypersurface is singular and we can stop. 
 \end{itemize}
\item If the above checks have been passed for all four-dimensional cones $\sigma\in\Sigma$ the generic $R(\Gamma)$-invariant TCY hypersurface $X^{R(\Gamma)}$ is smooth. 
 \end{itemize}
 
\subsection{Fixed point freedom} \label{sec:fpsec}
In order to check the fixed points of $R(\Gamma)$ we return to the global description of the toric variety. The idea is to check the appropriate fixed point condition in the upstairs space $B$ in Eq.~\eqref{Bdef} rather than in the more complicated toric space $A=B/{\cal G}$.

To set this up, we denote by $q:B\rightarrow A$ the projection to the quotient space. We also require lifts $\hat{\Gamma}$ and $\hat{G}_A$ of the groups $\Gamma$ and $G_A$ to the upstairs space, together with projections $\Pi:\hat{\Gamma}\rightarrow\Gamma$ and $\nu:\hat{G}_A\rightarrow G_A$ and the upstairs version $\hat{R}:\hat{\Gamma}\rightarrow\hat{G}_A$ of the homomorphism $R$ such that $R\circ\Pi=\nu\circ\hat{R}$. (In practice, $G_A$ and $\hat{G}_A$ are generated by the same matrices, where the projective relations are taken into account for $G_A$ and ignored for $\hat{G}_A$.) Writing $\gamma=\Pi(\hat{\gamma})$, we have
\begin{equation}
 q\circ \hat{R}(\hat{\gamma})=R(\gamma)\circ q\; . \label{projcom}
\end{equation} 
For $\gamma\in\Gamma$ we define the downstairs fixed point set as
\begin{equation}
 F_\gamma=\{{\bf x}\in A\;|\;R(\gamma){\bf x}={\bf x}\}\; ,
\end{equation} 
We would like to compare this set with its upstairs counterpart for $\hat{\gamma}\in\hat{\Gamma}$, where $\gamma=\Pi(\hat{\gamma})$, defined as
\begin{equation} 
\hat{F}_{\hat{\gamma}}=\{\hat{\bf x}\in\hat{A}\;|\;\hat{R}(\hat{\gamma})(\hat{\bf x})\in{\cal G}\hat{\bf x}\}\; . \label{fixup}
\end{equation} 
Then Eq.~\eqref{projcom} implies that
\begin{equation}
 q(\hat{F}_{\hat{\gamma}})=F_{\Pi(\hat{\gamma})}\; ,
\end{equation} 
that is, we obtain the downstairs fixed point set simply by projecting its upstairs counterpart to the quotient. 

Of course we are not primarily interested in the fixed point set of $R(\Gamma)$ in $A$ but in its intersection with the TCY manifold $X\subset A$. We denote this intersection by ${\cal F}_\gamma :=F_\gamma\cap X$. By $\hat{X}\subset B$ we denote a lift of the TCY manifold such that $q(\hat{X})=X$. Then, the upstairs version of the fixed point set intersected with $\hat{X}$ given by $\hat{\cal F}_{\hat{\gamma}}=\hat{F}_{\hat{\gamma}}\cap\hat{X}$ projects down to ${\cal F}_\gamma$, that is
\begin{equation}
 q(\hat{\cal F}_{\hat{\gamma}})={\cal F}_\gamma\; .
\end{equation} 
What we want is for the Calabi-Yau manifold not to intersect the fixed point set, that is, we require ${\cal F}_\gamma$ to be empty for all $\gamma\in\Gamma$. From the previous equation this is the same as saying that $\hat{\cal F}_{\hat{\gamma}}$ is empty for all $\hat{\gamma}\in\hat{\Gamma}$. Therefore, from Eq.~\eqref{fixup}, we should, in a first instance, solve the equations
\begin{equation}
 \hat{R}(\hat{\gamma}){\bf z}=g{\bf z}\; , \label{fpeq}
\end{equation} 
where $g=({\bf s}^{q_{1}},\ldots ,{\bf s}^{q_n})$ is an element of the toric group ${\cal G}$, parametrized ${\bf s}=(s_1,\ldots , s_{n-4})$. We can think of these equations as defining an ideal $I$ in the ring $\mathbb{C}[{\bf z},{\bf s}]$. (It may be necessary to multiply some components of Eq.~\eqref{fpeq} with powers of some parameters $s_r$ if the corresponding charges $q_r^i$ are negative, in order to achieve polynomial form.) To this ideal we have to add the defining TCY equation in order to get the ideal which described $\hat{\cal F}_{\hat{\gamma}}$. Then, in order for the action to be fixed point free we need the dimension of this ideal to be $-1$ for all $\hat{\gamma}\in\hat{\Gamma}$. \\[3mm]
\underline{\bf Example:} Finally, let us check for our standard example that the action of $\mathbb{Z}_2$ symmetry given in Eq.~\eqref{symmex1} is indeed fixed point free.

For this, we need to find  all fixed loci on the ambient space and intersect with the hypersurface equation $p^{R(\Gamma)} = 0$, with $p^{R(\Gamma)}$ given in Eq.~\eqref{eq:ex3hs}. The fixed loci of the ambient space are $x_3 = x_4 = x_5 = x_6=0$, together with their images under any permutation $x_1 \leftrightarrow x_4$, $x_2 \leftrightarrow x_5$ or $x_6 \leftrightarrow x_7$. Since the situation is symmetric under these permutations, it again suffices to check one such locus only. Solving $p^{R(\Gamma)} = x_3 = x_4 = x_5 = x_6= 0 $ we find that $x_1 = 0$ or $x_7=0$, a locus which is contained within in the exceptional set $Z(\Sigma)$. In conclusion, the action of the symmetry~\eqref{symmex1} is indeed fixed point free. 
\shortline

\subsection{Algorithm of classification} \label{sec3}

We now summarise the entire algorithm classifying symmetries of TCY manifolds which are realised by a linear action on the toric ambient space. For each TCY hypersurface $X$ defined by a pair of reflexive polytopes $\Delta,\Delta^\circ$, this algorithm consists of four main steps:
\begin{itemize}
\label{it:shortalg}
\item[I.] {\bf Possible symmetry groups (Section~\ref{sect:possiblegroups}):} For freely-acting symmetries, compute the common divisors of the Euler number and the signatures of twisted bundles, which provides us with a finite list of possible groups orders and, hence, a finite list of groups $\Gamma$. For non freely-acting symmetries we merely search through all finite groups $\Gamma$ up to some maximal group order. 
\item[II.] {\bf Compute the symmetries of the ambient space (Section~\ref{symzeroset}):} Compute the toric symmetry group $G_A=P_A\ltimes(H_A/{\cal G})$ from the toric data. This group consist of blocks of general linear groups contained in $H_A$ and permutations of these blocks in $P_A$. 
\item[III.] {\bf Generate the representations (Section~\ref{constrep}):} For each group $\Gamma$ found in step I, generate all group monomorphisms $R:\Gamma\rightarrow G_A$.
\item[IV.] {\bf Check the symmetry actions (Section~\ref{sect:smoothness} and \ref{sec:fpsec}):} Generate defining equations $p^{R(\Gamma)}$ which are invariant under the action of $R(\Gamma)$. Then check smoothness of the hypersurface $X^{R(\Gamma)}$ defined by $p^{R(\Gamma)}=0$ and, for the case of freely-acting symmetries,  the absence of $R(\Gamma)$ fixed points on $X^{R(\Gamma)}$.
\end{itemize}
An implementation of the above algorithm requires the data of a fan $\Sigma$ obtained from a triangulation of $\Delta^\circ$ for each model. We are using the database compiled in \cite{Altman:2014bfa} as input data containing such models. Furthermore, we are using GAP\cite{GAP4} for the required group-theoretical data. Our computation system is composed of Mathematica modules, and the computations in commutative algebra are performed using Singular \cite{DGPS}. The Calabi-Yau hypersurfaces are generated with random coefficients $c_{\bf m}$ for each monomial ranging from $0$ to $100$. If the computation in Singular is too time-consuming over the field of rational numbers, we alternatively use finite fields and check over a number of primes.

\section{Examples}
\label{eg}
We will now carry out our algorithm for two specific examples. The first example is for a freely-acting $\mathbb{Z}_4$ symmetry on the tetra-quadric hypersurface in $\mathbb{P}^1\times\mathbb{P}^1\times\mathbb{P}^1\times\mathbb{P}^1$. Since the tetra-quadric is a TCY as well as a CICY manifold, this $\mathbb{Z}_4$ symmetry is, in fact, well-known in the literature~\cite{Braun:2010vc} . This example serves as an illustration of our method in an established context. The second example is for a TCY manifold with $h^{1,1}(X)=6$ and a freely-acting $\mathbb{Z}_2$ symmetry which was not previously known. 

\subsection{A $\mb{Z}_4$ symmetry of the tetra-quadric}\label{sect:twtetraquadric}
The tetra-quadric with Hodge numbers $h^{1,1(X)}=4$ and $h^{1,2}(X)=68$ is described by a polytope $\Delta^\circ$ with eight vertices given by
\begin{equation}
\begin{array}{lllllllllllllll}
{\bf v}_1&=& (-1,0,0,0)&& {\bf v}_2&=&(-1,0,0,1)&& {\bf v}_3&=&(-1,0,1,0)&& {\bf v}_4&=&(-1,1,0,0)\\
{\bf v}_5&=&(1,-1,0,0)&&{\bf v}_6&=&(1,0,-1,0)&&{\bf v}_7&=&(1,0,0,-1)&&{\bf v}_8&=&(1,0,0,0) \; , 
\end{array}
\end{equation}
The zero set is then characterised by 
\begin{eqnarray}
\mc{I} &=& \{\{1,8\},\{2,7\},\{3,6\},\{4,5\}\} \\
Z &=& \{x_1 = x_8 = 0\} \cup \{x_2 = x_7 = 0 \} \cup \{x_3 = x_6 = 0 \} \cup \{x_4 = x_5 = 0 \} \label{esTQ}
\end{eqnarray}
The linear relations between the above vertices leads to the charge matrix
\begin{equation}
 Q= \left(
\begin{array}{cccccccc}
 0 & 0 & 0 & 1 & 1 & 0 & 0 & 0 \\
 0 & 0 & 1 & 0 & 0 & 1 & 0 & 0 \\
 0 & 1 & 0 & 0 & 0 & 0 & 1 & 0 \\
 1 & 0 & 0 & 0 & 0 & 0 & 0 & 1 \\
\end{array}\right)\; ,
\end{equation}
with the associated toric group
\begin{equation}
\mc{G} = \{ \left\{s_4,s_3,s_2,s_1,s_1,s_2,s_3,s_4\right\} | s_i \in \mb{C}^* \}\; .\label{tgTQ}
\end{equation}\\[3mm]
\underline{\bf Step I:} The Euler number of the tetra-quadric is $\eta(X)=-128$ and the other relevant twisted indices in Eqs.~\eqref{indkl} and \eqref{sigkl} can be worked out as 
\begin{equation}
 \chi_{0,1}=-64\;,\quad \chi_{1,0}=80\;,\quad  \chi_{2,0}=544 \;\quad  \chi_{3,0}=1776  \;,\quad \sigma_{1,1}=-1280\; .
\end{equation} 
The greatest common divisor of these numbers and the Euler number is $16$ so for freely-acting symmetry groups $\Gamma$ we should consider the group orders $|\Gamma|\in\{2,4,8,16\}$.\\[3mm]
\underline{\bf Step II:} In order to compute the ambient space symmetry group we need to determine the index sets ${\cal J}$ and ${\cal K}$ from the above toric data. In this particular case it turn out they are both equal to ${\cal I}$ so that
\begin{gather}
\mc{J} = \mc{K} = \mc{I} =  \{\{1,8\},\{2,7\},\{3,6\},\{4,5\}\}\; .
\end{gather}
Following the general procedure explained in Section~\ref{symzeroset} it is then straightforward to determine the ambient space symmetry group $G_A=P_A\ltimes(H_A/{\cal G})$. One finds that
\begin{eqnarray}
 P_A&\cong&S_4\\
 H_A &=&GL(\{1,8\},\mb{C}) \times GL(\{2,7\},\mb{C}) \times GL(\{3,6\},\mb{C}) \times GL(\{4,5\},\mb{C})  \label{HATQex}
\end{eqnarray} 
where $P_A$ acts by permuting the four blocks $\{1,8\},\{2,7\},\{3,6\},\{4,5\}\in{\cal K}$.\\[3mm]
\underline{\bf Step III:} In general, we need to study all group monomorphisms $R:\Gamma\rightarrow G_A$ for all groups $\Gamma$ with order $|\Gamma|\in\{2,4,8,16\}$.  For the purpose of this example, we focus on the group $\Gamma=\mathbb{Z}_4$ and we denote its generator by $\gamma$, so that $\gamma^4=1$. The representation matrices will be written down for the coordinate ordering $((1,8),(2,7),(3,6),(4,5))$, in line with the grouping into $2\times 2$ blocks in Eq.~\eqref{HATQex}. 

First, we need to construct the permutation representations $\pi:\mathbb{Z}_4\rightarrow P_A\cong S_4$. For illustration purpose we choose $\pi$ as
\begin{equation}
 \pi(\gamma)=\left(\begin{array}{cccc}0&\mathbbm{1}_2&0&0\\\mathbbm{1}_2&0&0&0\\0&0&0&\mathbbm{1}_2\\0&0&\mathbbm{1}_2&0\end{array}\right)\; , \label{piTQex}
\end{equation} 
that is, as a simultaneous swap of blocks $(1,8)\leftrightarrow (2,7)$ and $(3,6)\leftrightarrow (4,5)$. There are two orbits of blocks under the action of $\pi$, namely ${\cal K}_1=\{\{1,8\},\{2,7\}\}$ and ${\cal K}_2=\{\{3,6\},\{4,5\}\}$. We can focus on ${\cal K}_1$ (the other orbit ${\cal K}_2$ will work in exactly the same way) and note that it has two blocks $K_1=\{1,8\}$ and $K_2=\{2,7\}$. The stabilizer groups $\Gamma_1$ of $K_1$ is given by $\Gamma_1=\{1,\gamma^2\}\cong\mathbb{Z}_2\subset \mathbb{Z}_4$. We choose the representation $\tilde{r}_1$ of $\Gamma_1$ on the block $K_1$ as
\begin{equation}
 \tilde{r}_1(\gamma^2)=\sigma_3={\rm diag}(1,-1)\; .
\end{equation} 
The full map $r:\Gamma\rightarrow H_A/{\cal G}$ can then be re-constructed from $\tilde{r}_1$ and is it straightforward to show that
\begin{equation}
 r(\gamma)={\rm diag}(\mathbbm{1}_2,\sigma_3,\mathbbm{1}_2,\sigma_3)\; .
\end{equation} 
Combining this with the permutation representation~\eqref{piTQex} we finally get
\begin{equation}
  R(\gamma)=\left(\begin{array}{cccc}0&\sigma_3&0&0\\\mathbbm{1}_2&0&0&0\\0&0&0&\sigma_3\\0&0&\mathbbm{1}_2&0\end{array}\right)\; .\label{RTQ}
\end{equation}\\[3mm]
The general invariant polynomial under this symmetry action is given in Appendix \ref{app:polytq}.
\underline{\bf Step IV:} To check fixed points we have to write down Eq.~\eqref{fpeq} for all $\mathbb{Z}_4$ group elements using the group action defined by Eq.~\eqref{RTQ} and the elements~\eqref{tgTQ} of the toric group. This leads to 
\begin{equation}
\begin{array}{ll}
\gamma&\left\{
\begin{array}{lllllllllllllll}
x_2-s_4 x_1&=&0&&-s_4 x_8-x_7&=&0&&x_1-s_3 x_2&=&0&&x_8-s_3 x_7&=&0\\
x_4-s_2 x_3&=&0&&-s_2 x_6-x_5&=&0&&x_3-s_1x_4&=&0&&x_6-s_1 x_5&=&0
\end{array}\right.\\
\gamma^2&\left\{
\begin{array}{lllllllllllllll}
x_1-s_4 x_1&=&0&&-s_4 x_8-x_8&=&0&&x_2-s_3 x_2&=&0&&-s_3 x_7-x_7&=&0\\
x_3-s_2 x_3&=&0&&-s_2 x_6-x_6&=&0&&x_4-s_1x_4&=&0&&-s_1 x_5-x_5&=&0
\end{array}\right.\\
\gamma^3&\left\{
\begin{array}{lllllllllllllll}
x_2-s_4 x_1&=&0&&x_7-s_4 x_8&=&0&&x_1-s_3 x_2&=&0&&-s_3 x_7-x_8&=&0\\
x_4-s_2 x_3&=&0&&x_5-s_2 x_6&=&0&&x_3-s_1x_4&=&0&&-s_1 x_5-x_6&=&0
\end{array}\right.
\end{array}
\end{equation}
It can be shown in each case that the solutions to these equations are either contained in the exceptional set~\eqref{esTQ} or are point-like in the ambient space and, hence, do not lie on a generic tetra-quadric hypersurface.

The most general defining polynomial $p$ consistent with the $\mathbb{Z}_4$ symmetry contains $21$ independent terms and coefficients, including, for example
\begin{equation}
 p\;\ni\; x_5^2 x_6^2 x_7^2 x_8^2\,,\; x_1^2 x_2^2 x_5^2 x_6^2\,,\; x_1^2 x_2^2 x_3^2 x_4^2\,,\; x_1 x_2 x_3 x_4 x_5 x_6 x_7 x_8\,,\cdots
\end{equation} 
A calculation, for example using Singular~\cite{DGPS}, shows that a sufficiently general combination of these terms does indeed produce a smooth hypersurface.

\subsection{A $\mb{Z}_2$ symmetry of the twisted tetra-quadric}\label{h116}
We now proceed to a more complicated example which leads to a freely-acting $\mb{Z}_2$ symmetry which was not previously known. Moreover, the underlying manifold is not among the $16$ manifolds identified by Batyrev and Kreuzer and, hence, this example shows that the number of manifolds with freely-acting symmetries within the Kreuzer Skarke list is definitely larger than $16$. 

The space is defined by a polytope $\Delta^\circ$  with the $10$ vertices
\begin{equation}
\begin{array}{lllllllllll}
{\bf v}_1&=& (-1,-1,-1,-1)&& {\bf v}_2&=&(-1,0,0,0)&& {\bf v}_3&=&(0,-1,0,0)\\
{\bf v}_4&=&(0,0,-1,0)&&{\bf v}_5 &=& (0,0,0,-1)&&{\bf v}_6&=&(0,0,0,1)\\
{\bf v}_7&=& (0,0,1,0)&&{\bf v}_8&=&(0,1,0,0)&&{\bf v}_9&=&(1,0,0,0)\\
{\bf v}_{10}&=&(1,1,1,1)\; .
\end{array}\label{newexvert}
\end{equation}
The zero set is specified by
\begin{eqnarray}
{\cal I}&=&\{\{1,10\},\{2,9\},\{3,8\},\{4,7\},\{5,6\},\{1,6,7\},\{1,6,8\},\{1,6,9\},\{1,7,8\},\{1,7,9\},\nonumber\\
&&\;\;\{1,8,9\},\{2,3,4\},\{2,3,5\},\{2,3,10\},\{2,4,5\},\{2,4,10\},\{2,5,10\},\{3,4,5\},\nonumber\\
&&\;\;\{3,4,10\},\{3,5,10\},\{4,5,10\},\{6,7,8\},\{6,7,9\},\{6,8,9\},\{7,8,9\}\}\; .
\end{eqnarray}
The linear dependencies among the vertices~\eqref{newexvert} lead to the charge matrix
\begin{equation}
\cal{Q} =\left(
\begin{array}{cccccccccc}
 0 & 1 & 0 & 0 & 0 & 0 & 0 & 0 & 1 & 0 \\
 0 & 0 & 1 & 0 & 0 & 0 & 0 & 1 & 0 & 0 \\
 0 & 0 & 0 & 1 & 0 & 0 & 1 & 0 & 0 & 0 \\
 0 & 0 & 0 & 0 & 1 & 1 & 0 & 0 & 0 & 0 \\
 0 & 1 & 1 & 1 & 1 & 0 & 0 & 0 & 0 & 1 \\
 1 & 0 & 0 & 0 & 0 & 1 & 1 & 1 & 1 & 0 \\
\end{array}
\right) \; ,
\end{equation}
which identifies the corresponding TCY as a `twisted version' of the tetra-quadric. Its Hodge numbers are $h^{1,1}(X)=6$ and $h^{2,1}(X)=46$. The above charge matrix gives rise to the toric group
\begin{equation}
{\cal G} = \{(s_6,s_1 s_5,s_2 s_5,s_3 s_5,s_4 s_5,s_4 s_6,s_3 s_6,s_2 s_6,s_1 s_6,s_5)\; |\; s_i \in \mb{C}^*\}\; .
\end{equation}\\[3mm]
\newline
\underline{\bf Step I:} The Euler number of the twisted tetra-quadric is $\eta(X)=-80$ and the other relevant twisted indices are
\begin{equation}
 \chi_{0,1}=-40\;,\quad \chi_{1,0}=50\;,\quad  \chi_{2,0}=330  \;\quad  \chi_{3,0}=1070\;,\quad \sigma_{1,1}=-1080\; .
\end{equation} 
The greatest common divisor of these numbers and the Euler number is $10$ so for freely-acting symmetry groups $\Gamma$ we should consider the group orders $|\Gamma|\in\{2,5,10\}$.\\[3mm]
\\[3mm]
\underline{\bf Step II:} As can be seen from the charge matrix and the exceptional set given above, we have that $\mc{J}  = \mc{K} = \{ \{1\}, \{2\}, \{3\}, \{4\}, \{5\}, \{6\}, \{7\},\{8\}, \{9\}, \{10\} \}$. \\[3mm]
The ambient space symmetry group $G_A=P_A\ltimes(H_A/{\cal G})$ is found from our general procedure as
\begin{eqnarray}
 P_A&\cong&S_4 \times \Z_2 \\
 H_A &=& GL(1,\C)^{10} = ( \C^*)^{10}
\end{eqnarray} 
where the $S_4$ in $P_A$ acts by permuting the four ordered tuples $(2,9),(3,8),(4,7),(5,6)$ and the $\Z_2$ acts by mapping
$(1,2,3,4,5) \leftrightarrow (10,9,8,7,6)$. \\[3mm] 
\underline{\bf Step III:} Here, we need to study all group monomorphisms $R:\Gamma\rightarrow G_A$ for all finite groups $|\Gamma|\in\{2,5,10\}$. For the purpose of this example, we focus on the group $\Gamma=\mathbb{Z}_2$ and we denote its generator by $\gamma$. A simple choice of representation is to take $r$ to be trivial and to set the permutation part to
\begin{equation}\label{eq:groupmyEx}
\pi(\gamma) = ((1,10),(2,9),(3,8),(4,7),(5,6))\; .
\end{equation}
\\[3mm]
\underline{\bf Step IV:}
As described earlier, a generic polynomial symmetric under the $\Z_2$ \eqref{eq:groupmyEx} may easily be found by summing a generic polynomial with its image, the result is found in Appendix \ref{app:monosttq}. Such a polynomial gives rise to a smooth hypersurface $X^{R(\Z_2)}$. To see the fixed point freedom, we have to find solutions to \eqref{fpeq}. For the action above, fixed points are found as solutions of 
\begin{equation}\label{eq:fixedpointsmyEX}
\begin{aligned}
& (x_1 s_6, x_2 s_1 s_5, x_3 s_2 s_5, x_4 s_3 s_5, x_5 s_4 s_5, x_6 s_4 s_6, x_7 s_3 s_6, x_8 s_2 s_6, x_9 s_1 s_6, x_{10} s_5) \\
= & (x_{10}, x_9,x_8,x_7,x_6,x_5,x_4,x_3,x_2,x_1) \, .
\end{aligned}
\end{equation}
To find the corresponding loci on the toric variety, we may consider the elimination ideal with respect to the variables $s_i$. It is generated by all polynomials of the form $x_i^2 x_j^2 - x_{10-i}^2x_{10-j}^2$ for $i,j = 1..5$ and has dimension $5$, so that there are only fixed points under this group action. As none of these points meet a generic smooth symmetric polynomial, it follows that \eqref{eq:fixedpointsmyEX} has no simultaneous solutions with a generic smooth symmetric polynomial and there are not fixed points under the group action \eqref{eq:groupmyEx}.
\\[3mm]

\section{Systematic Search for Symmetries} \label{scan}
The classification of reflexive four-dimensional polytopes by Kreuzer and Skarke \cite{Kreuzer:2000xy} forms a natural starting point to apply our algorithm and find new examples of Calabi-Yau threefolds with freely acting symmetries. A scan over all triangulations off all reflexive polytopes, however, is presently beyond reach. In particular, it is already a daunting task to find all triangulations, as the number of triangulations of each reflexive polytope increases dramatically with the Picard number. In Ref.~\cite{Altman:2014bfa} all triangulations for $h^{1,1}(X)\leq 6$ were found\footnote{The largest number of reflexive polytopes for fixed $h^{1,1}(X)$ exists for $h^{1,1}(X)=27$.}. For the purpose of our symmetry classification we will focus on the subset of those triangulations with $h^{1,1}(X)\leq 3$, some $350$ manifolds to which we apply our classification algorithm. 

\subsection{Scan results for freely-acting symmetries}
In this sub-section we present the result of our classification, a list of all freely-acting symmetries, linearly realised in the ambient space, for TCY manifolds with $h^{1,1}(X)\leq 3$. We find that only five of these TCY manifolds allow for freely-acting symmetries of this kind and we briefly list those five cases, together with the most important data, in the following. The invariant polynomials are recorded in Appendix \ref{app:invpolys}. \\[3mm]
\newline
\underline{\bf Case \#1, quintic in $\mathbb{P}^4$:}
\begin{itemize}
\item Hodge numbers: $h^{1,1}(X)=1$, $h^{2,1}(X)=101$
\item Possible symmetry orders from indices: $|\Gamma|\in\{5,25\}$
\item Ray generators of fan $\Sigma$:
\begin{equation}\label{eq:raysquintic}
({\bf v}_1,\ldots ,{\bf v}_5)=
  \left(
\begin{array}{rrrrr}
-1&-1&-1&-1&4\\
0&0&0&1&-1\\
0&0&1&0&-1\\
0&1&0&0&-1
\end{array}
\right)
\end{equation}
\item Charge matrix:
\begin{equation}
Q=\left(\begin{array}{ccccc}1&1&1&1&1\end{array}\right)
\end{equation}
\item Symmetry $\Gamma=\mathbb{Z}_5=\langle \gamma\rangle$:
\begin{equation}\label{eq:z5quinticaction1}
 R(\gamma)={\rm diag}(1,\alpha_5,\alpha_5^2,\alpha_5^3,\alpha_5^4)\;,\qquad \alpha_5=e^{2\pi i/5}\; .
\end{equation} 
\item Symmetry $\Gamma=\mathbb{Z}_5\times\mathbb{Z}_5=\langle \gamma_1,\gamma_2\rangle$, with four possible actions of $\gamma_1$
\begin{equation}\label{eq:z5quinticaction2}
\begin{array}{llll}
(1)\;:&R(\gamma_1)&=&{\rm diag}(1,\alpha_5^3,\alpha_5,\alpha_5^4,\alpha_5^2)\\
(2)\;:&R(\gamma_1)&=&{\rm diag}(1,\alpha_5,\alpha_5^2,\alpha_5^3,\alpha_5^4)\\
(3)\;:&R(\gamma_1)&=&{\rm diag}(1,\alpha_5^4,\alpha_5^3,\alpha_5^2,\alpha_5)\\
(4)\;:&R(\gamma_1)&=&{\rm diag}(1,\alpha_5^2,\alpha_5^4,\alpha_5,\alpha_5^3)\\
\end{array}
\end{equation}
and the same action of $\gamma_2$
\begin{equation}
R(\gamma_2)=\left(
\begin{array}{ccccc}
 0 & 0 & 0 & 0 & 1 \\
 1 & 0 & 0 & 0 & 0 \\
 0 & 1 & 0 & 0 & 0 \\
 0 & 0 & 1 & 0 & 0 \\
 0 & 0 & 0 & 1 & 0 \\
\end{array}\right)
\end{equation}
in all four cases.
\item Remarks: This is polytope ID\# 2 of the database of \cite{Altman:2014bfa}. The quintic is of course a CICY manifold and all the above freely-acting symmetries are well-known. The $\mathbb{Z}_5$ symmetry is toric and among the $16$ cases found by Batyrev and Kreuzer, while the four $\mathbb{Z}_5\times\mathbb{Z}_5$ actions are non-toric.
\end{itemize}
\vskip 3mm
\underline{\bf Case \#2, bi-cubic in $\mathbb{P}^2\times \mathbb{P}^2$:}
\begin{itemize}
\item Hodge numbers: $h^{1,1}(X)=2$, $h^{2,1}(X)=83$
\item Possible symmetry orders from indices: $|\Gamma|\in\{3,9\}$
\item Ray generators of fan $\Sigma$:
\begin{equation}
({\bf v}_1,\ldots ,{\bf v}_6)=
  \left(
\begin{array}{rrrrrr}
-1&-1&2&-1&-1&2\\
0&0&0&0&1&-1\\
0&0&0&1&0&-1\\
0&1&-1&0&0&0
\end{array}
\right)
\end{equation}
\item Charge matrix:
\begin{equation}
Q=\left(\begin{array}{cccccc}0&0&0&1&1&1\\1&1&1&0&0&0\end{array}\right)
\end{equation}
\item Symmetry $\Gamma=\mathbb{Z}_3=\langle \gamma\rangle$:
\begin{equation}
 R(\gamma)={\rm diag}(1,\alpha_3,\alpha_3^2,1,\alpha_3,\alpha_3^2)\;,\qquad \alpha_3=e^{2\pi i/3}\; .
\end{equation} \label{eq:actionbicubic1}
\item Symmetry $\Gamma=\mathbb{Z}_3\times\mathbb{Z}_3=\langle \gamma_1,\gamma_2\rangle$, with four possible actions of $\gamma_1$
\begin{equation}\label{eq:actionbicubic2}
\begin{array}{llll}
(1)\;:&R(\gamma_1)&=&{\rm diag}(1,\alpha_2^2,\alpha_3,1,\alpha_3^2,\alpha_3)\\
(2)\;:&R(\gamma_1)&=&{\rm diag}(1,\alpha_3^2,\alpha_3,1,\alpha_3,\alpha_3^2)\\
(3)\;:&R(\gamma_1)&=&{\rm diag}(1,\alpha_3,\alpha_3^2,1,\alpha_3^2,\alpha_3)\\
(4)\;:&R(\gamma_1)&=&{\rm diag}(1,\alpha_3,\alpha_3^2,1,\alpha_3,\alpha_3^2)\\
\end{array}
\end{equation}
and the same action of $\gamma_2$
\begin{equation}
R(\gamma_2)=
\left(
\begin{array}{cccccc}
 0 & 0 & 1 & 0 & 0 & 0 \\
 1 & 0 & 0 & 0 & 0 & 0 \\
 0 & 1 & 0 & 0 & 0 & 0 \\
 0 & 0 & 0 & 0 & 0 & 1 \\
 0 & 0 & 0 & 1 & 0 & 0 \\
 0 & 0 & 0 & 0 & 1 & 0 \\
\end{array}\right)
\end{equation}
in all four cases.
\item Remarks: This is polytope ID\# 10 of the database of \cite{Altman:2014bfa}. The bi-cubic is of course a CICY manifold and all the above freely-acting symmetries have been found before~\cite{Braun:2010vc}. The $\mathbb{Z}_3$ symmetry is toric and among the $16$ cases found by Batyrev and Kreuzer, while the four $\mathbb{Z}_3\times\mathbb{Z}_3$ actions are non-toric.
\end{itemize}
\vskip 3mm
\underline{\bf Case \#3:}
\begin{itemize}
\item  Hodge numbers: $h^{1,1}(X)=3$, $h^{2,1}(X)=83$
\item Possible symmetry orders from indices: $|\Gamma|\in\{2,4,8\}$
\item Ray generators of the fan $\Sigma$:
\begin{equation}
({\bf v}_1,\ldots ,{\bf v}_7)=
\left(\begin{array}{rrrrrrr}
-1 & -1 &  1  & -1  & -1 &-1 &-1\\
 0 &  0 & -1  &  0  &  4 &2&0\\
 0 &  0 &  0  &  2  & -2 &-1&1\\
 0 &  1 &  0  &  0  & -1&0&0
\end{array}\right)
\end{equation} \label{vertc3}
\item Charge matrix:
\begin{equation}
Q= \left(
\begin{array}{rrrrrrr}
 1 &  0 & 0 & 1 &  0 & 0  & -2  \\ 
 0 &  1 & 0 & 0  & 1 & -2 & 0 \\ 
 0 &  0 & 2 & 0  & 0 & 1  & 1    
\end{array}  
\right)\label{cmc3}
\end{equation}
\item Symmetry $\Gamma=\mathbb{Z}_2=\langle \gamma\rangle$:
\begin{equation}\label{eq:actioncase3}
 R(\gamma)={\rm diag}(1,1,-1,-1,-1,-1,1)
\end{equation} 
\item Remarks: This is polytope ID\# 147 of the database of \cite{Altman:2014bfa}. The manifold and its freely-acting $\mathbb{Z}_2$ symmetry is among the $16$ cases identified by Batyrev and Kreuzer.
\end{itemize}
\vskip 3mm
\underline{\bf Case \#4:}
\begin{itemize}
\item  Hodge numbers: $h^{1,1}(X)=3$, $h^{2,1}(X)=115$
\item Possible symmetry orders from indices: $|\Gamma|\in\{2,4,8\}$
\item Ray generators of fan $\Sigma$:
\begin{equation}
({\bf v}_1,\ldots ,{\bf v}_7)=
\left(\begin{array}{rrrrrrr}
-1&-1&-1&-1&-1&1&-1\\
0&0&0&2&2&-1&0\\
0&0&1&-1&0&0&0\\
0&2&1&-1&-1&0&1
\end{array}\right)\label{vertc4}
\end{equation}
\item Charge matrix:
\begin{equation}
Q= \left(
\begin{array}{rrrrrrr}
0&0&0&0&1&2&1\\
0&0&1&1&0&2&0\\
1&1&0&0&2&4&0
\end{array}\right)\label{Qc4}
\end{equation}
\item Symmetry $\Gamma=\mathbb{Z}_2=\langle \gamma\rangle$:
\begin{equation}\label{eq:actioncase4_1}
 R(\gamma)={\rm diag}(1,-1,1,-1,1,-1,-1)
\end{equation} \label{eq:actioncase4_2}
\item Symmetry $\Gamma=\mb{Z}_2 \ltimes \mb{Z}_2=\langle\gamma_1,\gamma_2\rangle$:
\begin{equation}
 R(\gamma_1)={\rm diag}(1,-1,1,-1,1,-1,-1)\,,\;\; R(\gamma_2)={\rm diag}(\sigma_1,\sigma_1,1,-1,-1)\,,\;\;
  \sigma_1=\left(\begin{array}{cc}0&1\\1&0\end{array}\right)
\end{equation}  
\item Remarks: This is polytope ID\# 245 of the database of \cite{Altman:2014bfa}. This manifold with its freely-acting $\mathbb{Z}_2$ symmetry is among the $16$ cases identified by Batyrev and Kreuzer. The polytope has two triangulations realizing different regions of the K\"ahler cone. They both allow the same fixed-point free actions as given above. However, the $\mb{Z}_2 \ltimes \mb{Z}_2$ is non-toric and has not been found before. We will take a closer look at this symmetry in the following sub-section.
\end{itemize}
\vskip 3mm
\underline{\bf Case \#5:}
\begin{itemize}
\item  Hodge numbers: $h^{1,1}(X)=3$, $h^{2,1}(X)=115$
\item Possible symmetry orders from indices: $|\Gamma|\in\{2,4,8\}$
\item Rays of fan $\Sigma$:
\begin{equation}
({\bf v}_1,\ldots ,{\bf v}_7)=
\left(\begin{array}{rrrrrrr}
-1&-1&-1&-1&-1&-1&1\\
0&0&0&2&2&2&-1\\
0&1&1&-1&-1&0&0\\
0&0&1&-1&0&0&0
\end{array}\right)
\end{equation}
\item Charge matrix:
\begin{equation}
Q= \left(
\begin{array}{rrrrrrr}
0&0&1&1&0&0&2\\
0&1&0&0&1&0&2\\
1&0&0&0&0&1&2
\end{array}\right)
\end{equation}
\item Symmetry $\Gamma=\mathbb{Z}_2=\langle \gamma\rangle$:
\begin{equation}\label{eq:actioncase5}
 R(\gamma)={\rm diag}(-1,1,-1,1,-1,1,-1)
\end{equation}
\item Remarks: This is polytope ID\# 249 of the database of \cite{Altman:2014bfa}. This manifold with its freely-acting $\mathbb{Z}_2$ symmetry is among the $16$ cases identified by Batyrev and Kreuzer. The polytope has three different triangulations realizing different regions of the K\"ahler cone. All of them allow the same fixed-point free action given above. 
\end{itemize}

\subsection{A closer look at the $\mb{Z}_2 \ltimes \mb{Z}_2$ symmetry in case \#4}
In the previous classification, we have found one new case, not previously known from the classification of symmetries for CICYs or the Batyrev Kreuzer classification. We would now like to discuss this case in more detail.

We recall that the vertices of $\Delta^\circ$ for this TCY manifold are given by Eq.~\eqref{vertc4}, the associated charge matrix by~\eqref{Qc4} and that the Hodge numbers are $h^{1,1}(X)=3$ and $h^{2,1}(X)=115$. We consider the following triangulation\footnote{There is a second triangulation with zero set $\mc{I} = \{\{1,2\},\{3,4\},\{5,6,7\}\}$ for which the discussion is similar to the one presented here.}
\begin{align}
&[{\bf v}_1,{\bf v}_3,{\bf v}_4,{\bf v}_5],\;[{\bf v}_1,{\bf v}_3,{\bf v}_4,{\bf v}_7],\;[{\bf v}_1,{\bf v}_3,{\bf v}_5,{\bf v}_6],\;[{\bf v}_1,{\bf v}_3,{\bf v}_6,{\bf v}_7],\;[{\bf v}_1,{\bf v}_4,{\bf v}_5,{\bf v}_6],\;[{\bf v}_1,{\bf v}_4,{\bf v}_6,{\bf v}_7], \nonumber \\
&[{\bf v}_2,{\bf v}_3,{\bf v}_4,{\bf v}_5],\;[{\bf v}_2,{\bf v}_3,{\bf v}_4,{\bf v}_7],\;[{\bf v}_2,{\bf v}_3,{\bf v}_5,{\bf v}_6],\;[{\bf v}_2,{\bf v}_3,{\bf v}_6,{\bf v}_7],\;[{\bf v}_2,{\bf v}_4,{\bf v}_5,{\bf v}_6],\;[{\bf v}_2,{\bf v}_4,{\bf v}_6,{\bf v}_7]
\end{align}
which consists of twelve four-dimensional cones and leads to the zero set
\begin{eqnarray}
\mc{I} &=& \{\{1,2\},\{5,7\},\{3,4,6\}\}\label{Ic4}\\
 Z(\Sigma)&=& \{x_1 = x_2 = 0\} \cup \{x_5 = x_7 = 0 \} \cup \{x_3 = x_4 = x_6 = 0 \} \label{zsc4}
\end{eqnarray}
The charge matrix~\eqref{Qc4} implies the following toric group:
\begin{equation}
\mc{G} = \{ (s_3,s_3,s_2,s_2,s_1s_3^2,s_1^2s_2^2s_3^4,s_1)\; |\; s_1,s_2,s_3 \in \mb{C}^* \} \label{tgc4}
\end{equation}
\\[3mm]
\underline{\bf Step I:} The Euler number for this space is given by $\eta(X)=-224$ and the relevant indices~\eqref{indkl} and \eqref{sigkl} by
\begin{equation}
 \chi_{0,1}=-112\;,\quad \chi_{1,0}=152\;,\quad \chi_{2,0}=1072\;,\quad \chi_{3,0}=3528\;,\quad \sigma_{1,1}=-1280\; .
\end{equation} 
The greatest common divisors of the numbers and the Euler number is $8$, so the possible orders of freely-acting groups are $\Gamma\in\{2,4,8\}$.\\[3mm]
\underline{\bf Step II:} From the index set ${\cal I}$ in Eq.~\eqref{Ic4} we find the two associated index sets
\begin{equation}
\mc{J} = \{\{1,2\},\{5,7\},\{3,4,6\}\}\;,\qquad \mc{K} = \{\{1,2\},\{3,4\},\{5\},\{6\},\{7\}\}\; .
\end{equation}
These sets imply an ambient space symmetry group $G_A=P_A\ltimes H_A/{\cal G}$ with
\begin{equation}
 H_A={\rm Gl}(\{1,2\},\mathbb{C})\times{\rm Gl}(\{3,4\},\mathbb{C})\times{\rm Gl}(\{5\},\mathbb{C})\times{\rm Gl}(\{6\},\mathbb{C})\times{\rm Gl}(\{7\},\mathbb{C})\; ,
\end{equation}
while $P_A$ is trivial.\\[3mm]
\underline{\bf Step III:} In general, we should consider all groups of order $|\Gamma|\in\{2,4,8\}$ but for the present discussion we focus on the order four group $\Gamma=\mb{Z}_2 \ltimes \mb{Z}_2=\langle\gamma_1,\gamma_2\rangle$. There are numerous possible group monomorphisms $R:\Gamma\rightarrow G_A$ and we consider the specific example with
\begin{equation}
 R(\gamma_1)={\rm diag}(1,-1,1,-1,1,-1,-1)\,,\;\; R(\gamma_2)={\rm diag}(\sigma_1,\sigma_1,1,-1,-1)\,,\;\;
  \sigma_1=\left(\begin{array}{cc}0&1\\1&0\end{array}\right)\; . \label{Rc4}
\end{equation} 
Note that, while these two matrices generate a group of order eight, this is reduced to order four once projective identifications are taken into account.\\[3mm]
\underline{\bf Step IV:} The most general TCY manifold of this type is defined by linear combination of $153$ monomials (corresponding to the 153 lattice points on $\Delta$) which, in line with Eq.~\eqref{eq:cydefequn}, correspond to the points in the polar polytope $\Delta$. These monomials have the form
\begin{equation}
 x_1^8 x_4^4 x_7^4\;,\quad x_1^8 x_3 x_4^3 x_7^4\;,\quad x_1^8 x_3^2 x_4^2 x_7^4\;,\;\cdots\;,\quad x_2^4 x_3^4 x_5^2 x_7^2\;,\quad
 x_3^3 x_4 x_5^4\;,\quad x_2^2 x_3^4 x_5^3 x_7\;,\quad x_3^4 x_5^4\;,\;\cdots
\end{equation} 
Requiring invariance under the action of the symmetry~\eqref{Rc4} reduces this list to $40$ independent invariant polynomials of the form
\begin{equation}
x_6^2 \;,\quad x_2^2 x_3^4 x_5^3 x_7 -  x_1^2 x_4^4 x_5^3 x_7 \;,\;\cdots\; ,\quad x_2^8 x_3^4 x_7^4 + x_1^8 x_4^4 x_7^4 \; .
\end{equation}
It can be checked that a generic linear combination $p^{R(\Gamma)}$ of these $40$ invariant building blocks leads to a smooth TCY manifold $X^{(R(\Gamma)}$. 

In order to check fixed point freedom, we focus on the action of $R(\gamma_2)$. With the toric group~\eqref{tgc4}, Eq.~\eqref{fpeq} leads to the following fixed point equations
\begin{equation}
\begin{array}{rrrrrrrrrrrrrrr}
s_1 s_3^2 x_5-x_5&=&0&&s_1^2 s_2^2 s_3^4 x_6+x_6&=&0&&s_1 x_7+x_7&=&0&&s_3 x_1-x_2&=&0 \\
s_3 x_2-x_1&=&0&&s_2 x_3-x_4&=&0&&s_2 x_4-x_3&=&0
\end{array}
\end{equation}
for $R(\gamma_2)$. After excluding the zero set~\eqref{zsc4}, it can be checked that no component of the so-defined fixed point set intersects the generic TCY manifold $X^{R(\Gamma)}$. The same conclusion can be reached for the action of the other group elements.

\section{Conclusion}\label{concl}
In this paper, we have analysed symmetry groups of CY three-folds defined as hypersurfaces in toric ambient spaces (TCY manifolds) and we have developed methods to classify those symmetries which descend from a linear action on the toric ambient spaces. Based on these results, we have set up an algorithm that allows for the classification of symmetries for TCY manifolds obtained by triangulation from the Kreuzer-Skarke list. We have carried out this classification for a small initial sample of all TCY manifolds with $h^{1,1}(X)\leq 3$, some $350$ manifolds. 

We find that only $5$ out of these $350$ TCY manifolds admit (linearly realised) freely-acting symmetries. Our algorithm reproduces all  freely-acting symmetries known previously from either the overlap with the list of CICY manifolds~\cite{Braun:2010vc} or the classification of toric symmetries by Batyrev and Kreuzer~\cite{Batyrev:2005jc}. We find one new, non-toric $\mathbb{Z}_2\ltimes\mathbb{Z}_2$ symmetry, where one of the $\mathbb{Z}_2$ factors corresponds to  one of the $16$ symmetries identified by Batyrev and Kreuzer.

We have also used our methods to construct a new, freely-acting $\mathbb{Z}_2$ symmetry on a TCY manifold with $h^{1,1}(X)=6$ which was not previously known to have freely-acting symmetries. This example shows that the subset of TCY manifolds with freely-acting symmetries is larger than the $16$ cases identified by Batyrev and Kreuzer. 

From our preliminary evidence it appears that the fraction of TCY manifolds with freely-acting (linearly realised) symmetries is relatively small, namely five manifolds out of $350$. This figure is broadly in line with what was found in the context of CICY manifolds~\cite{Braun:2010vc}, where about $2.5\%$ of the 7890 types carry linearly realised, freely-acting symmetries. 

Naturally, our algorithm can also be used to find non-freely acting symmetries. There is no a-priori constraint on the group order anymore but we can simply search all finite groups up to a maximal order and, of course, omit the fixed point check from the algorithm. 

It would be desirable to apply our algorithm to a larger class of TCY manifolds. As stands this is possible for TCY manifolds with $h^{1,1}(X)\leq 6$ since the corresponding triangulations have been worked out in Ref.~\cite{Altman:2014bfa}. To extend our analysis to the entire Kreuzer-Skarke list of about half a billion reflexive polytopes is a formidable challenge, not least because working out the triangulations for the entire list is computationally difficult. Dealing with this will certainly require some modifications and refinements of our algorithm. For example, it is conceivable that necessary conditions for the existence of freely-acting symmetries can be used to select a subset of polytopes for triangulation. These and related issues are the subject of future work.

\section*{Acknowledgements}
A.~L.~and A.~B.~would like to acknowledge support by the STFC grant~ST/L000474/1.  A.~L.~is also partially supported by the EPSRC network grant EP/N007158/1.
\newpage


\newpage
\appendix{}
\addtocontents{toc}{\protect\setcounter{tocdepth}{1}}

\section{The Automorphism Group of a Toric Variety}
\label{theoryfree}
  
\subsection{Demazure's structure theorem}

The automorphism group of a toric variety is elegantly captured in terms of Demazure's structure theorem \cite{Oda1988,cox2011toric}. It states that
the Lie algebra of the automorphism group of a toric variety is generated by the action of the open dense algebraic torus giving it its name, as well as maps 
\begin{equation}\label{liealgroots}
\phi_{i \nu}: x_i \mapsto x_i +  \lambda \prod_{k \neq i} x_k^{\langle \alpha_{i \nu}, {\bf v}_k \rangle} \, ,
\end{equation}
where the $\alpha_{i \nu}$ are lattice points in $M$ for which $\langle \alpha_{i \nu}, {\bf v}_i \rangle = -1$ and 
$\langle \alpha_{i \nu}, {\bf v}_k \rangle > -1$ for all $k \neq i$. The $\alpha_{i \nu}$ are called Demazure roots. 
We denote the group generated by \eqref{liealgroots} and the action of the algebraic torus by ${\rm Aut}_0(\P_{\Sigma'})$.\\

The full automorphism group is recovered as follows \cite{Oda1988,cox2011toric}. For any automorphism  of the fan, i.e. maps in ${\rm Gl}(N)$ which preserve the set of cones in $\Sigma$, there is an associated permutation of homogeneous coordinates which gives rise to an automorphism of
$\P_{\Sigma}$. Let us denote the group of all such maps by ${\rm Aut}(N,\Sigma)$. This automorphism is trivial if it is a Weyl reflection, defined by
\begin{equation}
 W(N,\Sigma): \{ n \mapsto n - \langle \alpha,n \rangle ({\bf v}_i (\alpha) - {\bf v}_j(-\alpha)) \}
\end{equation}
for any pair of semi-simple roots $\alpha$ (a root of ${\bf v}_i(\alpha)$) and $-\alpha$ (a root of ${\bf v}_j(-\alpha)$. Then, the following relation holds:
\begin{equation}
 {\rm Aut}(\P_{\Sigma'})/{\rm Aut}_0(\P_{\Sigma'}) = {\rm Aut}(N,\Sigma')/ W(N,\Sigma')\, .
\end{equation}
Note that the Demazure roots, and hence the continuous part of the automorphism group, is completely independent of the triangulation, 
as it only depends on the one-dimensional cones. We are going to show in section \ref{sect:setJ} that this remains true when we 
restrict to the Lie algebra of linearly realized automorphisms.\\

In the following, we will prove a statement analogous to Demazure's structure theorem for automorphisms linearly realized on the homogeneous coordinates of a toric variety which is suited for our classification of symmetries. This is based on the ``global" presentation of a toric variety $A$ as the quotient
\begin{equation}
 A=B/{\cal G}\;,\qquad B=\C^n - Z(\Sigma)\; .
\end{equation}

\subsection{Zero set and permutation group}

\label{theozeroset}

We start with an analysis of the symmetries of the zero set $Z(\Sigma)$ of a toric variety corresponding to the index set ${\cal I}$. Its elements are sets $I\subset\{1,\ldots,n\}$ which encode the components of the zero set as defined in Eq.~\eqref{Zdef}. For this index set we define the following  permutations groups:
\begin{equation}
 S=\{\sigma\in S_n\,|\,\sigma(I)=I,\;\forall\, I\in{\cal I}\}\subset S_n\; ,\qquad R=\{\sigma\in S_n\,|\,\sigma(I)\in{\cal I},\;\forall I\in{\cal I}\}\subset S_n\; .
\end{equation} 
The group $S$ leaves all sets $I$ in ${\cal I}$ invariant individually and we refer to it as the stabiliser group while $R$ permutes the sets $I$ into each other. Clearly, $S$ is a sub-group of $R$. It is actually a normal sub-group  since for $s\in S$ and $r\in R$ we have $rsr^{-1}(I)=I$ and, hence, $rsr^{-1}\in S$. We can, therefore, define the quotient group 
\begin{equation}
 P=R/S\; ,
\end{equation}
 which can be thought of as the group of permutations on ${\cal I}$, induced 
from permutations in $S_n$. We can write down the short exact group sequence
 \begin{equation}
  1\rightarrow S\stackrel{\imath}{\rightarrow} R\stackrel{[\cdot]}{\rightarrow} 
P\rightarrow 1\; . \label{seq}
 \end{equation} 
where $\imath$ is the inclusion map and $[\cdot ]$ denotes taking the 
equivalence class. We now show that this sequence splits. We begin with the 
\begin{remark} Denote by ${\cal I}^c=\{\{1,\ldots ,n\}\backslash I\,|\,I\in{\cal 
I}\}$ the complementary sets. Then the groups $S$ and $R$ are the same for 
${\cal I}$ and for ${\cal I}^c$. This simply follows from the fact that a set 
$I$ is invariant (is mapped to another set $I'$) under a permutation $\sigma\in 
S_n$ if and only if the complement $\{1,\ldots  ,n\}\backslash I$ is invariant 
(is mapped to $\{1,\ldots  ,n\}\backslash I'$) under $\sigma$.
\end{remark}
The complication is that the subsets in ${\cal I}$ can overlap. It is, 
therefore, useful to split the set $\{1,\ldots ,n\}$ up in a different way. 
Define an equivalence relation on $\{1,\ldots ,n\}$ by
\begin{equation}
 i\sim j :\,\Leftrightarrow i\mbox{ and }j \mbox{ are contained in the same sets 
}I\in{\cal I}\; . \label{equiv}
\end{equation} 
We denote the equivalence classes by ${\cal J}=\{1,\ldots ,n\}/\sim$ and it 
follows that
\begin{equation}
 \{1,\ldots ,n\}=\bigcup_{J\in{\cal J}}J
\end{equation}
as a disjoint union. Note that for a $J\in{\cal J}$ and $I\in{\cal I}$ with 
$J\cap I\neq\emptyset$ it follows that $J\subset I$. Indeed, if $j\in J\cap I$ 
and $k\in J$ then $k\sim j$ and, hence, $k$ and $j$ must be contained in the 
same sets of ${\cal I}$. Since $j\in J$ it follow that $k\in J$, so that 
$J\subset I$. This means every $I\in{\cal I}$ can be written as a disjoint union
\begin{equation}
 I=\bigcup_{J\in{\cal J},J\cap I\neq\emptyset}J\; . \label{decomp}
\end{equation} 
\begin{lemma} For $s\in S$ it follows that $s(J)=J$ for all $J\in{\cal J}$. 
Further, for $r\in R$ it follows that $r(J)\in{\cal J}$ for all $J\in{\cal J}$.
\label{lemmaJ}\end{lemma}
\begin{proof} For the first part consider a set $J\in{\cal J}$, a permutation 
$s\in S$ and any $j\in J$. Assume that $i$ is contained in precisely the sets 
$I_1,\ldots ,I_p\in{\cal I}$. Hence, $s(i)\in s(I_k)=I_k$, so that $s(j)$ is 
also contained in $I_1,\ldots ,I_p\in{\cal I}$. Also, $s(j)$ cannot be contained 
in any other $I\in{\cal I}$, which is not among $I_1,\ldots ,I_p\in{\cal I}$, 
since $j$ would have to be in $I$ in this case. This implies that $i\sim s(j)$, 
so $s(j)\in J$ and, hence $s(J)\subset J$. Since $s$ is injective and $J$ is a 
finite set this means $s(J)=J$. \\
For the second part, again focus on a $J\in {\cal J}$ and a permutation $r\in 
R$. Pick a $j\in J$ and call $\tilde{J}\in{\cal J}$ the equivalence class of the 
image $r(j)$. I want to show that $r(J)\subset\tilde{J}$. Let $k\in J$, so that 
$j$ and $k$ are contained in the same sets $I_1,\ldots ,I_p\in{\cal I}$. Then, 
both $r(j)$ and $r(k)$ are contained in $r(I_1),\ldots ,r(I_p)\in{\cal I}$ and 
in no further $I\in{\cal I}$. This means that $r(k)\sim r(j)$, so 
$r(k)\in\tilde{J}$ and, hence, $r(J)\subset\tilde{J}$. Further, every 
$l\in\tilde{J}$ is the image under $r$ of $r^{-1}(l)\in J$, so that 
$r(J)=\tilde{J}$. $\Box$
\end{proof}
We can now understand the structure of the stabilizer group $S$.
\begin{theorem} The stabilizer group $S$ is given by $S=\bigotimes_{J\in{\cal 
J}}S(J)$, where $S(J)$ is the group of permutations on the set $J$. \label{thS} 
\end{theorem}
\begin{proof} From Lemma~\eqref{lemmaJ} it is clear that $S\subset 
\bigotimes_{J\in{\cal J}}S(J)$. Conversely, let $s\in \bigotimes_{J\in{\cal 
J}}S(J)$, so that $s(J)=J$ for all $J\in{\cal J}$. Since, from 
Eq.~\eqref{decomp}, every $I\in{\cal I}$ can be written as a union of certain 
$J\in{\cal J}$ it follows that $s(I)=I$ for all $I\in{\cal I}$, so $s\in S$ and, 
hence, $\bigotimes_{J\in{\cal J}}S(J)\subset S$. $\Box$
\end{proof}

Returning to the general story, we call a permutation $r\in R$ {\em 
order-preserving} iff for all $J\in{\cal J}$ and all $i,j\in J$ with $i<j$ it 
follows that $r(i)<r(j)$. In other words, such permutations preserve the natural 
ordering of the numbers $1,\ldots ,n$ but only within each set $J\in{\cal J}$. 
The set of all order-preserving permutations forms a sub-group of $R$ which we denote by $P'$. 
\begin{lemma} Each equivalence class $p\in P$ contains exactly one 
order-preserving permutation in $P'$. \end{lemma}
 \begin{proof} For existence, consider an arbitrary representative $s\in p$ of 
an equivalence class $p\in P$ and consider the restriction $s|_J:J\rightarrow 
s(J)$ to any $J\in {\cal J}$.  We can find a permutation $r_J\in S(J)$ of $J$ 
such that $s|_J\circ r_J$ has the order-preserving property. The maps 
$r_J:J\rightarrow J$ combine into permutation a $r\in S_n$ which, from 
Theorem~\ref{thS} is an element of the stabilizer group $S$. Then, $s\circ r\in 
p$ and it is order-preserving.\\
 For uniqueness, consider two permutations $s,\tilde{s}\in p$ which are both 
order-preserving. This means than $r\equiv s\circ\tilde{s}^{-1}$ is in $S$ that 
is, $r|_J\in S(J)$ is a permutation of $J$. For $i,j\in J$ with $i<j$ we have 
$r(i)<r(j)$ which means that $r|_J={\rm id}_J$. This is true for all $J\in{\cal 
J}$ so that $r={\rm id}$ and, hence, $s=\tilde{s}$. $\Box$
 \end{proof}
This result means that we can define a group isomorphism $\tau:P\rightarrow 
P'\subset R$ which assigns to a class $p\in P$ the unique order-preserving 
permutation $\tau(p)\in p$. Of course, $[\tau (p)]=p$, so $\tau$ splits the 
sequence~\eqref{seq}. This means that $R$ is a semi-direct product
\begin{equation}
 R\cong P\ltimes S \label{sd}
\end{equation}
with the isomorphism $\varphi: P\times S\rightarrow R$ defined by $\varphi 
((p,s))=\tau(p)s$ and the multiplication in $P\ltimes S$ given by
\begin{equation}
 (p,s)(\tilde{p},\tilde{s})=(p\tilde{p},\tilde{p}^{-1}s \tilde{p}\tilde{s})\; .
\end{equation} 
Hence, we can determine $P\cong P'$ by finding the permutations in $S_n$ which map 
the sets ${\cal I}$ into each other (that is, which are elements of $R$) and are 
order-preserving. The group $R$ can then be obtained from $S$ and $P$ by forming 
the semi-direct product~\eqref{sd}.

\subsection{Invariance group of the upstairs space}
We now discuss the structure of $G_B$ of the upstairs space $B=\mathbb{C}^n-Z(\Sigma)$.
First note that there is an obvious embedding $S_n\hookrightarrow {\rm 
Gl}(n,\mathbb{C})$ defined by $\sigma ({\bf e}_i)={\bf e}_{\sigma (i)}$, with 
the standard unit vectors ${\bf e}_i$. Hence, the above subgroups $S,R,P'\subset 
S_n$ have isomorphic images in ${\rm Gl}(n,\mathbb{C})$ which we denote by 
$S_B,R_B,P_B$, respectively. In fact, from the definition of $S, R,P$ it is 
clear that $S_B,R_B,P_B$ leave the zero set $Z$ invariant and are, hence, 
sup-groups of $G_B$. Another obvious sub-group of $G_B$ is 
\begin{equation}
 H_B=\{g\in G_B\,|\, g(Z(I))=Z(I),\;\forall I\in{\cal I}\}\subset G_B\; ,
\end{equation}
that is, the sub-group which leaves the components $Z(I)$ of the zero set 
invariant individually.  Clearly, $S_B\subset H_B$. 
\begin{lemma} For $g\in G_B$ we have $g(Z(I))\in{\cal Z}$ for all $I\in {\cal 
I}$, that is, elements of $G_B$ map zero sets $Z(I)$ into other zero sets. 
\end{lemma}
\begin{proof} For $g\in G_B$, the image $g(Z(I))$ is a vector-space of the same 
dimension as $Z(I)$. But the only vector spaces contained in $Z$ are the $Z(I)$ 
and their sub-spaces, so $g(Z(I))$ must be contained in some $Z(I')$. Say that 
$Z(I')=g(Z(I))\oplus X$, for some vector space $X$. Then 
$g^{-1}(Z(I'))=Z(I)\oplus g^{-1}(X)$ which, by the same argument applied to 
$g^{-1}$, should be a subset of some $Z(\tilde{I})$. This means that 
$Z(I)\subset Z(\tilde{I})$ but, since we have dropped trivial zero sets which 
are already contained in others, it follows that $Z(I)=Z(\tilde{I})$. So we have 
$Z(I)\oplus g^{-1}(X)\subset Z(\tilde{I})=Z(I)$ which means that $g^{-1}(X)=0$ 
and, hence, $X=0$. Therefore, $g(Z(I))=Z(I')$. $\Box$
\end{proof}
\begin{lemma} The group $G_B$ can be decomposed as $G_B=P_B\,H_B$. Further, 
$H_B$ is a normal sub-group of $G_B$ and $P_B\cap H_B=1$. 
\end{lemma}
\begin{proof} For the first part of the statement, consider a $g\in G_B$ and its 
restriction $g_I:Z(I)\rightarrow Z(I')=g(Z(I))$ to the zero set $Z(I)$. We can 
find a permutation matrix $r_I:Z(I)\rightarrow Z(I')$ which maps the standard 
basis $\{{\bf e}_i\,|\,i\notin I\}$ of $Z(I)$ into the basis $\{{\bf 
e}_i\,|\,i\notin I'\}$ of $Z(I')$. Clearly, $\tilde{h}_I\equiv r_I^{-1}\circ 
g_I:Z(I)\rightarrow Z(I)$ and $g_I=r_I\circ\tilde{h}_I$. The maps $r_I$ can be 
made consistent on intersections of zero sets since $g(Z(I)\cap 
Z(I'))=g(Z(I))\cap g(Z(I'))$. Hence, the $r_I$ can be combined to a map $r\in 
Gl(n,\mathbb{C})$ with $r|_{Z(I)}=r_I$ which permutes the standard basis vectors 
${\bf e}_i$ and maps zero sets $Z(I)$ into each other, so $r\in R_B$. Further, 
defining $\tilde{h}=r^{-1}\circ g$ it follows that $\tilde{h}\in H_B$. From the 
previous section, we know that we can write $r=ps$, where $p\in P_B$ and $s\in 
S_B$. Finally, with $h\equiv s\tilde{h}\in H_B$ we have $g=ph$, where $p\in P_B$ 
and $h\in H_B$.\\
To show that $H_B$ is a normal sub-group of $G_B$, consider a $\eta\in H_B$ and 
a $g=ph\in G_B$. Then $g\eta g^{-1}=ph\eta h^{-1}p^{-1}=ph'p^{-1}$, where 
$h'=h\eta h^{-1}\in H_B$. Hence, $g\eta g^{-1}(Z(I))=ph'p^{-1}(Z(I))=Z(I)$, so 
$g\eta g^{-1}\in H_B$.\\
Finally, from the definition of $R$ it is clear that $R_B\cap H_B=S_B$. But 
$S_B$ contains precisely one order-preserving element, the identity. $\Box$
\end{proof}
The result of the previous Lemma shows that $G_B$ is a semi-direct product
$
 G_B\cong P\ltimes H_B
$.

If we denote by $\phi: P\rightarrow P_B$ the isomorphism between 
order-preserving permutations in $P$ and corresponding $Gl(n,\mathbb{C})$ 
matrices, then a pair $(p,h)\in P\ltimes H_B$ is mapped to $\phi(p)h\in G_B$ and 
the multiplication in $P\ltimes H_B$ is
\begin{equation}
 (p,h)(\tilde{p},\tilde{h})=(p\tilde{p},\phi(\tilde{p})^{-1}h\phi(\tilde{p}
)\tilde{h})\; .
\end{equation} 
We have already understood the permutation group $P$ and how to compute it from 
the basic data. It remains to consider $H_B$. To this end, we define the group
\begin{equation}
 G({\cal J})=\bigotimes_{J\in {\cal J}}{\rm Gl}(J,\mathbb{C})\; , \label{GJ}
\end{equation}
 where ${\rm Gl}(J,\mathbb{C})$ is the general linear group acting in the 
directions of the coordinates ${\bf x}_J\equiv (x_j\,|\,j\in J)$. Recall that 
the sets ${\cal J}$ were defined around Eq.~\eqref{equiv}. The elements 
of~\eqref{GJ} can be viewed as block-diagonal $n\times n$ matrices with the 
blocks corresponding to the coordinates ${\bf x}_J$. 
\begin{lemma} It follows that $G({\cal J})\subset H_B$. \end{lemma}
\begin{proof}
 From Eq.~\eqref{decomp} we know that every set $I\in{\cal I}$ can be written as 
a disjoint union $I=\bigcup J$ of some $J\in {\cal J}$. It follows that 
$Z(I)=Z(\cup J)=\cap Z(J)$. The transformations~\eqref{GJ} leave all $Z(J)$ and, 
hence, $Z(I)$ invariant so that $G({\cal J})\subset H_B$. $\Box$
 \end{proof}
Unfortunately, $H_B$ can be larger than $G({\cal J})$. In general we know that 
the zero sets $Z(I)$ and their intersections $Z(I_1)\cap\dots\cap 
Z(I_p)=Z(I_1\cup\dots\cup I_p)$ have to be invariant under transformations 
$h=(h_{ij})\in H_B$. This implies, for any union $I_1\cup\dots\cup I_p$, that
\begin{equation}
 h_{ij}=0\;\mbox{ for }\; i\in I_1\cup\dots\cup I_p\; ,\;\;j\in 
(I_1\cup\dots\cup I_p)^c\; . \label{hcond}
\end{equation}
We call a zero set ${\cal I}$ and associated set ${\cal J}$ {\em regular} if 
for every $J\in{\cal J}$ and all $i\in J$, $j\in J^c$ there exists an $I\in 
{\cal I}$ with $i\in I$ but $j\notin I$. Then we have:
\begin{theorem}
If the zero set ${\cal I}$ is regular then $G({\cal J})=H_B$.
\end{theorem}
\begin{proof}
Focus on a specific $J\in{\cal J}$ and fix a $i\in J$ and a $j\in J^c$. From 
regularity there exists an $I\in{\cal I}$ such that $i\in I$ and $j\in I^c$. 
Applying~\eqref{hcond} for $I$ and $I^c$ means that $h_{ij}=0$. Since $i$ and 
$j$ were arbitrary this means that $h_{ij}=0$ for all $i\in J$ and all $j\in 
J^c$.\\
Now consider an $i\in J^c$ and a $j\in J$. We know that $J^c$ is the disjoint 
union of the other $J$, so $i\in\tilde{J}\subset J^c$ for a particular 
$\tilde{J}\in{\cal J}$. Regularity, applied to $\tilde{J}$, means that there 
exists an $I\in{\cal I}$ with $i\in I$ and $j\in I^c$. As before, \eqref{hcond} 
applied to $I$ leads to $h_{ij}=0$. Since $i$, $j$ were arbitrary this means 
that $h_{ij}=0$ for all $i\in J^c$ and all $j\in J$.\\
In summary, for all $J\in{\cal J}$ we have $h_{ij}=0$ if $i\in J$ and $j\in J^c$ 
or if $i\in J^c$ and $j\in J$. This means that $h\in G({\cal J})$, so 
$H_B\subset G({\cal J})$. The opposite inclusion has already been shown in the 
previous Lemma, so $G({\cal J})=H_B$. $\Box$
 \end{proof}
 Therefore the complete symmetry group of $B=\mathbb{C}^n-Z(\Sigma)$ is given by
 \begin{equation}
  G_B\cong P\ltimes G(\mc{J})\; .
 \end{equation}

\subsection{Normalizer in ${\rm GL}(n,\mathbb{C})$}
In order to find the symmetry group $G_A$ of the toric space we need to compute 
the normalizer of the toric group ${\cal G}$ within $G_B$, the symmetry group of 
the upstairs space. This is relatively easy if the group elements of ${\cal G}$ 
are proportional to the unit matrix within each block~\eqref{GJ} of $G_B$. In order to sort this out 
it is instructive to first deal with a related - but simpler - problem, namely 
to compute the normalizer of ${\cal G}$ within ${\rm Gl}(n,\mathbb{C})$. 

To this end it is useful to split the various coordinate directions up into 
disjoint blocks ${\cal K}=\{K\,|\,K\subset \{1,\ldots ,n\}\}$, collecting the 
directions with the same charges ${\bf q}_i$, such that
\begin{equation}
 {\cal G}=\bigotimes_{K\in{\cal K}}{\cal G}_{{\bf q}_K}(K)
\end{equation}
where ${\cal G}_{{\bf q}_K}(K)=\{{\bf s}^{{\bf q}_K}\,{\bf 1}_{|K|}\,|\,{\bf 
s}\in(\mathbb{C}^*)^{n-d}\}$ consists of matrices proportional to the unit 
matrix and by ${\bf q}_K$ we denote the charge in the directions in $K$. As 
usual, we think of the symmetric group $S_n$ as being embedded into ${\rm 
Gl}(n,\mathbb{C})$ via $\sigma ({\bf e}_i)={\bf e}_{\sigma (i)}$ for $\sigma\in 
S_n$.  We would like to work out the normalizer
\begin{equation}
 N_{{\rm Gl}(n,\mathbb{C})}({\cal G})=\{g\in{\rm Gl}(n,\mathbb{C})\,|\, g{\cal 
G}={\cal G}g\}=
 \{g\in{\rm Gl}(n,\mathbb{C})\,|\, \forall\gamma\in{\cal G}\;\;\exists 
\tilde{\gamma}\in{\cal G}\,:\,g\gamma=\tilde{\gamma}g\}\; .
\end{equation} 
It is also useful to introduce the commutant
\begin{equation}
C_{{\rm Gl}(n,\mathbb{C})}({\cal G})=\{g\in{\rm Gl}(n,\mathbb{C})\,|\, 
g\gamma=\gamma g\;\forall\gamma\in{\cal G}\}\; ,
\end{equation}
of ${\cal G}$ within ${\rm Gl}(n,\mathbb{C})$ as well as the sub-groups of $S_n$ defined by
\begin{equation}
 {\cal R}=S_n\cap  N_{{\rm Gl}(n,\mathbb{C})}({\cal G})\; ,\quad {\cal 
S}=S_n\cap C_{{\rm Gl}(n,\mathbb{C})}({\cal G})\; .
\end{equation}
They consist of permutations which normalize or commute ${\cal G}$, respectively.  
We start by characterizing the groups ${\cal R}$ and ${\cal S}$ in a different 
way which is more useful for practical calculations.
\begin{lemma} i) ${\cal R}\subset\bar{\cal R}\equiv\{\sigma\in S_n\,|\, \sigma 
(K)\in{\cal K}\;\;\forall K\in{\cal K}\}$, ii) ${\cal S}=\{\sigma\in S_n\,|\, 
\sigma (K)=K\;\;\forall K\in{\cal K}\}$
\end{lemma}
\begin{proof} i) Consider a permutation $\sigma\in S_n$ which normalizes ${\cal 
G}$ so that for all $\gamma\in{\cal G}$ there exists a $\tilde{\gamma}\in{\cal 
G}$ such that $\gamma\sigma=\sigma\tilde{\gamma}$ with $\gamma={\rm diag}({\bf 
s}^{{\bf q}_i})$ and $\tilde{\gamma}={\rm diag}(\tilde{\bf s}^{{\bf q}_i})$. It 
follows that $\tilde{\bf s}^{{\bf q}_i}={\bf s}^{{\bf q}_{\sigma(i)}}$ for all 
$i$, so that $\tilde{\gamma}={\rm diag}({\bf s}^{{\bf s}_{\sigma(i)}})$. 
Consider all $i\in K$ for a given block $K$, so that ${\bf q}_i={\bf q}_K$. Then 
${\bf s}^{{\bf q}_{\sigma(i)}}=\tilde{\bf s}^{{\bf q}_K}$ for all $i\in K$. This 
only works if all ${\bf q}_{\sigma(i)}={\bf q}_{K'}$ are equal, so that 
$\sigma(i)\in K'$ for all $i\in K$. Hence, $\sigma(K)\subset K'\in{\cal K}$. 
Applying the same argument to $\sigma^{-1}$ and $K'$ leads to 
$\sigma^{-1}(K')\subset K$, so combined we have $\sigma (K)=K'$.\\
ii) Clearly, permutations which only act within blocks are in the commutant. 
Conversely, permutations in the commutant have to be block-diagonal from 
Schur's Lemma. $\Box$
\end{proof}
So, ${\cal S}$ is the set of permutations of directions within each block and 
$\bar{\cal R}$ contains permutations of blocks and within blocks. Unfortunately, 
the group ${\cal R}$ we are actually interested in can be genuinely smaller than 
$\bar{\cal R}$. If there are more blocks than variables ${\bf s}$ it might 
happen that a permutation of these blocks cannot be compensated by a choice of 
$\tilde{\bf s}$ in the normalizer condition. Clearly, ${\cal S}\subset{\cal R}$ 
is a sub-group and indeed a normal sub-group so that we can define the quotient 
group
\begin{equation}
 {\cal P}={\cal R}/{\cal S}\; .
\end{equation} 
There is a monomorphism ${\cal P}\rightarrow {\cal R}$ which involves assigning 
to an element of ${\cal P}$ the element of ${\cal R}$ which permutes the blocks 
in ${\cal K}$ in the same way and preserves the natural order in each block 
$K\in{\cal K}$. In this way,
\begin{equation}
 {\cal R}={\cal P}\ltimes{\cal S}\; .
\end{equation}
We are now ready for the following theorem.
\begin{theorem} The normalizer can be expressed in terms of the commutant as 
$N_{{\rm Gl}(n,\mathbb{C})}({\cal G})={\cal P}\,C_{{\rm Gl}(n,\mathbb{C})}({\cal 
G})$, $C_{{\rm Gl}(n,\mathbb{C})}({\cal G})$ is a normal sub-group of $N_{{\rm 
Gl}(n,\mathbb{C})}({\cal G})$ and ${\cal P}\cap C_{{\rm Gl}(n,\mathbb{C})}({\cal 
G})=1$.
\end{theorem}
\begin{proof} We begin by showing that every matrix $g=\in N_{{\rm 
Gl}(n,\mathbb{C})}({\cal G})$ can be written as a product of a permutation and a 
matrix in the commutant. First, $g=(g_{ij})\in N_{{\rm Gl}(n,\mathbb{C})}({\cal 
G})$ means that for all $\gamma\in{\cal G}$ there exists a $\tilde{\gamma}\in 
{\cal G}$ such that $g\gamma=\tilde{\gamma} g$. Write $\gamma={\rm 
diag}(\chi_i)$ and $\tilde{\gamma}={\rm diag}(\tilde{\chi}_i)$ where 
$\chi_i={\bf s}^{{\bf q}_i}$ and $\tilde{\chi}_i=\tilde{\bf s}^{{\bf q}_i}$. 
Then it follows that $g_{ij}(\tilde{\chi}_i-\chi_j)=0$ for all $i$, $j$. The 
matrix $g$ is invertible which means that for every row $i$ there exists a 
column $j$ with $g_{ij}\neq 0$ and, moreover, for each row we can choose a 
different column. This means there exists a permutation $\sigma\in S_n$ so that 
$g_{i\sigma^{-1}(i)}\neq 0$ for all $i$. In order to satisfy the normalized 
condition we then need that $\tilde{\chi}_i=\chi_{\sigma^{-1}(i)}$. Now define 
$c=\sigma^{-1}g$ so that $g=\sigma c$. The 
normalizer condition then reads $c\gamma 
c^{-1}=\sigma^{-1}\tilde{\gamma}\sigma={\rm diag}(\tilde{\chi}_{\sigma 
(i)})={\rm diag}(\chi_i)=\gamma$ and, hence, $c\in C_{{\rm 
Gl}(n,\mathbb{C})}({\cal G})$ is in the commutant. This shows that every $g$ in 
the normalizer can be written as $g=\sigma c$ with $\sigma\in S_n$ and $c$ in 
the commutant. With this decomposition, the normalizer condition, 
$g\gamma=\tilde{\gamma}g$ becomes $\sigma\gamma=\tilde{\gamma}\sigma$ which 
shows that, in fact, $\sigma\in {\cal R}$. Further we know that every $r\in 
{\cal R}$ can be written as $r=ps$, where $p\in{\cal P}$ and $s\in{\cal S}$, so 
that $g=rc=psc$. But both $s$, $c$ and their product $sc$ are in the commutant. 
This proves the decomposition  $N_{{\rm Gl}(n,\mathbb{C})}({\cal G})={\cal 
P}\,C_{{\rm Gl}(n,\mathbb{C})}({\cal G})$.\\
We now need to show that $C_{{\rm Gl}(n,\mathbb{C})}({\cal G})$ is a normal 
sub-group of $N_{{\rm Gl}(n,\mathbb{C})}({\cal G})$. Choose a $c$ in the 
commutant and $g=pc'$ in the normalizer. Then $gcg^{-1}=p\tilde{c}p^{-1}$ where 
$\tilde{c}=c'c{c'}^{-1}$ is in the commutant, so we need to show that 
$p\tilde{c}p^{-1}$ is also in the commutant, that is 
$p\tilde{c}p^{-1}\gamma=\gamma p\tilde{c}p^{-1}$ for all $\gamma\in{\cal G}$. 
This is equivalent to $\tilde{c}\tilde{\gamma}=\tilde{\gamma}\tilde{c}$ where 
$\tilde{\gamma}=p^{-1}\gamma p$. However, $\tilde{\gamma}$ has the same 
block-structure than elements in ${\cal G}$ (as the action of $p$ permutes the 
blocks) so that $\tilde{c}$ and $\tilde{\gamma}$ indeed commute.\\ 
Finally, since the action $p^{-1}\gamma p$ of a $p\in{\cal P}$ on 
$\gamma\in{\cal G}$ permutes the blocks of $\gamma$, it follows that the only 
$p\in {\cal P}$ for which $p^{-1}\gamma p =\gamma$ is, in fact, $p=1$, so that 
$N_{{\rm Gl}(n,\mathbb{C})}({\cal G})$ and ${\cal P}\cap C_{{\rm 
Gl}(n,\mathbb{C})}({\cal G})=1$. $\Box$
\end{proof}
So, in summary, we have found that the normalizer can be expressed in terms of 
the commutant as
\begin{equation}
 N_{{\rm Gl}(n,\mathbb{C})}({\cal G})={\cal P}\ltimes C_{{\rm 
Gl}(n,\mathbb{C})}({\cal G})\; ,
\end{equation}
and the latter,  by means of Schur's Lemma, can be written as
\begin{equation}
 C_{{\rm Gl}(n,\mathbb{C})}({\cal G})=\bigotimes_{K\in{\cal K}}{\rm 
Gl}(K,\mathbb{C})\; .
\end{equation} 
\subsection{Invariance group of the toric ambient space}
In order to find the symmetry group $G_A$ of the toric ambient space $A$ we have to compute the normalizer $N_{G_B}({\cal G})$. 
Clearly, this normalizer is given by
\begin{equation}
 N_{G_B}({\cal G})=G_B\cap N_{{\rm Gl}(n,\mathbb{C})}({\cal G})\; .
\end{equation} 
The two groups on the right-hand side have been explicitly determined in the 
previous two sections, so our remaining task is to express their intersection in 
a useful form. Both groups have a similar structure in that they relate to a 
block-decomposition of the $n$ coordinates and consist of a semi-direct product 
of a permutation group which exchanges the blocks times block-diagonal matrices 
which generate general linear transformations within each block. However, the 
block-decomposition is in general different for the two cases. The group $G_B$ 
relates to the block-decomposition ${\cal J}$ which follows from the structure 
of the zero set and is given by
\begin{equation}
 G_B=P\ltimes H_B\; ,\quad H_B=\bigotimes_{J\in{\cal J}}{\rm Gl}(J,\mathbb{C})\; 
.
\end{equation} 
The group $N_{{\rm Gl}(n,\mathbb{C})}({\cal G})$, on the other hand, relates to 
the block-decomposition ${\cal K}$ which follows from the toric group ${\cal G}$ 
and is given by
\begin{equation}
 N_{{\rm Gl}(n,\mathbb{C})}({\cal G})={\cal P}\ltimes C_{{\rm 
Gl}(n,\mathbb{C})}({\cal G})\; ,\quad
 C_{{\rm Gl}(n,\mathbb{C})}({\cal G})=\bigotimes_{K\in{\cal K}}{\rm 
Gl}(K,\mathbb{C})\; .
\end{equation} 
The complication is that the block-decompositions ${\cal J}$ and ${\cal K}$ are 
not necessarily identical. It is, therefore, useful to define the refined 
block-decomposition given as the intersection of ${\cal J}$ and ${\cal K}$ 
defined by\footnote{For fans obtained from triangulations of reflexive polytopes, ${\cal K}$ always is a refinement of ${\cal J}$ (see Appendix \ref{sect:setJ}).}
\begin{equation}
 {\cal L}=\{L=J\cap K\,|\,J\in{\cal J}\, ,\;K\in{\cal K}\, ,\; L\neq 
\emptyset\}\; .
\end{equation}
We introduce the usual collection of groups associated to this 
block-decomposition. First, there are the block-diagonal matrices in
\begin{equation}
 H_A=\bigotimes_{L\in{\cal L}}{\rm Gl}(L,\mathbb{C})\; .
\end{equation}
Then we have the stabilizer group, $S_{\cal L}$, and the block-permutation 
group, $R_{\cal L}$, defined by
\begin{equation}
 S_{\cal L}=\{\sigma\in S_n\,|\, \sigma (L)=L\;\;\forall L\in{\cal L}\}\; ,\quad
 R_{\cal L}=\{\sigma\in S_n\,|\, \sigma (L)\in{\cal L}\;\;\forall L\in{\cal 
L}\}\; .
 \end{equation}
 As usual, we can then form the quotient and the semi-direct product
 \begin{equation}
  P_{\cal L}=R_{\cal L}/S_{\cal L}\; ,\quad  R_{\cal L}=P_{\cal L}\ltimes 
S_{\cal L}\; .
 \end{equation}  
 From the definition of the various block structures is it clear that
 \begin{equation}
  H_A=H_B\cap C_{{\rm Gl}(n,\mathbb{C})}({\cal G})\; ,
\end{equation} 
and we expect this to be the continuous part of $N_{G_B}({\cal G})$. We also 
define the intersection of the permutation groups
\begin{equation}
 R_A=R\cap{\cal R}\; ,\quad S_A=S\cap{\cal S}\; ,\quad P_A=R_A/S_A\; ,\quad 
R_A=P_A\ltimes S_A\; .
\end{equation}
The two sets of discrete groups relate in the following way.
\begin{lemma} i) $R_A\subset R_{\cal L}$, ii) $S_A= S_{\cal L}$, iii) 
$P_A\subset P_{\cal L}$. \end{lemma}
\begin{proof} i) Consider a $r\in R\cap{\cal R}$ so that $r(J)\in{\cal J}$ for 
all $J\in{\cal J}$ and $r(K)\in{\cal K}$ for all $K\in{\cal K}$. Then, for any 
$L=J\cap K\in{\cal L}$ we have $r(L)=r(J\cap K)=r(J)\cap r(K)$ so that 
$r(L)\in{\cal L}$. Hence, $R\cap{\cal R}\subset R_{\cal L}$.\\
ii) ``$\subset$": For $s\in S\cap{\cal S}$ we have $s(J)=J$ for all $J\in{\cal 
J}$ and $s(K)=K$ for all $K\in{\cal K}$. Then, for all $L=J\cap K\in{\cal L}$ we 
have $s(L)=s(J\cap K)=s(J)\cap s(K)=J\cap K=L$, so that $s\in S_{\cal L}$.\\
``$\supset$": Let $s\in S_{\cal L}$, hence $s(L)=L$ for all $L\in{\cal L}$. 
Every $J\in{\cal J}$ can be written as a disjoint union $J=\bigcup_{L:L\cap 
J\neq\emptyset}L$ and likewise every $K\in{\cal K}$ as $K=\bigcup_{L:L\cap 
K\neq\emptyset}L$. Hence, since all $L$ are left invariant by $s$ so are all $J$ 
and $K$ and it follows that $s\in S\cap {\cal S}$.\\
iii) This follows directly from i) and ii).  $\Box$ 
\end{proof}

\begin{theorem} The normalizer group $N_{G_B}({\cal G})$ can be written as 
$N_{G_B}({\cal G})=R_AH_A$.\end{theorem}
\begin{proof} ``$\supset$": Consider a $g=rh\in R_AH_A$, where $r\in 
R_A=R\cap{\cal R}$ and $h\in H_A=H_B\cap C_{{\rm Gl}(n,\mathbb{C})}({\cal G})$. 
It follows that $rh\in RH_B=G_B$ and $rh\in{\cal R}C_{{\rm 
Gl}(n,\mathbb{C})}=N_{{\rm Gl}(n,\mathbb{C})}({\cal G})$ and, hence, $g=rh\in 
N_{G_B}({\cal G})$.\\
``$\subset$": Start with a $g\in N_{G_B}({\cal G})=G_B\cap N_{{\rm 
Gl}(n,\mathbb{C})}({\cal G})$ which can be written as 
$g=p\tilde{h}=\pi\tilde{c}$ with $p\in P$, $\tilde{h}\in H_B$, $\pi\in{\cal P}$, 
$\tilde{c}\in C_{{\rm Gl}(n,\mathbb{C})}({\cal G})$. It follows that 
$\pi^{-1}p=\tilde{c}\tilde{h}^{-1}\in C_{{\rm Gl}(n,\mathbb{C})}({\cal G})H_B$. 
Since $\pi^{-1}p$ is a permutation it follows that $\pi^{-1}p\in {\cal S}S$ 
which is the set of permutations in $C_{{\rm Gl}(n,\mathbb{C})}({\cal G})H_B$. Write 
$\pi^{-1}p=\sigma s^{-1}$, where $\sigma\in{\cal S}$ and $s\in S$, so that 
$ps=\sigma\pi\equiv\gamma$. Since $ps\in R$ and $\pi\sigma\in{\cal R}$ it 
follows that $\gamma\in R\cap{\cal R}=R_A$. Define $h=s^{-1}\tilde{h}\in H_B$ 
and $c=\sigma^{-1}\tilde{c}\in C_{{\rm Gl}(n,\mathbb{C})}({\cal G}$ so that 
$g=\gamma h=\gamma c$. It follows that $h=c\in H_B\cap C_{{\rm 
Gl}(n,\mathbb{C})}({\cal G}=H_A$ and, therefore $g=\gamma h\in R_AH_A$. $\Box$
\end{proof}
It is clear that we can write $N_{G_B}({\cal G})$ as a semi-direct product
\begin{equation}
 N_{G_B}({\cal G})=P_A\ltimes H_A
\end{equation}
by dividing out the group $S_A$. To find the invariance group, $G_A$, of the 
toric space we have to divide this by ${\cal G}$ which results in
\begin{equation}\label{eq:toricsymgroup}
 G_A=P_A\ltimes(H_A/{\cal G})\; .
\end{equation} 
This completes the calculation of $G_A$.

\subsection{Toric Varieties Obtained from Reflexive Polytopes}

The results outlined above hold for any toric variety constructed from a simplicial fan. Here, however, we are interested in toric varieties which are associated with triangulations of reflexive polytopes, which gives extra structure. Let us first recall the definition of Demazure roots, \eqref{liealgroots}. In the context of reflexive polytopes, it directly follows that only $x_i$ corresponding to vertices ${\bf v}_i$ can have Demazure roots, which are in turn given by all interior points\footnote{This is the reason for the appearance of one of the correction terms in Batyrev's formula for $h^{2,1}$ of a toric hypersurface Calabi-Yau threefold.} of the dual faces $\Theta^{[3]}_i$. Restricting to actions which act linearly on the homogeneous coordinates we recover the theorem stated in section \ref{sect:setK}. The set of all roots $\alpha$ for which $-\alpha$ is a root as well are known as semi-simple roots in the literature. 

\subsubsection{The set $\mathcal{K}$} \label{sect:setK}

Let us review the set $\mathcal{K}$. Its elements $K$ are sets 
containing a collection of indices $\{1,..,n\}$ such that ${\bf q}_i = {\bf q}_j$ for all 
$i,j$ in a specific $K$. As we discuss now, 
this has fairly strong implications for any two $i,j$ in the same $K$. 
\begin{theorem}
On a simply connected toric variety, any two lattice points ${\bf v}_i$ and ${\bf v}_j$ on 
$\Delta^\circ$ for which ${\bf q}_i = {\bf q}_j$
are vertices of $\Delta^\circ$. Furthermore, there exist points $\alpha_{ij}$ 
inside of $\Theta^{[3]}_i$ as well as $\alpha_{ji}$ inside of $\Theta^{[3]}_j$ for which $\alpha_{ij} = - 
\alpha_{ji}$.
\end{theorem}
\begin{proof}
From the equality of  ${\bf q}_i$ and ${\bf q}_j$ it follows that the 
corresponding divisors $D_i$ and $D_j$ are linearly equivalent. 
As all linear equivalences take the form
\begin{equation}
 \sum_{k=1}^n \langle \alpha_{ij},{\bf v}_k \rangle D_k = 0 \, 
\end{equation}
for some vector $\alpha_{ij}$ in $M \otimes \R$, there must be
a vector $\alpha_{ij}$ such that
\begin{align}\label{lineqcond}
\langle \alpha_{ij}, {\bf v}_i \rangle = -1 && \langle \alpha_{ij},{\bf v}_j \rangle = 1 && 
\langle \alpha_{ij}, {\bf v}_k \rangle = 0 \,\,\, \forall\,\, k \neq i,j 
\end{align}
For a simply connected toric variety, the set of all vectors ${\bf v}_i$ generate the 
lattice $N$. Hence 
a first consequence of this is that $\alpha_{ij}$ is contained in the dual 
lattice $M$. Furthermore, the above implies that
it is a lattice point on $\Delta$, as $\langle \alpha_{ij}, {\bf v}_k \rangle \geq -1 
$ for all points ${\bf v}_k$ on $\Delta^\circ$.
$\langle \alpha_{ij},{\bf v}_k \rangle = 0 \,\,\, \forall\,\, k \neq i,j $ implies 
that all ${\bf v}_k$ except ${\bf v}_i$ and ${\bf v}_j$ are contained
in a three-dimensional hyperplane of $N\otimes \R$ passing through the origin. 
As $\Delta^\circ$ is four-dimensional and
the convex hull of all of the ${\bf v}_l$, it follows that ${\bf v}_i$ and ${\bf v}_j$ must be 
above and below this hyperplane, so that they
must be vertices. In this case \eqref{lineqcond} implies that $\alpha_{ij}$ is 
inside the dual three-dimensional face of
${\bf v}_i$, $\Theta^{[3]}_i$ and  $\alpha_{ji} = - \alpha_{ij}$ is inside 
$\Theta^{[3]}_j$. $\Box$
\end{proof}

The above provides an alternative algorithm to determine the set 
$\mathcal{K}$. All indices $l$ corresponding to
non-vertices ${\bf v}_l$ form sets $K_l = \{l\}$. For any vertex ${\bf v}_i$, we have to 
find the interior points of the dual
face $\Theta^{[3]}_i$. If the negative of such an interior point is also inside 
of a three-dimensional face $\Theta^{[3]}_j$
(defining a dual vertex ${\bf v}_j$ of $\Delta^\circ$) there is a set $K \in 
\mathcal{K}$ containing both $i$ and $j$. This defines an
equivalence relation $\sim_\mathcal{K}$ and the equivalence classes are 
precisely the sets $K \in \mathcal{K}$. \\

This way of thinking is particularly useful for polytopes 
$\Delta^\circ$ with many integral points, where many points (all non-vertices) can never occur in any non-trivial set $K \in 
\mathcal{K}$.

\subsubsection{The set $\mathcal{J}$}\label{sect:setJ}

The data exploited in the last paragraph only depends on the structure of the 
polytopes $\Delta^\circ$ and $\Delta$. Let us now
consider the triangulation data. We have described the SR ideal by an index sets $I \in \mathcal{I}$ and have constructed the refinement $\mathcal{J}$ of $\mathcal{I}$ as the set of equivalence classes under the relation
\begin{equation}
 j \sim_\mathcal{J} i \,\, \leftrightarrow \,\, \mbox{$i$ and $j$ are contained 
in the same sets $I \in \mathcal{I} $} \, . 
\end{equation}
In other words, $\mathcal{J}$ is obtained by first forming all intersections 
between sets in $\mathcal{I}$ and then discarding 
sets which are already containing smaller sets.
\begin{lemma}
If all cones in a fan $\Sigma$ are symmetrical with respect to exchanging $i 
\leftrightarrow j$ then $i \sim_\mathcal{J} j$. 
\end{lemma}
\begin{proof}
To see this, all we have to do is go through the construction of $\mathcal{J}$ 
from the data of cones. We first find the SR ideal by forming 
the complement of the set of cones in the power set of $\{1, \cdots, n\}$. By 
construction this will be symmetrical under 
$i \leftrightarrow j$. But then also the set of generators $\mathcal{I}$ enjoys 
the same symmetry so that $i$ and $j$ are in the equivalence class in
$\mathcal{J}$ as claimed. $\Box$
\end{proof}

\begin{theorem}
If $D_i = D_j$, so that $i \sim_\mathcal{K} j$, it follows that $i 
\sim_\mathcal{J} j$.
\end{theorem}
\begin{proof}
Due to the lemma, it is sufficient to show that any triangulation of the 
polytope $\Delta^\circ$ is such that the collection of
cones is symmetric under the exchange $i \leftrightarrow j$. First of all, the 
triangulation $tr(\Delta^\circ)$ is completely
fixed in terms of its cones of maximal (i.e. 4) dimension. Every one of these 
cones is spanned by four linearly independent
vectors in $N$. If such a cone contains both ${\bf v}_i$ and ${\bf v}_j$ among its 
generators, there is nothing to show, as this is symmetric
in $i$ and $j$. A cone which just contains ${\bf v}_i$ but not ${\bf v}_j$ is hence spanned 
by ${\bf v}_i$ and three vectors ${\bf v}_{k_1},{\bf v}_{k_2},{\bf v}_{k_3}$ in the plane orthogonal to 
$\alpha_{ij}$. This means the cone in question has a face (which is a 
three-dimensional cone) spanned only by ${\bf v}_{k_1},{\bf v}_{k_2},{\bf v}_{k_3}$.
For a fan constructed from a triangulation, every three-dimensional cone is the 
intersection of precisely two four-dimensional cones. Hence
there must be another four-dimensional cone spanned by ${\bf v}_{k_1},{\bf v}_{k_2},{\bf v}_{k_3}$ 
and another vector. The only vector which turns ${\bf v}_{k_1},{\bf v}_{k_2},{\bf v}_{k_3}$ into a 
four-dimensional cone (besides ${\bf v}_i$) is ${\bf v}_j$. Hence there is a cone $\{ 1,k_1, 
k_2, k_3 \}$ as well. $\Box$
\end{proof}

As a direct consequence of the above, $\mathcal{K}$ is a refinement of 
$\mathcal{J}$ (the converse is wrong, i.e. $\mathcal{K} 
\neq \mathcal{J}$ in general) for any triangulation. This means we do not need to consider the index set 
$\mathcal{L} = \mathcal{K} \cap \mathcal{J}$ since it equals $\mathcal{K} $. Also ${\cal K}$  is ultimately 
independent of the triangulation chosen for $\Delta^\circ$.

\section{Toric $\pi$-twisted Representations}
\label{app:toricpitwrep}

In this appendix, we collect some results on the representation theory of finite groups $\Gamma$ acting linearly on the 
homogeneous coordinates of toric varieties. Due to the decomposition \eqref{eq:toricsymgroup} for the toric symmetry groups $G_A$, 
our representations $R:\Gamma\rightarrow G_A$ can similarly be split into a permutation representation $\pi:\Gamma\rightarrow P_A$ and a map $r:\Gamma\rightarrow H_A/{\cal G}$ into the continuous part of the ambient space symmetry group. Since this continuous part has the block structure 
 $H_A/{\cal G} = \bigoplus_{K \in {\cal K}} {\rm Gl}(K,\C)/{\cal G}$ the map $r$ can be split up into the corresponding blocks which we denote by $r_i$. The semi-direct product structure of the ambient space symmetry group $G_A$ means that $r$ is a $\pi$-representation, that is, it satisfies
 \begin{equation}
  r(\gamma\tilde{\gamma})=\pi(\tilde{\gamma})^{-1}r(\gamma)\pi(\tilde{\gamma})r(\tilde{\gamma})\; . \label{rpitwist}
\end{equation}

\subsection{Toric representations}

Let us first discuss the case where $\pi$ is trivial so that $r:\Gamma\rightarrow H_A/{\cal G}$ is a regular group homomorphism.

\subsubsection{Schur covers}\label{sect:schurcoversstandard}

We begin with the review of Schur covers, see \cite{aschbacher2000finite,rotman2012introduction} for a proper introduction. The central idea it to construct projective representations
\begin{equation}
\bar{r}: \Gamma \rightarrow {\rm Gl}(d,\mathbb{C}) / \mathbb{C}^*\, ,
\end{equation}
via linear representations 
\begin{equation}
 \hat{r}:  \hat{\Gamma} \rightarrow {\rm Gl}(d,\mathbb{C})
\end{equation}
of a larger group $\hat{\Gamma}$, called the Schur cover of $\Gamma$. The Schur cover is a central extension of $\Gamma$ by the Schur multiplier  $\kappa=H^2(\Gamma,\mathbb{C}^*)$, that is, there is a short exact sequence
\begin{equation}\label{eq:schurextension}
 1 \rightarrow \kappa \rightarrow \hat{\Gamma} \rightarrow \Gamma \rightarrow 1 \, ,
\end{equation}
We will discuss the definition of group cohomology below. The extension sequence \eqref{eq:schurextension} can be defined by fixing a so-called factor set 
$e:\Gamma \times \Gamma \rightarrow K$ via the multiplication rule
\begin{equation}
 (k,\gamma)(\tilde{k},\tilde{\gamma}) = 
(k\tilde{k} e(\gamma,\tilde{\gamma}),\gamma \tilde{\gamma}) \, .
\end{equation}

For a given extension \eqref{eq:schurextension} and projective representation $\bar{r}$, let $l(\gamma)$ be a lift 
of $\bar{r}$ to ${\rm Gl}(d,\mathbb{C})$ which satisfies
\begin{equation}\label{eq:ffactset}
 l(\gamma) l (\tilde{\gamma}) = f(\gamma,\tilde{\gamma}) l(\gamma\tilde{\gamma}) 
\, .
\end{equation}
Here, $f$ is map $f:\Gamma \times \Gamma \rightarrow \mathbb{C}^*$. The factor set $f$ captures the departure of the projective representation $\bar{r}$ from a linear one. Note that two such lifts $l,l'$ are equivalent if we can find a 
$h:\Gamma \rightarrow \mathbb{C}^*$ such that 
\begin{equation}
h(\gamma)l(\gamma) = l'(\gamma) 
\end{equation}
By Eq.~\eqref{eq:ffactset}, the factor sets are then related by
\begin{equation}\label{eq:equivf}
f'(\gamma,\tilde{\gamma}) = f(\gamma,\tilde{\gamma}) 
h(\gamma)h(\tilde{\gamma})h^{-1}(\gamma\tilde{\gamma})
\end{equation}
These relations define the group cohomology. Two factor sets $f$, $f'$ are in the same class $[f]=[f']$ in $H^2(\Gamma,\mathbb{C}^*)$ if they are related to equivalent lifts $l$ and $l'$.

One can define the map $\delta: {\rm Hom}(K,\mathbb{C}^*)\rightarrow H^2(\Gamma,\mathbb{C}^*)$ by $\delta(\varphi)=[\varphi\circ e]$. Now suppose that the class $[f]\in H^2(\Gamma,\mathbb{C}^*)$, associated to the lift~\eqref{eq:ffactset}, is in the image of $\delta$. This means there exists a $\varphi\in {\rm Hom}(K,\mathbb{C}^*)$ such that $[\varphi\circ e]=[f]$ or, dropping cohomology classes, there exists a function $h:\Gamma\rightarrow \mathbb{C}^*$ such that
\begin{equation}
  \varphi\circ e(\gamma,\tilde{\gamma})=f(\gamma,\tilde{\gamma})h(\tilde{\gamma})h(\gamma\tilde{\gamma})^{-1}h(\gamma)\; .
\end{equation}
Then defining a map $\hat{r}:\hat{\Gamma}\rightarrow H_A$ by
\begin{equation}
 \hat{r}(k,\gamma)=\varphi(k)h(\gamma)l(\gamma)
\end{equation}
gives us a linear representation which becomes the projective representation $\bar{r}$ in ${\rm Gl}(d,\mathbb{C}) / \mathbb{C}^*$.\\

The question is now if we can `undo' the factor set $f$ by enlarging the group $\Gamma$ to $\hat{\Gamma}$ and studying linear representations of 
$\hat{\Gamma}$. As discussed in the last section, if we can find a map
$\varphi \in {\rm Hom}(K,\mathbb{C}^*)$ (that is, $\varphi$ can be used to map 
$e(\gamma,\tilde{\gamma}) \rightarrow f(\gamma,\tilde{\gamma})$ and
 we can think of $\varphi(k)$ as a being proportional to the unit matrix) and 
 $h(\gamma)$ such that
 \begin{equation}\label{eq:deltacond}
  \varphi(e(\gamma,\tilde{\gamma})) =  f(\gamma,\tilde{\gamma}) 
 h(\gamma)h(\tilde{\gamma})h^{-1}(\gamma\tilde{\gamma})
 \end{equation}
 holds for every factor set $f(\gamma,\tilde{\gamma})$ our choice of 
 $e(\gamma,\tilde{\gamma})$ is sufficiently general. The central statement
 is now that the Schur cover obtained by using the Schur multiplier is precisely such that 
 this holds. For the choice $K = H^2(\Gamma,\mathbb{C}^*)$, $\delta$ is indeed surjective, and any projective representation can be lifted to a linear one. Hence, using the Schur cover we can lift every projective representation to a linear one.

\subsubsection{Multiple Schur covers} \label{sect:schurcoverstoric}

Let us now discuss the next more complicated cases, which are representations on products of projective spaces
\begin{equation}\label{eq:prodprojspacerep}
\bar{r}: \Gamma \rightarrow \bigotimes_i {\rm Gl}(d_i,\mathbb{C})/\mathbb{C}^* \, .
\end{equation}
First of all, if we are given such a representation and a lift $l(\gamma)$, 
\eqref{eq:ffactset} implies that $f$ is a map $f:\Gamma \times \Gamma 
\rightarrow (\mathbb{C}^*)^n$. We can think of $f$ 
as a matrix
\begin{equation}
 f(\gamma,\tilde{\gamma}) = \bigoplus \mathbb{I}_{d_i \times d_i} 
c_i(\gamma,\tilde{\gamma})\; ,
\end{equation}
that is, it is block diagonal and each block is proportional to the unit matrix. We 
can of course rewrite
\begin{align}\label{eq:decompf}
f(\gamma,\tilde{\gamma}) & = \prod_i f_i(\gamma,\tilde{\gamma}) \\
f_i(\gamma,\tilde{\gamma}) & =  \mathbb{I}_{d_i \times d_i} 
c_i(\gamma,\tilde{\gamma}) \oplus_{j \neq i}  \mathbb{I}_{d_j \times d_j} 
\end{align}
where now $f_i:\Gamma \times \Gamma \rightarrow \mathbb{C}^*$.

Clearly, equivalent $f(\gamma,\tilde{\gamma})$ are again related by 
\eqref{eq:equivf}. This equivalence is derived using the action of
$\mathbb{C}^*$ on the homogeneous coordinates and we may consider each of 
the $n$ actions in turn and denote the corresponding map
$h_i:\Gamma \rightarrow \mathbb{C}^*$ $h_i(\gamma)$. This gives rise to a 
similar decomposition into $h(\gamma) = \prod h_i(\gamma)$ 
with $h_i:\Gamma \rightarrow \mathbb{C}^*$.  With this we find that 
$f(\gamma,\tilde{\gamma}),f'(\gamma,\tilde{\gamma})$ define the same class in 
$H^2(\Gamma,(\mathbb{C}^*)^n)$ if all of the $f_i$ and $f_i'$ define the same 
class in $H^2(\Gamma,\mathbb{C}^*)$
\begin{equation}
f_i'(\gamma,\tilde{\gamma}) = f_i(\gamma,\tilde{\gamma}) 
h_i(\gamma)h_i(\tilde{\gamma})h_i^{-1}(\gamma\tilde{\gamma})\; .
\end{equation}

We now use the same Schur cover as before, (that is, we have the same exact sequence 
\eqref{eq:schurextension} and
$e:\Gamma \times \Gamma \rightarrow K$), and consider maps $\varphi$, which now 
maps $\varphi: \mathbb{K}\rightarrow (\mathbb{C}^*)^n$. 
As is $f(\gamma,\tilde{\gamma})$, $\varphi(k)$ will be a block-diagonal matrix 
with each block proportional to the unit matrix.
We can hence also decompose $\varphi$ into homomorphisms $\varphi_i: K 
\rightarrow \mathbb{C}^*$,
\begin{align}
\varphi(k) & = \prod_i \varphi_i(k) \\
\varphi_i(k) & =  \mathbb{I}_{d_i \times d_i} p_i(k) \oplus_{j \neq i}  
\mathbb{I}_{d_j \times d_j} 
\end{align}

We are now ready to confront \eqref{eq:deltacond} for this setup. As we have 
decomposed $\varphi,h,f$ in the same way, a solution
to the generalization of \eqref{eq:deltacond} can be found if we can solve
\begin{equation}
 \varphi_i(e(\gamma,\tilde{\gamma})) =  f_i(\gamma,\tilde{\gamma}) 
h_i(\gamma)h_i(\tilde{\gamma})h_i^{-1}(\gamma\tilde{\gamma})
\end{equation}
The existence of such $h_i,\varphi_i$ for each $f(\gamma,\tilde{\gamma})$ is already implied if $e(\gamma,\tilde{\gamma})$ originates from the standard Schur cover, so that we are done and the same construction applies here. The linear representation is found as
\begin{equation}
\hat{r}(k,\gamma)=\varphi(k)h(\gamma)l(\gamma)  
\end{equation}

The case of a toric variety may now be discussed in a similar fashion. The difference to having independent products of projective linear groups as in Eq.~\eqref{eq:prodprojspacerep} is that we may have more factors of ${\rm Gl}(d_i,\C)$ than $\C^*$ actions in \eqref{eq:prodprojspacerep}. Alternatively, we may think of the $\C^*$ actions on each of the ${\rm Gl}(d_i,\C)$ `blocks' as being correlated. Let us assume we have in total $b$ blocks, with $c$ blocks correlated to the other $(b-c)$ blocks, so that $\mathcal{G} \cong (\mathbb{C}^*)^{b-c}$. 
We can find all multi-projective representations in this case by treating all of the blocks as uncorrelated and then checking two more conditions to remove the invalid ones. A sufficient condition for a linear representation to descend to a multi-projective representation in the case of correlated blocks is:
\begin{equation}
\hat{r}(k,\gamma)\hat{r}(\tilde{k},\gamma)^{-1} \in \mathcal{G}
\end{equation}
This reason for this is that given the same $\gamma$, varying $k$ does not change the 
identification of $\bar{r}(\gamma)$ inside $\hat{r}(k,\gamma)$. Denoting the projection from $\hat{\Gamma}$ to $\Gamma$ by $\hat{\pi}$
we have that $\hat{\pi}(k,\gamma) = \hat{\pi}(\tilde{k},\gamma)$ for all $k,\tilde{k} \in K$. Hence by the definition of $\hat{\pi}$ we have 
$\hat{r}(k,\gamma)\hat{r}(\tilde{k},\gamma)^{-1} \in \mathcal{G}$.\\

Finally, we need to identify under which circumstances two different linear 
representations of $\Gamma$ give rise to the same multi projective representation.
Two projective representations $\bar{r}_1$ and $\bar{r}_2$ are equivalent only 
when there exists a automorphism $\mathcal{P}$ of the space of homogeneous coordinates 
and a function $\theta : \Gamma \rightarrow \mathcal{G}$, such that for each $\gamma \in \Gamma$, we have 
$\mathcal{P}(\bar{r}_1(\gamma)) = \theta(\gamma)(\bar{r}_2(\gamma)\mathcal{P})$.
Inserting the definition of $\bar{r} = \hat{\pi}(\hat{r}(k,\gamma))$, this means that
\begin{equation}
\mathcal{P}(\hat{\pi}(\hat{r}(k_1,\gamma)) = 
\theta(\gamma)(\hat{\pi}(\hat{r}_2(k_2,\gamma))\mathcal{P})\; .\label{equiprojrepr}
\end{equation}
We see that choosing different $k_1$ and $k_2$ is equivalent to inserting
another factor into $\theta$, namely $\theta$ is $k_1,k_2$ dependent. However in 
the sense of equivalent representation, we could just make one convenient choice 
of $k_1,k_2$ for each $\gamma$, checking the formula (\ref{equiprojrepr}). This is 
sufficient to remove the redundancies.

\subsection{Twisted Representations}

Let us now discuss the case of a non-trivial representation $\pi$ and the resulting interplay between $\pi$ and $r$. Before considering the full problem of finding all (multi-)projective  $\pi$-twisted representations, we discuss how we can manage to find all $\pi$-twisted linear representations $r:\Gamma\rightarrow P_A \ltimes \bigoplus_{K \in {\cal K}} {\rm Gl}(K,\C)$. 
To start the construction, let us assume that we have found a, not necessarily injective, group homomorphism
\begin{equation}
 \pi: \Gamma \rightarrow P_A \; .
\end{equation}
How can we find all compatible $\pi$-twisted representations $r$ satisfying \eqref{rpitwist}? 
Furthermore, we assume for simplicity that $\pi(\Gamma)$ acts transitively on the 'blocks' $K_i\in{\cal K}$ (of course all of these have to have the same size) and leaves everything else unchanged. We will discuss the general case below.

Labelling the relevant blocks by $K_i$, where $i=1,\ldots,b$, we can single out the block $K_i$  
and consider its stabilizer groups $\Gamma_i \subset \Gamma$ under the action of $\pi$. We denote by $r_i:\Gamma_i\rightarrow GL(B_i,\mathbb{C})/\mathbb{C}^*$ the map $r$ restricted to the block $K_i$. Equivalently, we can write the entire representation $R$ as the matrix product
\begin{equation}
 R(\gamma) = \pi(\gamma) \, \cdot \,\mbox{diag}(r_1(\gamma),\cdots, r_n(\gamma) \, .
\end{equation}
Note that the restriction $\tilde{r}_i$ and $r_i$ to the stabiliser group $\Gamma_i$ defines a group homomorphism (since, by definition, $\pi$ acts trivially on the stabiliser). We now explain that it is sufficient to fix a single one of the representations $\tilde{r}_i$ to recover the whole $\pi$-twisted representation $r$. For definiteness, let us hence consider $\Gamma_1$, the stabilizer of the first block, and fix the representation
$\tilde{r}_1 : \Gamma_1 \mapsto {\rm Gl}(K_1,\mathbb{C})$. To reconstruct the whole action of $\Gamma$ we pick a set of group
elements $\gamma_i$ such that $\pi(\gamma_i)(1) = i$. For any block $i$, we can then write an arbitrary group element $\gamma\in\Gamma$ uniquely as
\begin{equation}\label{pirmatrix}
\gamma = \gamma_{\pi(\gamma)(i)} h \gamma_{i}^{-1} \, ,
\end{equation}
where $h\in\Gamma_1$. We can hence think of $h$ to depend on $\gamma$ and $i$.
To see this, note that we can choose $h$ as
\begin{equation}\label{hfromiandg}
 h = \gamma_{\pi(\gamma)(i)}^{-1} \gamma \gamma_{i} \, . 
\end{equation}
This is indeed an element of $\Gamma_1$ since
\begin{equation}
\begin{aligned}
 \pi(h)(1)& =  \pi(\gamma_{\pi(\gamma)(i)})^{-1} \pi(\gamma) \pi(\gamma_{i})(1) \\
   &= \pi(\gamma_{\pi(\gamma)(i)})^{-1} \pi(\gamma)(i) \\
   &= 1\; .
 \end{aligned}
\end{equation}

We may hence write 
\begin{equation}
 R(\gamma) = R(\gamma_{\pi(\gamma)(i)}) R(h) R(\gamma_i^{-1})
\end{equation}
for any $i$ and $\gamma$ and using the appropriate $h\in\Gamma_1$. Going through the definitions this means that we have
\begin{equation}\label{rifromr1}
r_i(\gamma) = \tilde{r}_1(h) 
\end{equation}
This allows us to recover all of the matrices $r_i(\gamma)$ and hence the entire $\pi$-twisted representation from $\tilde{r}_1$. The group homomorphisms $\tilde{r}_1$ can, in turn, be constructed from the linear representations of the Schur cover $\hat{\Gamma}_1$, as discussed above.

In general $\pi:\Gamma \rightarrow P_A$ does not act transitively on the blocks and there may be several orbits.
In this case, we can apply the above method separately for each of the orbits and then combine all of the data to find the
representation $R:\Gamma\rightarrow P_A \ltimes H_A$.

In summary, in order to find the representations $R=(\pi,r):\Gamma\rightarrow P_A \ltimes H_A$, we have to carry out the following steps. 
\begin{itemize}
 \item[1)] Find all permutation representations $\pi:\Gamma\rightarrow P_A$, not necessarily faithful.
 \item[2)] Find all orbits of the blocks under the action of $\pi$.
 \item[3)] For each orbit $\{K_i\}$, pick out the block $K_1$ and determine its stabiliser group, $\Gamma_1$, under the action of $\pi$.
 \item[4)] Choose the group elements $\gamma_i$ such that $\pi(\gamma_i)(1) = i$.
 \item[5)] Study all representations $\tilde{r}_1: \Gamma_1\rightarrow  {\rm Gl}(K_1,\mathbb{C})/\mathbb{C}^*$ by considering the Schur cover, $\hat{\Gamma}_1$ of $\Gamma_1$. 
 \item[6)] Re-construct the entire map $r$ from $\tilde{r}_1$ by using Eq.~\eqref{hfromiandg}. 
  \item[7)] Repeat this process for each orbit, assemble the results into the complete representation $R$ and check if $R$ is faithful.
\end{itemize}

\section{Patches on Toric Varieties and Smoothness of Hypersurfaces}\label{sect:patches}
Let us describe how the coordinate patches corresponding to each four-dimensional cone $\sigma\in\Sigma$ of a toric variety with fan $\Sigma$ are found. Let us start with cones corresponding to smooth affine varieties, in which case the corresponding patch is simply $\C^4$. In this case, the rays of
$\sigma$ are generated by four lattice vectors ${\bf v}_i^\sigma$ such that the matrix $({\bf v}_1^\sigma,{\bf v}_2^\sigma,{\bf v}_3^\sigma,{\bf v}_4^\sigma)$ has determinant one and the ${\bf v}_i$ generate the N-lattice with integer coefficients. The four rays of the dual cone\footnote{The dual cone is defined by $\langle \sigma^\vee, \sigma \rangle\geq 0$.} $\sigma^\vee$ are then generated by by four lattice vectors ${\bf v}_i^\vee$ which generate the M-lattice and satisfy 
\begin{equation}
 \langle {\bf v}_j^\vee,{\bf v}_i^\sigma \rangle =\delta_{ij} \, .
\end{equation}
The affine toric variety $V(\sigma) = (\C^*)^4$ associated with each smooth four-dimensional cone has a stratification
\begin{equation}\label{eq:sigmastrat}
V(\sigma) = (\C^*)^4 \amalg_{i=1 .. 4} (\C^*)^3  \amalg_{j=1 .. 6} (\C^*)^2  \amalg_{k=1 .. 4} (\C^*) \amalg pt 
\end{equation}
in which each stratum $(\C^*)^n$ is associated with an $4-n$ dimensional cone contained in $\sigma$. We can find good coordinates 
for all of these strata using points on the $M$-lattice. In the present case, the ${\bf v}_i^\vee$ defined above give rise to coordinates on $(\C^*)^4$ which have a well-defined limit to the other strata in \eqref{eq:sigmastrat}, that is, we can use them as coordinates for the whole of $V(\sigma)$. These coordinates can be written as
\begin{equation}
\hat{x}_i^\sigma \equiv  \prod_{j} (x_j)^{\langle{\bf v}_i^\vee, {\bf v}_j\rangle} = x_i^\sigma \prod_{\{x_j\} \setminus \{x_i^\sigma\}} x_j^{\langle {\bf v}_i^\vee,{\bf v}_j \rangle} \, ,
\end{equation}
where we denote the homogeneous coordinates corresponding to the ${\bf v}_i^\sigma$ by $x_i^\sigma$ and keep the notation of denoting all 
homogeneous coordinates simply by $x_i$. In the above expression, the first product ranges over all homogeneous coordinates (ray generators of the fan $\Sigma$) and the second one only over homogeneous coordinates associated with rays not contained in $\sigma$.

On the patch $V(\sigma)$, we can rewrite the defining polynomial $p$ of a TCY hypersurface $X$ in terms of the local coordinates $\hat{x}_i$. Every monomial $p_{\bf m}$ in $p$ is associated with a vector ${\bf m}$ in the M-lattice and we may write ${\bf m} = \sum_{i=1..4} m_i {\bf v}_i^\vee$ as the ${\bf v}_i^\vee$ generate the M-lattice. Hence
\be
\begin{aligned}
p_m &= \prod x_j^{\langle {\bf m} , {\bf v}_j\rangle +1 } = \prod_j x_j \prod_i x_j^{m_i \langle {\bf v}_i^{\vee},{\bf v}_j \rangle}  \\
    &= \prod_i (\hat{x}_i^\sigma)^{m_i} \left[\prod_j x_j \right]\\
    &= \prod_i  (\hat{x}_i^\sigma)^{\langle {\bf m},{\bf v}_i^\sigma \rangle+1} \left[ \prod_{\{x_j\} \setminus \{x_i^\sigma\}}x_j\right] 
\end{aligned}
\ee
Not surprisingly, we can rewrite every monomial purely in terms of the $\hat{x}_i$ such that $p(x_i^\sigma)= p(\hat{x}_i^\sigma)$ when we set all
other coordinates $=1$. Let us take $p^\sigma $ to be the polynomial which is obtained when we set 
$\{x_j\} \setminus \{x_i^\sigma\}$ to $1$. This is of course the same as 'gauge fixing' all coordinates $x_i$ except the local coordinates $x_i^\sigma$. 
For each patch $V(\sigma)=\C^4$ we can define the Jacobi ideal of the local defining polynomial $p^\sigma$ by
\begin{equation}
I^{\sigma} = \langle p^\sigma, \frac{\partial p^\sigma}{\partial x_1^\sigma} ,\cdots , \frac{\partial p^\sigma}{\partial x_4^\sigma} \rangle\; .
 \end{equation}
To check smoothness on $V(\sigma)$ we have to verify that the dimension of this ideal is $-1$. 

Let us now discuss the singular patches $V(\sigma)$. Naively, singular patches arise because gauge fixing all coordinates except
the ones in $\sigma$ may leave some residual discrete group $G\subset{\cal G}$ which still has to be divided out. Hence, the local patch has the structure $V(\sigma)=\C^4/G$ and the local coordinates $x_i^\sigma$ really parametrize the covering space $\C^4$. As before, we may find coordinates on $(\C^*)^4 \subset V(\sigma)$ by considering four vectors generating the M-lattice. What we cannot hope for, however, is for the ray generators of $\sigma_i^\vee$ to generate the M-lattice (over $\Z$). We can however extend these to a set of generators of the M-lattice by considering a larger set $\{{\bf v}_\mu^\vee\} \supset \{{\bf v}_i^\vee\}$ . 
These can be used to define a number of coordinates on $V(\sigma)$
\begin{equation}\label{eq:defxmuhat}
\hat{x}_\mu^\sigma \equiv  \prod_{i} (x_i)^{\langle {\bf v}_\mu^\vee, {\bf v}_i\rangle} \, .
\end{equation}
As we can expand any vector of the $m$-lattice in terms of the $ {\bf v}_\mu^\vee$ with integer coefficients, we can again
rewrite $p$ as a function solely of the ${\bf v}_\mu^\vee$ upon dividing out $\prod_{\{x_j\} \setminus \{x_i^\sigma\}}x_j$.
(This also follows more abstractly from the fact that $\C[{\bf v}_\mu^\vee]$ is the polynomial ring of $V(\sigma)$.)
However, we have more than four ${\bf v}_\mu^\vee$, so that there are a number of linear relations 
\begin{equation}
 \sum_\mu \alpha_\mu {\bf v}_\mu^\vee =0
\end{equation}
with $\alpha_\mu \in \Z$. These turn into relations
\begin{equation}\label{eq:binomialrelations}
\prod_\mu \hat{x}_\mu^{\alpha_\mu} =1 \, .
\end{equation}
Hence we now find a complete intersection as the appropriate expression for $p^\sigma$. The additional defining equations capture the orbifold singularities on $V(\sigma) = \C^4/G$. \\

As a simple example, let us consider a Calabi-Yau (elliptic curve) hypersurface of $\P_{123}$. We take homogeneous coordinates
$(y,x,z) \sim (\lambda^3 y, \lambda^2 x, \lambda z )$. The fan of $\P_{123}$ has rays generated by 
\begin{equation}
{\bf v}_y = \left(\begin{array}{c}
  1 \\ 0
 \end{array}\right) \hspace{1cm}
{\bf v}_x = \left(\begin{array}{c}
  0 \\ 1
 \end{array}\right)\hspace{1cm}
{\bf v}_z = \left(\begin{array}{c}
  -3 \\ -2
 \end{array}\right)
\end{equation}
The hypersurface can be described by
\begin{equation}\label{eq:ellcurve}
 y^2 = x^3 + fxz^4 +g z^6
\end{equation}
Let us discuss the cone $\sigma$ spanned by $x$ and $z$. Gauge fixing $y=1$ leaves a residual $\Z_3$ acting on $x$ and $z$ so that
we expect $V(\sigma) = \C^2/\Z_3$. The generators of $\sigma^\vee$ are
\begin{equation}
 \left(\begin{array}{c}
  -1 \\ 0 
 \end{array}\right) \rightarrow \hat{z} = z^3/y \hspace{1cm}
 \left(\begin{array}{c}
  -2 \\ 3 
 \end{array}\right) \rightarrow \hat{x} = x^3/y^2  
\end{equation}
Clearly, we cannot rewrite \eqref{eq:ellcurve} in terms of these coordinates alone, which corresponds to the fact that the above vectors
do not generate the M-lattice by only a sublattice $M'$ such that $M/M' =\Z_3$. We can cure this by defining a new coordinate
\begin{equation}
 \left(\begin{array}{c}
  1 \\ -1 
 \end{array}\right) \rightarrow  \hat{\xi} \equiv x z /y 
\end{equation}
so that \eqref{eq:ellcurve} becomes (after dividing by $y^2$)
\begin{equation}
 1 = \hat{x} + f \hat{\xi} \hat{z} + g \hat{z}^2 = 0
\end{equation}
However, we now have the additional relation
\begin{equation}
 \hat{\xi}^3 = \hat{z} \hat{x}
\end{equation}
which is nothing but the defining equation of an $A_3$ singularity (which is the same as $\C^2/\Z_3$) embedded into $\C^3$.\\

As we have already discussed, we have to make sure our hypersurface $X$ defined by $p=0$ must not meet the singularities of the ambient space as this would lead to singularities on $X$ as well. As the singularities arise through fixed points of the group action of $G$ on $\C^4$, we can lift $X \subset V(\sigma)$ to a (different) hypersurface $\tilde{X}$ in $\C^4$ which descends to $X$ via a free group action of $G$. Smoothness of the covering space is then equivalent to smoothness of the quotient. We are hence interested in studying coordinates on the cover of $V(\sigma)$. These coordinates can be found by writing the $\hat{x}_\mu$ as monomials such that \eqref{eq:binomialrelations} is automatically solved. This, however, can be done by simply using \eqref{eq:defxmuhat} and setting all $\{x_j\} \setminus \{x_i^\sigma\}$ to unity. We have hence found that we can treat all four-dimensional cones in the very same way provided that the orbifold points of the ambient space $A$ do not meet the invariant hypersurface $X$.

\section{Invariant Polynomials of Examples}\label{app:invpolys}

\subsection{Case \#1, first symmetry action}

The invariant polynomial under \eqref{eq:z5quinticaction1} is a general linear combination of monomials with exponents
\begin{equation}
\begin{aligned}[]
(5, 0, 0, 0, 0),
 (3, 1, 0, 0, 1),
 (3, 0, 1, 1, 0),
 (2, 0, 1, 0, 2),
 (2, 1, 2, 0, 0),
 (2, 2, 0, 1, 0),
 (2, 0, 0, 2, 1),\\
 (1, 2, 0, 0, 2),
 (1, 3, 1, 0, 0),
 (1, 0, 3, 0, 1),
 (1, 0, 0, 1, 3),
 (1, 1, 1, 1, 1),
 (1, 0, 2, 2, 0),
 (1, 1, 0, 3, 0),\\
 (0, 0, 0, 0, 5),
 (0, 5, 0, 0, 0),
 (0, 1, 1, 0, 3),
 (0, 2, 2, 0, 1),
 (0, 0, 5, 0, 0),
 (0, 3, 0, 1, 1),
 (0, 0, 2, 1, 2),\\
 (0, 1, 3, 1, 0),
 (0, 1, 0, 2, 2),
 (0, 2, 1, 2, 0),
 (0, 0, 1, 3, 1),
 (0, 0, 0, 5, 0),
\end{aligned}
\end{equation}
where we have used the same labelling of coordinates as in \eqref{eq:raysquintic}. 

\subsection{Case \#1, second symmetry action}

The invariant polynomial under \eqref{eq:z5quinticaction2} is found from the above by finding the orbits of the homogeneous coordinates under the permutations $R(\gamma_2)$. The monomials in each orbit must then appear with an identical coefficient in the defining polynomial of an invariant hypersurface. The orbits under the choice $(2)$ in \eqref{eq:z5quinticaction2} are 
\begin{equation}
\begin{aligned}[]
 [(0, 5, 0, 0, 0), (5, 0, 0, 0, 0), (0, 0, 0, 5, 0), (0, 0, 5, 0, 0), (0,
0, 0, 0, 5)]
\,,\\ 
[(0, 0, 1, 3, 1), (1, 0, 0, 1, 3), (3, 1, 0, 0, 1), (0, 1, 3, 1, 0), (1,
3, 1, 0, 0)]
\,,\\ 
[(1, 0, 3, 0, 1), (0, 3, 0, 1, 1), (3, 0, 1, 1, 0), (1, 1, 0, 3, 0), (0,
1, 1, 0, 3)]
\,,\\ 
[(1, 0, 2, 2, 0), (2, 2, 0, 1, 0), (0, 1, 0, 2, 2), (0, 2, 2, 0, 1), (2,
0, 1, 0, 2)]
\,,\\ 
[(2, 1, 2, 0, 0), (0, 2, 1, 2, 0), (1, 2, 0, 0, 2), (2, 0, 0, 2, 1), (0,
0, 2, 1, 2)]
\,,\\ 
[(1, 1, 1, 1, 1)]\, .
\end{aligned}
\end{equation}
The other groups actions in \eqref{eq:z5quinticaction2} can be found from this by a permutation of variables. 

\subsection{Case \#2, first symmetry action}
The exponents of the invariant monomials under \eqref{eq:actionbicubic1} are given by
\begin{equation}
\begin{aligned}[]
(3, 0, 0, 0, 0, 3),
 (3, 0, 0, 3, 0, 0),
 (3, 0, 0, 1, 1, 1),
 (3, 0, 0, 0, 3, 0),
 (2, 0, 1, 1, 0, 2),
 (2, 1, 0, 2, 0, 1),\\
 (2, 1, 0, 0, 1, 2),
 (2, 0, 1, 2, 1, 0),
 (2, 0, 1, 0, 2, 1),
 (2, 1, 0, 1, 2, 0),
 (1, 1, 1, 0, 0, 3),
 (1, 2, 0, 1, 0, 2),\\
 (1, 0, 2, 2, 0, 1),
 (1, 1, 1, 3, 0, 0),
 (1, 0, 2, 0, 1, 2),
 (1, 1, 1, 1, 1, 1),
 (1, 2, 0, 2, 1, 0),
 (1, 2, 0, 0, 2, 1),\\
 (1, 0, 2, 1, 2, 0),
 (1, 1, 1, 0, 3, 0),
 (0, 0, 3, 0, 0, 3),
 (0, 3, 0, 0, 0, 3),
 (0, 1, 2, 1, 0, 2),
 (0, 2, 1, 2, 0, 1),\\
 (0, 0, 3, 3, 0, 0),
 (0, 3, 0, 3, 0, 0),
 (0, 2, 1, 0, 1, 2),
 (0, 0, 3, 1, 1, 1),
 (0, 3, 0, 1, 1, 1),
 (0, 1, 2, 2, 1, 0),\\
 (0, 1, 2, 0, 2, 1),
 (0, 2, 1, 1, 2, 0),
 (0, 0, 3, 0, 3, 0),
 (0, 3, 0, 0, 3, 0).
\end{aligned}
\end{equation}

\subsection{Case \#2, second symmetry action}

As in case $\#1$, the invariant polynomial under \eqref{eq:actionbicubic2} is found from the above by finding the orbits of the homogeneous coordinates under the permutations $R(\gamma_2)$. The monomials in each orbit must then appear with an identical coefficient in the defining polynomial of an invariant hypersurface. The orbits under the choice $(4)$ in \eqref{eq:actionbicubic2} are 
\begin{equation}
\begin{aligned}[]
[(3, 0, 0, 0, 0, 3), (0, 3, 0, 3, 0, 0), (0, 0, 3, 0, 3, 0)]
\,,\\ 
[(0, 3, 0, 0, 3, 0), (0, 0, 3, 0, 0, 3), (3, 0, 0, 3, 0, 0)]
\,,\\ 
[(3, 0, 0, 1, 1, 1), (0, 3, 0, 1, 1, 1), (0, 0, 3, 1, 1, 1)]
\,,\\ 
[(0, 0, 3, 3, 0, 0), (3, 0, 0, 0, 3, 0), (0, 3, 0, 0, 0, 3)]
\,,\\ 
[(1, 2, 0, 2, 1, 0), (2, 0, 1, 1, 0, 2), (0, 1, 2, 0, 2, 1)]
\,,\\ 
[(2, 1, 0, 2, 0, 1), (0, 2, 1, 1, 2, 0), (1, 0, 2, 0, 1, 2)]
\,,\\ 
[(2, 1, 0, 0, 1, 2), (0, 2, 1, 2, 0, 1), (1, 0, 2, 1, 2, 0)]
\,,\\ 
[(2, 0, 1, 2, 1, 0), (0, 1, 2, 1, 0, 2), (1, 2, 0, 0, 2, 1)]
\,,\\ 
[(0, 1, 2, 2, 1, 0), (1, 2, 0, 1, 0, 2), (2, 0, 1, 0, 2, 1)]
\,,\\ 
[(0, 2, 1, 0, 1, 2), (2, 1, 0, 1, 2, 0), (1, 0, 2, 2, 0, 1)]
\,,\\ 
[(1, 1, 1, 0, 0, 3), (1, 1, 1, 0, 3, 0), (1, 1, 1, 3, 0, 0)]
\,,\\ 
[(1, 1, 1, 1, 1, 1)]\, .
\end{aligned}
\end{equation}

\subsection{Case \#3}
The exponents of the invariant monomials under \eqref{eq:actioncase3} are given by 
\begin{equation}
\begin{aligned}[]
(8, 0, 0, 0, 0, 0, 4),
 (6, 1, 0, 0, 1, 1, 3),
 (6, 0, 0, 2, 0, 0, 4),
 (5, 0, 0, 1, 2, 1, 3),
 (5, 2, 0, 1, 0, 1, 3),\\
 (4, 0, 0, 0, 4, 2, 2),
 (4, 2, 0, 0, 2, 2, 2),
 (4, 4, 0, 0, 0, 2, 2),
 (4, 1, 0, 2, 1, 1, 3),
 (4, 0, 0, 4, 0, 0, 4),\\
 (3, 0, 1, 1, 0, 0, 2),
 (3, 1, 0, 1, 3, 2, 2),
 (3, 3, 0, 1, 1, 2, 2),
 (3, 0, 0, 3, 2, 1, 3),
 (3, 2, 0, 3, 0, 1, 3),\\
 (2, 0, 1, 0, 2, 1, 1),
 (2, 2, 1, 0, 0, 1, 1),
 (2, 1, 0, 0, 5, 3, 1),
 (2, 3, 0, 0, 3, 3, 1),
 (2, 5, 0, 0, 1, 3, 1),\\
 (2, 0, 0, 2, 4, 2, 2),
 (2, 2, 0, 2, 2, 2, 2),
 (2, 4, 0, 2, 0, 2, 2),
 (2, 1, 0, 4, 1, 1, 3),
 (2, 0, 0, 6, 0, 0, 4),\\
 (1, 1, 1, 1, 1, 1, 1),
 (1, 0, 1, 3, 0, 0, 2),
 (1, 0, 0, 1, 6, 3, 1),
 (1, 2, 0, 1, 4, 3, 1),
 (1, 4, 0, 1, 2, 3, 1),\\
 (1, 6, 0, 1, 0, 3, 1),
 (1, 1, 0, 3, 3, 2, 2),
 (1, 3, 0, 3, 1, 2, 2),
 (1, 0, 0, 5, 2, 1, 3),
 (1, 2, 0, 5, 0, 1, 3),\\
 (0, 0, 2, 0, 0, 0, 0),
 (0, 1, 1, 0, 3, 2, 0),
 (0, 3, 1, 0, 1, 2, 0),
 (0, 0, 1, 2, 2, 1, 1),
 (0, 2, 1, 2, 0, 1, 1),\\
 (0, 0, 0, 0, 8, 4, 0),
 (0, 2, 0, 0, 6, 4, 0),
 (0, 4, 0, 0, 4, 4, 0),
 (0, 6, 0, 0, 2, 4, 0),
 (0, 8, 0, 0, 0, 4, 0),\\
 (0, 1, 0, 2, 5, 3, 1),
 (0, 3, 0, 2, 3, 3, 1),
 (0, 5, 0, 2, 1, 3, 1),
 (0, 0, 0, 4, 4, 2, 2),
 (0, 2, 0, 4, 2, 2, 2),\\
 (0, 4, 0, 4, 0, 2, 2),
 (0, 1, 0, 6, 1, 1, 3),
 (0, 0, 0, 8, 0, 0, 4).
\end{aligned}
\end{equation}

\subsection{Case \#4, first symmetry action}
The exponents of the invariant monomials under \eqref{eq:actioncase4_1} are given by 
\begin{equation}\label{eq:case4_1invmonos}
\begin{aligned}[]
(8, 0, 0, 4, 0, 0, 4),
 (8, 0, 2, 2, 0, 0, 4),
 (8, 0, 4, 0, 0, 0, 4),
 (7, 1, 1, 3, 0, 0, 4),
 (7, 1, 3, 1, 0, 0, 4),\\
 (6, 2, 0, 4, 0, 0, 4),
 (6, 0, 0, 4, 1, 0, 3),
 (6, 2, 2, 2, 0, 0, 4),
 (6, 0, 2, 2, 1, 0, 3),
 (6, 2, 4, 0, 0, 0, 4),\\
 (6, 0, 4, 0, 1, 0, 3),
 (5, 3, 1, 3, 0, 0, 4),
 (5, 1, 1, 3, 1, 0, 3),
 (5, 3, 3, 1, 0, 0, 4),
 (5, 1, 3, 1, 1, 0, 3),\\
 (4, 0, 1, 1, 0, 1, 2),
 (4, 4, 0, 4, 0, 0, 4),
 (4, 2, 0, 4, 1, 0, 3),
 (4, 0, 0, 4, 2, 0, 2),
 (4, 4, 2, 2, 0, 0, 4),\\
 (4, 2, 2, 2, 1, 0, 3),
 (4, 0, 2, 2, 2, 0, 2),
 (4, 4, 4, 0, 0, 0, 4),
 (4, 2, 4, 0, 1, 0, 3),
 (4, 0, 4, 0, 2, 0, 2),\\
 (3, 1, 0, 2, 0, 1, 2),
 (3, 1, 2, 0, 0, 1, 2),
 (3, 5, 1, 3, 0, 0, 4),
 (3, 3, 1, 3, 1, 0, 3),
 (3, 1, 1, 3, 2, 0, 2),\\
 (3, 5, 3, 1, 0, 0, 4),
 (3, 3, 3, 1, 1, 0, 3),
 (3, 1, 3, 1, 2, 0, 2),
 (2, 2, 1, 1, 0, 1, 2),
 (2, 0, 1, 1, 1, 1, 1),\\
 (2, 6, 0, 4, 0, 0, 4),
 (2, 4, 0, 4, 1, 0, 3),
 (2, 2, 0, 4, 2, 0, 2),
 (2, 6, 2, 2, 0, 0, 4),
 (2, 0, 0, 4, 3, 0, 1),\\
 (2, 4, 2, 2, 1, 0, 3),
 (2, 2, 2, 2, 2, 0, 2),
 (2, 6, 4, 0, 0, 0, 4),
 (2, 0, 2, 2, 3, 0, 1),
 (2, 4, 4, 0, 1, 0, 3),\\
 (2, 2, 4, 0, 2, 0, 2),
 (2, 0, 4, 0, 3, 0, 1),
 (1, 3, 0, 2, 0, 1, 2),
 (1, 1, 0, 2, 1, 1, 1),
 (1, 3, 2, 0, 0, 1, 2),\\
 (1, 1, 2, 0, 1, 1, 1),
 (1, 7, 1, 3, 0, 0, 4),
 (1, 5, 1, 3, 1, 0, 3),
 (1, 3, 1, 3, 2, 0, 2),
 (1, 7, 3, 1, 0, 0, 4),\\
 (1, 1, 1, 3, 3, 0, 1),
 (1, 5, 3, 1, 1, 0, 3),
 (1, 3, 3, 1, 2, 0, 2),
 (1, 1, 3, 1, 3, 0, 1),
 (0, 0, 0, 0, 0, 2, 0),\\
 (0, 4, 1, 1, 0, 1, 2),
 (0, 2, 1, 1, 1, 1, 1),
 (0, 0, 1, 1, 2, 1, 0),
 (0, 8, 0, 4, 0, 0, 4),
 (0, 6, 0, 4, 1, 0, 3),\\
 (0, 4, 0, 4, 2, 0, 2),
 (0, 8, 2, 2, 0, 0, 4),
 (0, 2, 0, 4, 3, 0, 1),
 (0, 6, 2, 2, 1, 0, 3),
 (0, 0, 0, 4, 4, 0, 0),\\
 (0, 4, 2, 2, 2, 0, 2),
 (0, 8, 4, 0, 0, 0, 4),
 (0, 2, 2, 2, 3, 0, 1),
 (0, 6, 4, 0, 1, 0, 3),
 (0, 0, 2, 2, 4, 0, 0),\\
 (0, 4, 4, 0, 2, 0, 2),
 (0, 2, 4, 0, 3, 0, 1),
 (0, 0, 4, 0, 4, 0, 0)  
\end{aligned}
\end{equation}

\subsection{Case \#4, second symmetry action}

The invariant polynomial under \eqref{eq:actioncase4_2} is found from \eqref{eq:case4_1invmonos} by finding the orbits of the homogeneous coordinates under the permutations $R(\gamma_2)$. The monomials in each orbit must then appear with an identical coefficient in the defining polynomial of an invariant hypersurface. The orbits under \eqref{eq:actioncase4_2} are 
\begin{equation}
 \begin{aligned}[]
[(0, 8, 4, 0, 0, 0, 4), (8, 0, 0, 4, 0, 0, 4)],
 [(0, 8, 2, 2, 0, 0, 4), (8, 0, 2, 2, 0, 0, 4)],\\
 [(0, 8, 0, 4, 0, 0, 4), (8, 0, 4, 0, 0, 0, 4)],
 [(1, 7, 3, 1, 0, 0, 4), (7, 1, 1, 3, 0, 0, 4)],\\
 [(1, 7, 1, 3, 0, 0, 4), (7, 1, 3, 1, 0, 0, 4)],
 [(2, 6, 4, 0, 0, 0, 4), (6, 2, 0, 4, 0, 0, 4)],\\
 [(0, 6, 4, 0, 1, 0, 3), -(6, 0, 0, 4, 1, 0, 3)],
 [(2, 6, 2, 2, 0, 0, 4), (6, 2, 2, 2, 0, 0, 4)],\\
 [(0, 6, 2, 2, 1, 0, 3), -(6, 0, 2, 2, 1, 0, 3)],
 [(2, 6, 0, 4, 0, 0, 4), (6, 2, 4, 0, 0, 0, 4)],\\
 [(0, 6, 0, 4, 1, 0, 3), -(6, 0, 4, 0, 1, 0, 3)],
 [(3, 5, 3, 1, 0, 0, 4), (5, 3, 1, 3, 0, 0, 4)],\\
 [(1, 5, 3, 1, 1, 0, 3), -(5, 1, 1, 3, 1, 0, 3)],
 [(3, 5, 1, 3, 0, 0, 4), (5, 3, 3, 1, 0, 0, 4)],\\
 [(1, 5, 1, 3, 1, 0, 3), -(5, 1, 3, 1, 1, 0, 3)],
 [(0, 4, 1, 1, 0, 1, 2), -(4, 0, 1, 1, 0, 1, 2)],\\
 [(4, 4, 0, 4, 0, 0, 4), (4, 4, 4, 0, 0, 0, 4)],
 [(2, 4, 4, 0, 1, 0, 3), -(4, 2, 0, 4, 1, 0, 3)],\\
 [(0, 4, 4, 0, 2, 0, 2), (4, 0, 0, 4, 2, 0, 2)],
 [(4, 4, 2, 2, 0, 0, 4)],\\
 [(2, 4, 2, 2, 1, 0, 3), -(4, 2, 2, 2, 1, 0, 3)],
 [(0, 4, 2, 2, 2, 0, 2), (4, 0, 2, 2, 2, 0, 2)],\\
 [(2, 4, 0, 4, 1, 0, 3), -(4, 2, 4, 0, 1, 0, 3)],
 [(0, 4, 0, 4, 2, 0, 2), (4, 0, 4, 0, 2, 0, 2)],\\
 [(1, 3, 2, 0, 0, 1, 2), -(3, 1, 0, 2, 0, 1, 2)],
 [(1, 3, 0, 2, 0, 1, 2), -(3, 1, 2, 0, 0, 1, 2)],\\
 [(3, 3, 1, 3, 1, 0, 3), -(3, 3, 3, 1, 1, 0, 3)],
 [(1, 3, 3, 1, 2, 0, 2), (3, 1, 1, 3, 2, 0, 2)],\\
 [(1, 3, 1, 3, 2, 0, 2), (3, 1, 3, 1, 2, 0, 2)],
 [(0, 2, 1, 1, 1, 1, 1), (2, 0, 1, 1, 1, 1, 1)],\\
 [(2, 2, 0, 4, 2, 0, 2), (2, 2, 4, 0, 2, 0, 2)],
 [(0, 2, 4, 0, 3, 0, 1), -(2, 0, 0, 4, 3, 0, 1)],\\
 [(2, 2, 2, 2, 2, 0, 2)],
 [(0, 2, 2, 2, 3, 0, 1), -(2, 0, 2, 2, 3, 0, 1)],\\
 [(0, 2, 0, 4, 3, 0, 1), -(2, 0, 4, 0, 3, 0, 1)],
 [(1, 1, 0, 2, 1, 1, 1), (1, 1, 2, 0, 1, 1, 1)],\\
 [(1, 1, 1, 3, 3, 0, 1), -(1, 1, 3, 1, 3, 0, 1)],
 [(0, 0, 0, 0, 0, 2, 0)],\\
 [(0, 0, 0, 4, 4, 0, 0), (0, 0, 4, 0, 4, 0, 0)],
 [(0, 0, 2, 2, 4, 0, 0)]
 \end{aligned}
\end{equation}
A sign in front of a vector of exponents indicates that the two monomials must have the same coefficient with a relative sign in the defining equation. 

\subsection{Case \#5}

The exponents of the invariant monomials under \eqref{eq:actioncase5} are given by 
\begin{equation}
\begin{aligned}[]
 (8, 0, 0, 4, 0, 0, 4),
 (8, 0, 2, 2, 0, 0, 4),
 (8, 0, 4, 0, 0, 0, 4),
 (7, 1, 1, 3, 0, 0, 4),
 (7, 1, 3, 1, 0, 0, 4),\\
 (6, 2, 0, 4, 0, 0, 4),
 (6, 0, 0, 4, 1, 0, 3),
 (6, 2, 2, 2, 0, 0, 4),
 (6, 0, 2, 2, 1, 0, 3),
 (6, 2, 4, 0, 0, 0, 4),\\
 (6, 0, 4, 0, 1, 0, 3),
 (5, 3, 1, 3, 0, 0, 4),
 (5, 1, 1, 3, 1, 0, 3),
 (5, 3, 3, 1, 0, 0, 4),
 (5, 1, 3, 1, 1, 0, 3),\\
 (4, 0, 0, 2, 0, 1, 2),
 (4, 0, 2, 0, 0, 1, 2),
 (4, 4, 0, 4, 0, 0, 4),
 (4, 2, 0, 4, 1, 0, 3),
 (4, 0, 0, 4, 2, 0, 2),\\
 (4, 4, 2, 2, 0, 0, 4),
 (4, 2, 2, 2, 1, 0, 3),
 (4, 0, 2, 2, 2, 0, 2),
 (4, 4, 4, 0, 0, 0, 4),
 (4, 2, 4, 0, 1, 0, 3),\\
 (4, 0, 4, 0, 2, 0, 2),
 (3, 1, 1, 1, 0, 1, 2),
 (3, 5, 1, 3, 0, 0, 4),
 (3, 3, 1, 3, 1, 0, 3),
 (3, 1, 1, 3, 2, 0, 2),\\
 (3, 5, 3, 1, 0, 0, 4),
 (3, 3, 3, 1, 1, 0, 3),
 (3, 1, 3, 1, 2, 0, 2),
 (2, 2, 0, 2, 0, 1, 2),
 (2, 0, 0, 2, 1, 1, 1),\\
 (2, 2, 2, 0, 0, 1, 2),
 (2, 0, 2, 0, 1, 1, 1),
 (2, 6, 0, 4, 0, 0, 4),
 (2, 4, 0, 4, 1, 0, 3),
 (2, 2, 0, 4, 2, 0, 2),\\
 (2, 6, 2, 2, 0, 0, 4),
 (2, 0, 0, 4, 3, 0, 1),
 (2, 4, 2, 2, 1, 0, 3),
 (2, 2, 2, 2, 2, 0, 2),
 (2, 6, 4, 0, 0, 0, 4),\\
 (2, 0, 2, 2, 3, 0, 1),
 (2, 4, 4, 0, 1, 0, 3),
 (2, 2, 4, 0, 2, 0, 2),
 (2, 0, 4, 0, 3, 0, 1),
 (1, 3, 1, 1, 0, 1, 2),\\
 (1, 1, 1, 1, 1, 1, 1),
 (1, 7, 1, 3, 0, 0, 4),
 (1, 5, 1, 3, 1, 0, 3),
 (1, 3, 1, 3, 2, 0, 2),
 (1, 7, 3, 1, 0, 0, 4),\\
 (1, 1, 1, 3, 3, 0, 1),
 (1, 5, 3, 1, 1, 0, 3),
 (1, 3, 3, 1, 2, 0, 2),
 (1, 1, 3, 1, 3, 0, 1),
 (0, 0, 0, 0, 0, 2, 0),\\
 (0, 4, 0, 2, 0, 1, 2),
 (0, 2, 0, 2, 1, 1, 1),
 (0, 0, 0, 2, 2, 1, 0),
 (0, 4, 2, 0, 0, 1, 2),
 (0, 2, 2, 0, 1, 1, 1),\\
 (0, 0, 2, 0, 2, 1, 0),
 (0, 8, 0, 4, 0, 0, 4),
 (0, 6, 0, 4, 1, 0, 3),
 (0, 4, 0, 4, 2, 0, 2),
 (0, 8, 2, 2, 0, 0, 4),\\
 (0, 2, 0, 4, 3, 0, 1),
 (0, 6, 2, 2, 1, 0, 3),
 (0, 0, 0, 4, 4, 0, 0),
 (0, 4, 2, 2, 2, 0, 2),
 (0, 8, 4, 0, 0, 0, 4),\\
 (0, 2, 2, 2, 3, 0, 1),
 (0, 6, 4, 0, 1, 0, 3),
 (0, 0, 2, 2, 4, 0, 0),
 (0, 4, 4, 0, 2, 0, 2),
 (0, 2, 4, 0, 3, 0, 1),\\
 (0, 0, 4, 0, 4, 0, 0).
 \end{aligned}
\end{equation}

\subsection{The tetra-quadric}\label{app:polytq}

The generic invariant hypersurface is given by a linear combination of polynomials with the following exponent vectors. Here, a sign indicates that the two monomials must be combined with a relative sign. 
\begin{equation}
\begin{aligned}[]
[(2, 1, 0, 0, 2, 2, 1, 0), (0, 0, 2, 1, 1, 0, 2, 2)],
[(2, 2, 0, 0, 2, 2, 0, 0), (0, 0, 2, 2, 0, 0, 2, 2)],\\
[(2, 1, 1, 0, 2, 1, 1, 0), (1, 0, 2, 1, 1, 0, 2, 1)],
[(2, 2, 1, 0, 2, 1, 0, 0), (1, 0, 2, 2, 0, 0, 2, 1)],\\
[(2, 0, 0, 1, 1, 2, 2, 0), -(0, 1, 2, 0, 2, 0, 1, 2)],
[(2, 0, 1, 1, 1, 1, 2, 0), -(1, 1, 2, 0, 2, 0, 1, 1)],\\
[(2, 2, 2, 1, 1, 0, 0, 0), -(2, 1, 2, 2, 0, 0, 1, 0)],
[(2, 0, 0, 2, 0, 2, 2, 0), -(0, 2, 2, 0, 2, 0, 0, 2)],\\
[(2, 0, 1, 2, 0, 1, 2, 0), -(1, 2, 2, 0, 2, 0, 0, 1)],
[(1, 0, 0, 0, 2, 2, 2, 1), -(0, 0, 1, 0, 2, 1, 2, 2)],\\
[(1, 1, 0, 1, 1, 2, 1, 1), (0, 1, 1, 1, 1, 1, 1, 2)],
[(1, 2, 0, 1, 1, 2, 0, 1), (0, 1, 1, 2, 0, 1, 1, 2)],\\
[(1, 2, 1, 1, 1, 1, 0, 1), (1, 1, 1, 2, 0, 1, 1, 1)],
[(1, 1, 0, 2, 0, 2, 1, 1), (0, 2, 1, 1, 1, 1, 0, 2)],\\
[(1, 2, 0, 2, 0, 2, 0, 1), (0, 2, 1, 2, 0, 1, 0, 2)],
[(0, 2, 0, 1, 1, 2, 0, 2), (0, 1, 0, 2, 0, 2, 1, 2)],\\
[(0, 2, 0, 2, 0, 2, 0, 2)] ,[(0, 1, 0, 1, 1, 2, 1, 2)], 
[(1, 1, 1, 1, 1, 1, 1, 1)],[(2, 0, 2, 0, 2, 0, 2, 0)], \\
[(1, 2, 1, 2, 0, 1, 0, 1)].
\end{aligned}
\end{equation}

\subsection{The twisted tetra-quadric}\label{app:monosttq}

The invariant polynomial of the model discussed in Section \ref{sect:twtetraquadric} is generated by the following linear combinations of monomials with equal coefficients 
\begin{equation}
\begin{aligned}[]
[(0, 0, 0, 1, 2, 0, 1, 2, 2, 2), (2, 2, 2, 1, 0, 2, 1, 0, 0, 0)],\\
 [(0, 0, 0, 2, 1, 1, 0, 2, 2, 2), (2, 2, 2, 0, 1, 1, 2, 0, 0, 0)],\\
 [(1, 0, 0, 2, 2, 0, 0, 2, 2, 1), (1, 2, 2, 0, 0, 2, 2, 0, 0, 1)],\\
 [(0, 0, 1, 0, 2, 0, 2, 1, 2, 2), (2, 2, 1, 2, 0, 2, 0, 1, 0, 0)],\\
 [(0, 0, 1, 1, 1, 1, 1, 1, 2, 2), (2, 2, 1, 1, 1, 1, 1, 1, 0, 0)],\\
 [(1, 0, 1, 1, 2, 0, 1, 1, 2, 1), (1, 2, 1, 1, 0, 2, 1, 1, 0, 1)],\\
 [(0, 0, 1, 2, 0, 2, 0, 1, 2, 2), (2, 2, 1, 0, 2, 0, 2, 1, 0, 0)],\\
 [(1, 0, 1, 2, 1, 1, 0, 1, 2, 1), (1, 2, 1, 0, 1, 1, 2, 1, 0, 1)],\\
 [(0, 2, 1, 0, 0, 2, 2, 1, 0, 2), (2, 0, 1, 2, 2, 0, 0, 1, 2, 0)],\\
 [(0, 0, 2, 0, 1, 1, 2, 0, 2, 2), (2, 2, 0, 2, 1, 1, 0, 2, 0, 0)],\\
 [(1, 0, 2, 0, 2, 0, 2, 0, 2, 1), (1, 2, 0, 2, 0, 2, 0, 2, 0, 1)],\\
 [(0, 0, 2, 1, 0, 2, 1, 0, 2, 2), (2, 2, 0, 1, 2, 0, 1, 2, 0, 0)],\\
 [(1, 0, 2, 1, 1, 1, 1, 0, 2, 1), (1, 2, 0, 1, 1, 1, 1, 2, 0, 1)],\\
 [(0, 2, 0, 1, 0, 2, 1, 2, 0, 2), (2, 0, 2, 1, 2, 0, 1, 0, 2, 0)],\\
 [(1, 0, 2, 2, 0, 2, 0, 0, 2, 1), (1, 2, 0, 0, 2, 0, 2, 2, 0, 1)],\\
 [(0, 2, 0, 0, 1, 1, 2, 2, 0, 2), (2, 0, 2, 2, 1, 1, 0, 0, 2, 0)],\\
 [(0, 1, 0, 0, 2, 0, 2, 2, 1, 2), (2, 1, 2, 2, 0, 2, 0, 0, 1, 0)],\\
 [(0, 1, 0, 1, 1, 1, 1, 2, 1, 2), (2, 1, 2, 1, 1, 1, 1, 0, 1, 0)],\\
 [(1, 1, 0, 1, 2, 0, 1, 2, 1, 1), (1, 1, 2, 1, 0, 2, 1, 0, 1, 1)],\\
 [(0, 1, 0, 2, 0, 2, 0, 2, 1, 2), (2, 1, 2, 0, 2, 0, 2, 0, 1, 0)],\\
 [(1, 1, 0, 2, 1, 1, 0, 2, 1, 1), (1, 1, 2, 0, 1, 1, 2, 0, 1, 1)],\\
 [(0, 1, 2, 0, 0, 2, 2, 0, 1, 2), (2, 1, 0, 2, 2, 0, 0, 2, 1, 0)],\\
 [(0, 1, 1, 0, 1, 1, 2, 1, 1, 2), (2, 1, 1, 2, 1, 1, 0, 1, 1, 0)],\\
 [(1, 1, 1, 0, 2, 0, 2, 1, 1, 1), (1, 1, 1, 2, 0, 2, 0, 1, 1, 1)],\\
 [(0, 1, 1, 1, 0, 2, 1, 1, 1, 2), (2, 1, 1, 1, 2, 0, 1, 1, 1, 0)],\\
 [(1, 1, 1, 1, 1, 1, 1, 1, 1, 1)].
\end{aligned}
\end{equation}
As before, we have given the vector of exponents of the monomials.

\end{document}